\documentclass[format=acmsmall]{acmart}
\renewcommand\footnotetextcopyrightpermission[1]{} 
\makeatletter
\renewcommand\@formatdoi[1]{\ignorespaces}
\makeatother

\usepackage{colortbl}
\usepackage{amsmath}
\usepackage{amssymb}
\usepackage{proof}
\usepackage{multicol}%
\usepackage{multirow}
\usepackage{listings}%
\usepackage{url}
\usepackage{mathpartir}%
\usepackage{stmaryrd}
\usepackage{fixltx2e}
\usepackage{float}
\usepackage[finalnew]{trackchanges}
\usepackage{wrapfig}
\usepackage{algorithmic}
\usepackage{xspace}
\usepackage{thmtools}
\usepackage{thm-restate}
\usepackage{hhline}
\usepackage{scalerel}

\acmJournal{PACMPL}
\acmVolume{}
\acmNumber{} 
\acmArticle{1}
\acmYear{}
\acmMonth{1}
\acmDOI{} 
\startPage{1}

\setcopyright{none}

\addeditor{Mike}
\addeditor{Jeremy}
\addeditor{Avik}

\title{Optimizing and Evaluating Transient Gradual Typing}
\date{\today}

\author{Michael M. Vitousek}
\affiliation{
  \institution{Indiana University}            
}
\email{mvitouse@indiana.edu}
\author{Jeremy G. Siek}
\affiliation{
  \institution{Indiana University}            
}
\email{jsiek@indiana.edu}
\author{Avik Chaudhuri}
\affiliation{
  \institution{Facebook Inc.}            
}
\email{avik@fb.com}

\newcommand{\benchcount}{ten\xspace}

\newcommand{\langs}{\lambda^\dyn_s}
\newcommand{\langd}{\lambda^\Downarrow_d}
\newcommand{\lange}{\lambda^\to_e}

\newcommand{\intt}{\textsf{int}}
\newcommand{\boolt}{\textsf{bool}}
\newcommand{\strt}{\textsf{str}}
\newcommand{\reft}[1]{\textsf{ref}~#1}

\newcommand{\funt}[2]{#1\to#2}

\newcommand{\dyn}{\star}
\newcommand{\refs}{\textsf{ref}}
\newcommand{\funs}{{\to}}

\newcommand{\ei}{d}
\newcommand{\ein}[1]{d_{#1}}

\newcommand{\fntx}[4]{\lambda (#1{:}#2){\to} #3.~#4}

\newcommand{\fnx}[2]{\lambda #1.~#2}

\newcommand{\refx}[1]{\texttt{ref}~#1}
\newcommand{\refxi}[2]{\texttt{ref}_{#2}~#1}

\newcommand{\letx}[3]{\texttt{let}~#1=#2~\texttt{in}~#3}
\newcommand{\derefx}[1]{{!}#1}
\newcommand{\derefxp}[2]{{!}#1^{#2}}

\newcommand{\mutx}[2]{#1{:=}#2}

\newcommand{\mutxpINV}[3]{#1{:=}^{#3}#2}

\newcommand{\checkx}[2]{#1{\Downarrow}#2}

\newcommand{\respos}          {\textsc{Res}}
\newcommand{\argpos}          {\textsc{Arg}}

\newcommand{\tyorig}{{\scaleobj{.8}{\lozenge}}}
\newcommand{\unorig}{{\scaleobj{.8}{\blacklozenge}}}

\newcommand{\cfg}[2]{\langle #1, #2\rangle}
\newcommand{\steps}{\longrightarrow}

\newcommand{\matches}{\rhd}
\newcommand{\matchrel}[2] {#1 \matches #2}

\newcommand{\ledyn}{\sqsubseteq}

\newcommand{\stuckop}      {\textsf{stuck}}

\newcommand{\stucknb}[3]{\cfg{#1}{#2}~\stuckop~#3}

\newcommand{\lowlty}[1]{\lfloor #1 \rfloor}
\newcommand{\uplty}[1]{\lceil #1 \rceil}

\newcommand{\vsol}[2]{#1~\text{solvable}~#2}
\newcommand{\flows}[2]{#1 <: #2}
\newcommand{\depcon}[3]{(#1{:}#2) = #3}
\newcommand{\tagcon}[2]{#1 : #2}
\newcommand{\defcon}[2]{#1 \triangleq #2}

\newcommand{\specmatches}[3]{#1\rhd_{#3}#2} 
\newcommand{\R}[2]{#1\approx #2}

\newcommand{\dom}[1]{\textit{dom}(#1)}

\newcommand{\hastypenb}[2]{\textit{hastype}(#1,#2)}

\newcommand{\GS}{\Gamma{;}\Sigma\vdash}

\newcommand{\yO}{;\Omega}

\newcommand{\insgraphw}[1]{\raisebox{-.9\height}{\includegraphics[scale=.5]{{#1}.pdf}}}
\newcommand{\inslegend}[1]{\includegraphics[scale=.3]{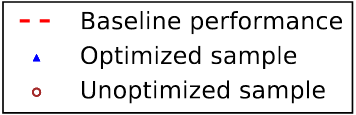}}

\newcommand{\LGE}{\Gamma{;}\emptyset\vdash}
\newcommand{\ES}{\emptyset{;}\Sigma\vdash}

\newtheorem{theorem}{Theorem}
\newtheorem{lemma}{Lemma}
\newtheorem{corollary}{Corollary}

\newtheorem{definition}{Definition}
\newtheorem{innercustomthm}{Theorem}
\newenvironment{customthm}[1]
  {\renewcommand\theinnercustomthm{#1}\innercustomthm}
  {\endinnercustomthm}

\newcommand{\ba}{\begin{array}}
\newcommand{\ea}{\end{array}}

\newcommand{\li}[1]{\lstinline[basicstyle=\small\tt]{#1}}

\newcommand{\xtimes}{{\scaleobj{1}{\times}}}
\def\Put(#1,#2)#3{\leavevmode\makebox(0,0){\put(#1,#2){#3}}}
\newcommand{\chapperfheader}{
\chapter{Optimizing and Evaluating Transient Gradual Typing}\label{chap:perf}
\section{Abstract}
\chapperfabstract}

\definecolor{dkgreen}{rgb}{0,0.6,0}
\definecolor{gray}{rgb}{0.5,0.5,0.5}
\definecolor{lightgray}{rgb}{0.95,0.95,0.95}
\definecolor{mauve}{rgb}{0.58,0,0.82}
\definecolor{tableblue}{rgb}{.8,.8,1}
\definecolor{tablered}{rgb}{1,.3,.3}
\definecolor{dkred}{rgb}{0.4,0,0}
\definecolor{dkyellow}{rgb}{.8, .8, 0}

\newcommand{\inflow}[1]{{?}#1}
\newcommand{\outflow}[1]{{!}#1}

\renewcommand{\chapperfheader}{
  \begin{abstract}
    \chapperfabstract
  \end{abstract}
  \maketitle}

\renewcommand{\insgraphw}[1]{\raisebox{-.9\height}{\includegraphics[scale=.4]{{#1}.pdf}}}
\renewcommand{\inslegend}{\includegraphics[scale=.3]{legend.png}}

\renewcommand{\cite}[1]{\citep{#1}}
\begin{document}
\thispagestyle{empty}

\newcommand{\chapperfabstract}{
  Gradual typing enables programmers to combine static and dynamic
  typing in the same language. However, ensuring a sound interaction
  between the static and dynamic parts can incur significant runtime
  cost.
  In this paper, we perform a detailed performance analysis of the
  \emph{transient gradual typing} approach implemented in Reticulated
  Python, a gradually typed variant of Python. The transient approach
  \cite{vitousek:2017trans} inserts lightweight checks throughout a
  program rather than installing proxies on higher order values. We
  show that, when running Reticulated Python and the transient
  approach on CPython, performance decreases as programs evolve from
  dynamic to static types, up to a 6$\xtimes$ slowdown compared to
  equivalent Python programs.

  To reduce this overhead, we design a static analysis and
  optimization that removes redundant runtime checks.  The
  optimization employs a static type inference algorithm
  that solves traditional subtyping constraints and also a new kind of
  \emph{check constraint}.
  We evaluate the resulting performance and find that for many
  programs, the efficiency of partially typed programs is close to
  their untyped counterparts, removing most of the slowdown
  of transient checks.
  Finally, we measure the efficiency of Reticulated Python programs
  when running on PyPy, a tracing JIT.  We find that combining PyPy
  with our type inference algorithm reduces the overall overhead to
  zero.
    }

\chapperfheader
\section{Introduction}

Gradual typing enables programmers to gradually evolve their programs
from the flexibility of dynamic typing to the security of static
typing \cite{Siek:2006bh, Tobin-Hochstadt:2006fk}. Over the last
decade, gradual typing has been of great interest to both the research
community
\cite{Rastogi:2012fk,Takikawa:2012ly,Swamy:2014aa,siek:2015mono,Ren:2013vn,allende:2013gradualtalk,ahmed:2011bfa},
and to industry, which has introduced several languages with elements
of gradual typing, such as TypeScript \cite{typescript}, Flow
\cite{flow}, and Dart \cite{dart}. Many existing
gradually typed languages operate by translating a surface language
program into a dynamically typed language by erasing types. We refer
to the latter as the \emph{target language}. This type erasure
approach is safe in the sense that programs do not elicit undefined
behavior as long as the target language is itself safe (for example,
because the target language performs runtime checking in every
primitive operation). However, we would like gradually typed languages
to also be \emph{sound with respect to type annotations}, that is, a
variable annotated with a static type should only be inhabited by
values of the appropriate type \cite{vitousek:2017trans}.  To achieve
this kind of soundness, runtime checks are required on the boundaries
between statically and dynamically typed code.

\subsection{Strategies for Runtime Checks}

Several strategies have been used to implement runtime type checking
for gradually typed languages, and these strategies are appropriate
for different target languages and design goals. The traditional
approach in the research literature is to insert casts during
translation at the site of every implicit conversion between typed
and dynamic code
\cite{Siek:2006bh,Tobin-Hochstadt:2006fk,allende:2013gradualtalk,Swamy:2014aa}. At
runtime, these casts ensure that values correspond to their expected
static type. A cast on a value with first-order type is a
constant-time operation that either succeeds or fails, but for
higher-order types such as functions and references a cast installs a
\emph{proxy} on the casted value.  The proxy ensures that, in the
future, the value behaves according to the target type of the
cast. This approach is called the \emph{guarded}
strategy~\cite{vitousek:2014retic}.

In recent papers, \citet{vitousek:2014retic, vitousek:2017trans}
identified a number of challenges that make the guarded strategy
inappropriate for certain domains. In particular, unless the target
language has powerful support for proxies (such as Racket's chaperones
\cite{Strickland:2012fk}), interaction between proxied values and
foreign functions or target-language code may fail in unexpected ways
\cite{allende:2013strats, cutsem:2013proxies}. Such is the case in
Python and these problems were borne out in Reticulated Python, an
experimental implementation of gradual typing for Python 3
\cite{vitousek:2014retic}.

As an alternative, \citet{vitousek:2014retic} introduced the
\emph{transient} strategy.
In this approach, the translation inserts constant-time checks
throughout the gradually typed program. These checks do not create
proxies but only inspect the type tag \cite{appel:2007cc} of the
value, failing if it does not correspond to the type constructor the
value is statically expected to have (such as $\intt$ or $\funs$, but
not $\intt\to\intt$).

\begin{figure*}
  \centering
  \begin{minipage}{1.0\linewidth}
    \centering
    \begin{minipage}{.48\linewidth}
\begin{lstlisting}[language=Python,basicstyle=\small\tt,morekeywords={Callable,Any,int,bool}]
# Gradual surface program

def idDyn(a:dyn)->dyn: return a

def makeEq(n:int)->int->bool:
  def internal(m:int)->bool:
    return n == m
  return internal



eqFive = makeEq(5) /*\label{cp:ex:surfi:e5}*/
eqFive(20)
eqFive(idDyn('hello world')) /*\label{cp:ex:surfi:id-dyn-string}*/
\end{lstlisting}
\centering (a)
    \end{minipage}
    \begin{minipage}{.5\linewidth}
\begin{lstlisting}[language=Python,basicstyle=\small\tt,morekeywords={Callable,Any,int}]
# Program with transient checks

def idDyn(a): return a

def makeEq(n):
  n@int/*\label{cp:ex:checks:out}*/
  def internal(m):
    m@int/*\label{cp:ex:checks:in}*/
    return n == m
  return internal

eqFive = makeEq(5)@->/*\label{cp:ex:checks:outcall}*/
eqFive(20)@bool/*\label{cp:ex:checks:incall1}*/
eqFive(idDyn('hello world'))@bool/*\label{cp:ex:checks:incall2}*/
\end{lstlisting}
\centering (b)
    \end{minipage}\\
\end{minipage}
\caption{Translation to target language using the transient gradual typing approach.}
\label{cp:fig:trans}
\end{figure*}

To make up for the ``transient'' nature of these checks, they are
inserted pervasively into the program rather than just at the sites of
implicit conversions. The translation inserts checks into the program
at every function call site, at the entry of every function body (to
check that each argument corresponds to its parameter's type), and at
the site of every dereference.  

For example, consider the program in Figure \ref{cp:fig:trans}a,
written in a gradually typed language.  Here, the \li{makeEq} function
is a curried equality function on integers with type
$\intt\to\intt\to\boolt$. It is called to produce the \li{eqFive}
function on line \ref{cp:ex:surfi:e5}, which is then called at line
\ref{cp:ex:surfi:id-dyn-string} on the result of calling the
\li{idDyn} function on a string. Because the result of \li{idDyn} has static
type \li{dyn} (representing the dynamic type), this program should
pass static typechecking, but a runtime check is needed to detect that
at runtime the value being passed into \li{eqFive} is actually a
string and raise an error.  

No matter what strategy is used, in a sound gradually typed language
the result should be a runtime error. The transient strategy achieves
this goal by inserting checks (such as \li{n@int}, which checks that
\li{n} is an integer) as shown in Figure \ref{cp:fig:trans}b.  Type
annotations have been erased, and the bodies of \li{makeEq} and its
internal function now contain checks at lines \ref{cp:ex:checks:out}
and \ref{cp:ex:checks:in} to ensure that, whatever arguments they are
passed are definitely $\intt$s. Similarly, the call to
\li{makeEq} at line \ref{cp:ex:checks:outcall} also contains a check,
to ensure that the result of the function call is a function. Note
that this check does not ensure that the call returns a value of type
$\intt\to\boolt$---such a check cannot be performed by immediate
inspection of the runtime value, but it can verify that the result is
a function. It is then up to the function itself to check that it is
only passed $\intt$s (as it does with the argument checks discussed
above), and additional checks are inserted when \li{eqFive} is called
on lines \ref{cp:ex:checks:incall1} and \ref{cp:ex:checks:incall2} to
ensure that the result of that call is a $\boolt$.

When this program executes, an error will be raised by the check at
line \ref{cp:ex:checks:in}, because the call to \li{eqFive} at line
\ref{cp:ex:checks:incall2} passed in a string. This result is expected and
correct; if the error had not arisen, there would be an uncaught type
error in the body of \li{makeEq}'s inner function, as a
string would inhabit the $\intt$-typed variable \li{m}. This error could then
(depending on the semantics of equality testing) lead to a confusing,
difficult to debug error. As it is, the programmer is simply informed
that a type mismatch occurred and where.

\citet{vitousek:2017trans} showed that this approach supports the
\emph{open-world soundness} property, which states that programs
written in a gradually typed language, translated into a dynamic target
language, and then embedded in arbitrary code native to that dynamic
language, will only ``go wrong'' due to errors in the native code. The
translated, gradually typed program will not be the source of any
errors (other than errors caught by transient checks) even in the
presence of unmoderated interaction with the ``open world.''

\subsection{Performance of Transient Gradual Typing}

Performance is also of critical concern for gradually typed
languages. The runtime checks required for sound gradual typing
inevitably impose some degree of runtime overhead, but ideally this
overhead would be minimized or made up for by type-based compiler
optimizations. Since gradual typing is designed to allow programmers
to gradually vary their programs between static and dynamic
\cite{siek:2015gg}, it is also important that adding or removing
individual annotations does not dramatically degrade the program's
performance.  \citet{takikawa:2016gtdead} examine the performance of
Typed Racket with this criterion in mind by studying programs through
the lens of a \emph{typing lattice} made up of differently-typed
\emph{configurations} of the same program. The top of the lattice is a
fully typed configuration of the program and the bottom is
unannotated, and incrementally adding types moves up the lattice.

Takikawa et al.\ show that in Typed Racket, certain configurations
result in catastrophic slowdown compared to either the top or bottom
configurations. This indicates that the guarded semantics incurs a
substantial cost when interaction between static and dynamic code is
frequent. Many of their benchmarks show mean overheads of over
30$\xtimes$ and worst cases of over 100$\xtimes$, which ``projects an
extremely negative image of \emph{sound} gradual typing''
\cite{takikawa:2016gtdead}.

In this work, we aim to establish whether the transient strategy faces
the same problem. \citet{vitousek:2017trans} performed an initial
performance evaluation of Reticulated Python benchmarks and found that
overheads (compared to an untyped, standard Python version of the same
program) ranged from negligible to over a 5$\xtimes$
slowdown. However, this analysis was limited to examining a single
configuration, the configuration closest to being fully typed. As
shown by \citet{takikawa:2016gtdead}, this is insufficient to make a
strong claim about the overall performance of Reticulated Python.

To obtain a clearer picture of Reticulated's performance, we analyze
the performance of \benchcount benchmarks across their typing
lattices. Since Reticulated Python uses \emph{fine-grained} gradual
typing (where the choice to use static types exists on the level of
individual identifiers) rather than \emph{coarse-grained} (where the
choice is per-module) as is Typed Racket \cite{takikawa:2016gtdead},
the size of the typing lattice is too large to generate and test every
possible configuration. Instead, we generate samples from the lattice
by randomly removing type annotations from a fully-typed version of
the benchmark, replacing them with the dynamic type and taking care to
ensure that each level in the lattice is equally sampled.  Each sample
is then translated to standard Python 3 using Reticulated Python and
executed with CPython, the reference Python runtime.

With this approach, we found that the cost of transient gradual typing
increases as the number of type annotations grows. As a program
evolves from dynamic to static, its performance linearly degrades,
with the worst performance in the most static configurations. This is
because each static type annotation induces checks to ensure that
values correspond to that type. This is, of course, counter to one
hypothetical benefit of static typing---ideally, static types should
aid performance, or at least not degrade it. On the other hand, the
linear degradation of performance to a worst case 6$\xtimes$ overhead
means that the catastrophic configurations encountered in Typed Racket
never occur and the cost of adding an individual type annotation to a
program is predictable.

\subsection{Reducing the Burden of Pervasive Checks}

The transient approach inserts checks throughout the program, but not
all checks are necessary for the program to be sound because some
checks may be redundant and always succeed. To remove unnecessary checks, we perform
type inference on the program after checks have been inserted. Our
inference algorithm is based on those of \citet{Aiken:1995ly} and
\citet{Rastogi:2012fk} and uses subtyping constraints as well as
new \emph{check constraints}, generated by transient checks. We
prove that our algorithm can soundly remove unnecessary checks in a
transient calculus similar to that of \citet{vitousek:2017trans}

We modified Reticulated Python to support this optimization and
measured its performance, again sampling from the typing lattices at
all levels. With redundant checks removed, the linear increase in
execution times disappears, resulting in the fully-typed
configurations displaying negligible overhead and a 6\% average
overhead over all sampled configurations.

\subsection{Transient Gradual Typing on a Tracing JIT}

While this analysis removes many checks statically, the
nature of transient checks suggests that they could also be
dynamically optimized away by a JIT. Fortunately, there is a tracing
JIT for Python 3, PyPy \cite{Bolz:2009kx}. Reticulated Python compiles
to standard Python 3, so it is suitable to use with PyPy.
We found that gradually typed programs running on PyPy displayed much
less overhead than the same configurations running on CPython---the
average overhead over all configurations was 3\% with PyPy compared
to 2.21$\xtimes$ with CPython, suggesting that PyPy is able to
optimize away most of the overhead of transient checks. Some benchmarks still incurred a linear
increase in time as types were added, but to a lesser degree than with
CPython (with a worst case overhead of 2.61$\xtimes$). By combining
PyPy with our type inference optimization, the average overhead was
reduced to zero.

\subsection{Contributions}

In this work, we measure the performance of transient gradual typing
in Reticulated Python and design techniques to improve it. Our
contributions are:

\begin{itemize}
\item We analyze the performance of Reticulated Python programs across their
  typing lattices, finding an average overhead of 2.21$\xtimes$ and much
  better worst-case performance than Typed Racket (Section
  \ref{cp:sec:nooptperf}).
\item We develop a type inference optimization for reducing the number
  of checks needed by the transient approach, and prove it correct
  (Section \ref{cp:sec:infer}).
\item We implement this optimization in Reticulated Python and show
  that it reduces the average overhead to just 6\% across all typing
  lattices (Section \ref{cp:sec:optperf}).
\item We analyze both the unoptimized and optimized versions of
  Reticulated when running under a tracing JIT, and find that it
  performs very well, especially in combination with our optimization
  (Section \ref{cp:sec:pypyperf}).
\end{itemize}

Section \ref{cp:sec:relwork} discusses related work, and Section
\ref{cp:sec:conclusion} concludes.

\section{Performance of Transient Gradual Typing} \label{cp:sec:nooptperf}

Before investigating approaches to improve the performance of
transient gradual typing, we first establish the performance
characteristics of Reticulated Python across the typing
lattice~\cite{takikawa:2016gtdead}. \citet{vitousek:2017trans}
developed a \emph{blame-tracking} technique for transient gradual
typing, allowing programmers to trace runtime type errors
back to the crossing-points between static and dynamic that led to the
error, but for this analysis we disable blame tracking. Conducting the
evaluation with blame tracking is important future work.

%

\subsection{Experimental Setup}

We selected \benchcount Python 3 programs and translated them to
Reticulated Python by inserting type annotations. These benchmarks are
mostly drawn from the official Python benchmark
suite,\footnote{\tt\url{https://github.com/python/performance}} with several
drawn from the analysis of
\citet{takikawa:2016gtdead} and translated from Racket to Python. In most cases, the
resulting Reticulated programs are fully annotated with static
types. However, even with a fully annotated program, the Reticulated
type checker can assign expressions the type \li{Any} (the Reticulated Python
name for the dynamic type), such as an if-then-else expression where
the branches have different types; we did not attempt to guarantee
that such cases do not arise.

To examine the typing lattice for a benchmark, we first count the
number of type constructors that appear in the program's
annotations. We call this the \emph{type weight} of a program. For
example, the presence of the type \li{List[int]} in a program's
annotations adds 2 to its weight, and \li{Callable[[int], bool]}
(which is the Reticulated representation of the type $\intt\to\boolt$)
adds 3. We then divide the type weight into a maximum of 100
intervals: a program with a type weight of 300 would have intervals
$[0, 3), [3, 6), \ldots, [297, 300)$.  Programs with a total type
weight of less than 100 naturally have fewer than 100 intervals. For
each interval, we randomly erase type annotations and replace them
with \li{Any}, until the program's type weight falls within the
interval.  This process can ``dynamize'' types underneath type
constructors; both \li{Any} and \li{List[Any]} are possible types that
could be generated from an original type annotation
\li{List[int]}. Each partially-dynamized program is a configuration
from the typing lattice at the level corresponding to its type
weight. We generate ten configurations per interval, plus a single
fully-typed configuration consisting of the original program, for a
maximum of 1001 configurations.

Each configuration was executed on an Intel Core i3-4130 CPU with 8GB
of RAM running Ubuntu Server 14.04. Configurations were executed
repeatedly on CPython 3.4.3 and average runtimes were recorded.

\subsection{Results}

%

\begin{figure*}
  \centering
   \begin{tabular}[t]{ll}
     Benchmark: \texttt{pystone}  & Benchmark: \texttt{chaos}  \\
     \small (206 SLoC, 532 configurations) &\small  (184 SLoC, 982 configurations) \\
     \insgraphw{pystone.py_python3} & \insgraphw{bm_chaos.py_python3} \\
     Benchmark: \texttt{snake}   & Benchmark: \texttt{go} \\
     \small  (112 SLoC, 662 configurations) & \small (394 SLoC, 1001 configurations) \\
     \insgraphw{snake.py_python3} & \insgraphw{bm_go.py_python3} \\
     Benchmark: \texttt{meteor\_contest}   & Benchmark: \texttt{suffixtree}  \\
     \small  (106 SLoC, 972 configurations) &\small  (338 SLoC, 1001 configurations) \\
     \insgraphw{bm_meteor_contest.noaliases.py_python3} & \insgraphw{suffixtree.py_python3} \\
     Benchmark: \texttt{float}   & Benchmark: \texttt{nbody}  \\
     \small  (48 SLoC, 162 configurations) & \small  (74 SLoC, 892 configurations) \\
     \insgraphw{bm_float.py_python3} & \insgraphw{bm_nbody.noaliases.py_python3}  \\
     Benchmark: \texttt{sieve}   & Benchmark: \texttt{spectral\_norm}\\
     \small (50 SLoC, 282 configurations) & \small (44 SLoC, 312 configurations)\\
     \insgraphw{sieve.py_python3} & \insgraphw{bm_spectral_norm.py_python3}  
   \end{tabular}
  \Put(-245,-1100){\inslegend}
  \caption{Typing lattices for Reticulated Python benchmarks under CPython.}
  \label{cp:fig:cpythonperf}
\end{figure*}

Figure \ref{cp:fig:cpythonperf} shows the execution times for
configurations across the typing lattice for each benchmark. Each
graph corresponds to one benchmark, and each red circle in the graph
represents the average execution time of one configuration. The dashed
line marks the execution time of the untyped version of the benchmark
in standard Python 3. (The blue triangles show the performance of
optimized configurations, discussed below in Section
\ref{cp:sec:optperf}). Moving from left to right moves up the lattice
from untyped to typed, execution time is shown on the left $y$ axis,
and relative overhead compared to the untyped program is shown on the
right $y$ axis---higher $y$ coordinates indicates slower performance.

Over the entire typing lattices of all benchmarks, Reticulated Python
incurs an average overhead of 2.21$\xtimes$ compared to the untyped
Python versions of the benchmarks. Typically the slowest
configurations are the ones with the highest type weight: fully typed
configurations have an average overhead of 3.63$\xtimes$. The slowest
configuration is from the \texttt{nbody} benchmark at 5.95$\xtimes$.

Performance degrades as types are added because changing an
annotation from \li{Any} to a static type results in checks being
inserted. The graphs of \texttt{pystone}, \texttt{snake}, and others
display a linear degradation because each check is executed
approximately the same number of times when the configuration
runs. Graphs with greater variance, such as that of
\texttt{spectral\_norm}, arise when some checks are executed more
often than others, so different configurations at the same point in
the lattice perform differently depending on which annotations are
dynamized. The configurations from \texttt{meteor\_contest} form two
large clusters: the parameter \li{fps} of the \li{solve} function has
type \li{List[List[List[Set[int]]]]}; \li{solve} loops over the second
and third dimensions of this value and every iteration of a loop
includes a check if the loop's target is typed, and so replacing \li{fps}'s
annotation with \li{Any} or \li{List[Any]} dramatically reduces the amount
of time spent performing transient checks.

\citet{vitousek:2017trans} examined some of these benchmarks and
compared the performance of untyped configurations to configurations
that were close to fully-typed. We find different results in some of
these cases: for example, we report an overhead of 5.19$\xtimes$ for
the \texttt{nbody} benchmark in the fully-typed configuration, while
they found an overhead of less than 2$\xtimes$. This is because we
increased the expressivity of Reticulated Python's type annotations
and type system to handle a fully static version of \texttt{nbody}
(for example, by allowing type annotations to be placed on functions
with default arguments). On the other hand, some of our changes to
Reticulated Python's semantics resulted in better performance. For
example, Reticulated now checks whether an object is an instance of a
class rather than checking that it supports all the methods of the
class. As a result, the \texttt{spectral\_norm} benchmark's overhead
was reduced from over 5$\xtimes$ to 2.98$\xtimes$ in fully-typed
configurations.

Overall, the performance cost of transient gradual typing in
Reticulated Python is significant, but unlike Typed
Racket~\citep{takikawa:2016gtdead}, the cost is generally predictable
across the lattice and it never approaches worst cases of over
100$\xtimes$. \citet{takikawa:2016gtdead} suggest that in some
contexts an overhead of less than 3$\xtimes$ is a cutoff for
real-world releasability (with the notation \emph{3-deliverable})
while an overhead in the range 3$\xtimes$ to 10$\xtimes$ is usable for
development purposes (written \emph{3/10-usable}). While they note
that such values are ``rather liberal'' and are unacceptable in many
applications, they provide a minimal criterion to evaluate the
acceptability of overheads. With Reticulated Python and the baseline
transient semantics, the average overhead over every sampled
configuration falls within the 3-deliverable range, and all
configurations are at least 3/10-usable. This result is in
concordance with the analysis of \citet{greenman:2017retic}, who found
an overall performance cost of no worse than one order of magnitude in
Reticulated Python programs and also found that the cost of gradual
typing increased as programs became more statically-typed.

\subsection{Module-Based Configurations}

Our sampling methodology is different from the methodology used by
\citet{takikawa:2016gtdead}, in which all possible configurations of
the tested Typed Racket programs were tested, because the space of
possible configurations is much larger in Reticulated Python since
Reticulated Python uses fine-grained gradual typing and Typed Racket
is coarse-grained. We expect that the fact that dynamic and static are
much more intermingled in most of our configurations than they are in
Typed Racket would result in our configurations showing even more
overhead due to run-time type enforcement. However, to ensure that our
sampling did not miss pathological cases that are drawn out by
module-level gradual typing, we generated typing lattices equivalent
to those that would be generated by Typed Racket for two of our test
cases (cases that \citet{takikawa:2016gtdead} also analyzed). For
these cases, the \texttt{sieve} and \texttt{snake} benchmarks, we
recreated configurations equivalent to those possible in Typed
Racket. While the difference in languages precludes a direct
comparison, this analysis ensures that if there was some specific
interaction between static and dynamic displayed in the Typed Racket
configurations that led to a mean overhead of 102.49$\xtimes$ (for
\texttt{sieve}) or 32.30$\xtimes$ (for \texttt{snake}), we would also
encounter it.

\begin{figure}
 \centering
\begin{tabular}{rl|rl}
\texttt{sieve} & (2 modules, 4 config.) & \texttt{snake} & (8 modules, 256 config.)\\
max overhead & 2.20$\times$ & max overhead & 3.81$\times$\\
mean overhead & 1.60$\times$ & mean overhead & 2.28$\times$
\end{tabular}
\caption{Performance of module-based lattices of \texttt{sieve} and \texttt{snake}.}
\label{cp:fig:modules}
\end{figure}

Figure \ref{cp:fig:modules} shows the performance of the typing
lattices of \texttt{sieve} and \texttt{snake} when generated on a
per-module basis \emph{\`{a} la} Typed Racket. The performance of
these configurations are in line with the overall performance of the
benchmarks using our sampling methodology---compare the graphs for
\texttt{sieve} and \texttt{snake} in Figure
\ref{cp:fig:cpythonperf}. In these graphs, the configurations (red
circles) that perform worst show approximately the same proportional
overhead (the scale shown on the right $y$-axis) as the worst case
configurations of the module-based lattices. Similarly, the mean
overheads of the module-based lattices would, if plotted on the graphs
in Figure \ref{cp:fig:cpythonperf}, be close to average. This suggests
that the configurations tested by \citet{takikawa:2016gtdead} are not
exceptional cases.

\subsection{Comparison to Guarded Gradual Typing}

While our performance results compare favorably to those reported by
\citet{takikawa:2016gtdead} in Typed Racket, this is an imprecise
comparison because of the different underlying languages involved. As
a better comparison between the transient approach and the traditional
proxy-based ``guarded'' approach, we took one configuration of the
\texttt{sieve} as a case study. Specifically, this was a configuration
whose equivalent Typed Racket program showed approximately 100$\times$
overhead. For this configuration, we manually created a cast-inserted,
proxy-based version of the program. \citet{vitousek:2014retic,
  vitousek:2017trans} showed that in general the proxy-based or
guarded approach is incompatible with Python, but those
incompatibilities do not arise in this limited example.

We found that the performance of the guarded version of the program
had an overhead of 11.73$\times$ over the standard Python version,
compared to an overhead of 1.43$\times$ for the transient version.
Through varying the parameters of the benchmark, we observe that the
execution times for guarded \texttt{sieve} show the same computational
complexity as untyped and transient \texttt{sieve}, indicating that
the difference in performance is not a result of the use of proxies
increasing the complexity of the program. By instrumenting the guarded
version of the program, we verified that there were no chains of
proxies---in no case was there a value more than two ``layers'' deep
(i.e.\ a proxy of a proxy of a value). The high overhead of the guarded
approach compared to the transient approach, therefore, is
attributable to a large constant factor. Through profiling, we found
that the largest contributor to this overhead was the indirection
performed by proxies at their use sites, followed by casting, proxy
instantiation, and calls from translated user code into the
casting code. By contrast, the transient version's reduced overhead
comes almost entirely from the runtime inspection of values performed
by the code implementing transient checks, and from the calls to these
checks. Python performs these inspections very efficiently, and so the
overall overhead is relatively low, despite the frequency of the
checks.

\section{Optimizing Transient Gradual Typing}\label{cp:sec:infer}

To improve Reticulated Python's performance, we aim to reduce the
number of checks while preserving those required for soundness.  The
basic idea of transient gradual typing is to use pervasive runtime
checks to verify that values correspond to their expected static
types. With the transient approach, type annotations are \emph{untrusted}: they do
not provide information to be relied on, but rather are claims that
must be verified. Therefore, to reduce the runtime burden of transient
gradual typing, we move this verification from runtime to compile time
wherever possible. We do so by using type inference to determine when
types can be \emph{trusted} and do not need runtime verification. 

Our inference process is based on the approaches of
\citet{aiken:1993inclusion} and \citet{Rastogi:2012fk}, using
subtyping constraints and also a new form of constraint, the
\emph{check constraint}. To determine which checks are redundant, our
inference algorithm occurs after transient checks have already been
inserted, because the existence of a check in one part of a program
can allow checks elsewhere to be removed. Check constraints let the
system reason conditionally about checks, and they express the idea of
transient checks: the type of the check expression
$\checkx{e}{S}$ and the type of the expression being checked $e$ are
constrained to be equal \emph{if}, when solved for, the type of $e$ corresponds to the
type tag $S$ (for example, if the type of $e$ is solved to be $\alpha\to\beta$ and $S$
is $\funs$). If that is not the case, for example if $S$ is $\funs$
and the type of $e$ is solved to be $\dyn$, then the type of the overall check
$\checkx{e}{S}$ is constrained to be the most general type that corresponds with $S$ (in this
example, $\dyn\to\dyn$). 

\subsection{Overall approach} \label{cp:sec:infer:approach} 
We generate sets of check constraints and subtype constraints from
programs and find a solution that maps type variables to types, and
then remove redundant checks. Our approach is as follows:

\begin{itemize}
\item Perform transient check insertion as described by
  \citet{vitousek:2017trans}.
\item Assign a unique type variable to every function argument, return
  type, and reference in the program. 
\item Perform a syntax-directed constraint generation pass.
\item Solve the constraint system to obtain a mapping from type
  variables to types.
\item Using this mapping, perform a syntax-directed translation to the
  final target language. For each check in the program, if the
  inferred type of the term being checked and the tag it is checked
  against agree, remove the check, otherwise retain it.
\end{itemize}

\begin{figure*}
  \centering
  \begin{minipage}{1.0\linewidth}
    \centering
    \begin{minipage}{.5\linewidth}
\begin{lstlisting}[language=Python,basicstyle=\small\tt,morekeywords={Callable,Any,int}]
# Program with transient checks and vars

def idDyn(a:alpha)->beta: return a

def makeEq(n:gamma)->delta:
  n@int/*\label{cp:ex:trans:eqenter}*/
  def internal(m:epsilon)->zeta:
    m@int/*\label{cp:ex:trans:intenter}*/
    return n == m
  return internal

eqFive = makeEq(5)@->/*\label{cp:ex:trans:e5}*/
eqFive(20)@bool/*\label{cp:ex:trans:e520}*/
eqFive(idDyn('hello world'))@bool/*\label{cp:ex:trans:e5str}*/
\end{lstlisting}
\centering (a)
    \end{minipage}
    \begin{minipage}{.48\linewidth}
\begin{lstlisting}[language=Python,basicstyle=\small\tt,morekeywords={Callable,Any,int}]
# Final optimized program

def idDyn(a): return a

def makeEq(n):
  def internal(m):
    m@int
    return n == m
  return internal


eqFive = makeEq(5)
eqFive(20)
eqFive(idDyn('hello world'))
\end{lstlisting}
\centering (b)
    \end{minipage}



  \end{minipage}
  \caption{Stages of optimized transient compilation for the program shown in Figure \ref{cp:fig:trans}.}
  \label{cp:fig:steps}
\end{figure*}





As an example, we return to the program shown in Figure
\ref{cp:fig:trans}a, which shows a curried equality function written
in a gradually typed language and which should pass static
typechecking but fail due to a transient check at runtime.  Figure
\ref{cp:fig:steps}a shows the result of this program after the first
phase of our optimizing translation. In this phase, transient checks
have been inserted exactly as in Figure \ref{cp:fig:trans}b, but
instead of the programmer's type annotations being erased, they have
been replaced by type variables
$\alpha,\beta,\gamma,\delta,\epsilon,\zeta$.





Our system infers types (which may be entirely different from the
programmer's annotations) for these type variables by generating
subtyping constraints and special check constraints. Check constraints
are generated by transient checks, and serve to connect the type of
the checked expression with the type it is used at after the
check. For example, at line \ref{cp:ex:trans:e5}, the result of
\li{makeEq(5)} has type $\delta$, and is then checked to ensure that
it is a function (\li{@->}). The type of the result of this check, and
therefore the type of \li{eqFive}, is $\eta\to\theta$, where
$\eta,\theta$ are fresh type variables. This type is linked to
$\delta$ by a check constraint
$\depcon{\delta}{{\funs}}{\eta\to\theta}$, which can be read as ``if
$\delta$ is solved to be a function, then its solution is equal to
$\eta\to\theta$.'' We use check constraints rather than equality
constraints \cite{Hindley:1969ri,Milner:1978kh} because the same
variable can be checked against many different types at different
points in the program. Check constraints are only generated by
transient checks where the checked type tag corresponds to a non-base
type, because constraints of the form $\depcon{\gamma}{\intt}{\intt}$
(as would be generated on line \ref{cp:ex:trans:eqenter}) add no new
information to the system: the type on the right will be $\intt$
whether $\gamma$ is solved to be $\intt$ or not.

Subtyping constraints are also generated from the program. For
example, because the call to \li{makeEq} on line \ref{cp:ex:trans:e5}
has an integer argument, it generates the constraint
$\flows{\intt}{\gamma}$, meaning that $\gamma$ is constrained to be a
supertype of $\intt$. The full set of constraints for this example is:
\begin{gather*}
  \{\flows{\alpha}{\beta},\flows{\epsilon\to\zeta}{\delta},\flows{\boolt}{\zeta},\flows{\intt}{\gamma}, \depcon{\delta}{{\funs}}{\eta\to\theta},\flows{\intt}{\eta},\flows{\beta}{\eta}\}
\end{gather*}


We then solve this constraint set to obtain a mapping from each
variable to a single non-variable type. The only subtyping
constraint on $\delta$ is that $\flows{\epsilon\to\zeta}{\delta}$, so we
determine that $\delta$ must be a function and that
$\epsilon\to\zeta=\eta\to\theta$. This, combined with the fact that
both $\intt$ and (transitively) $\strt$ must be subtypes of $\eta$ due
to the calls on lines \ref{cp:ex:trans:e520} and \ref{cp:ex:trans:e5str},
means that the only solution for $\eta$ and $\epsilon$ is $\dyn$ (the
dynamic type). In all, the solution we find for this constraint set is
\[
\alpha{=}\beta{=}\strt,~\gamma{=}\intt,~\delta{=}{\dyn}{\to}\boolt,~
\epsilon{=}\eta{=}\dyn, ~\zeta{=}\theta{=}\boolt
\]



Some of the transient checks in Figure \ref{cp:fig:steps}a verify
information that the constraint solution has already confirmed. For
example, the check at line \ref{cp:ex:trans:eqenter} verifies that \li{n}
is an integer---but \li{n}'s type $\gamma$ was statically inferred to
be integer, and so this check is not needed. However, the check at
line \ref{cp:ex:trans:intenter}, which verifies that \li{m} is an
integer, is not redundant: the type $\dyn$ was inferred for \li{m}'s
variable $\epsilon$. In fact, this check is needed for soundness
because it will fail with a string on line \ref{cp:ex:trans:e5str}. The
final program, with redundant checks removed and all annotations
erased, is listed in Figure \ref{cp:fig:steps}b.

For the purposes of this example we do not include constraints based
on potential interaction with the open world. If this program were to
be visible to the open world and potentially used by untyped Python
clients, we would need to generate the additional constraints
$\flows{\dyn}{\alpha}$ and $\flows{\dyn}{\gamma}$ in order to maintain
open-world soundness, because the open world could pass arbitrary values
into these functions \cite{vitousek:2017trans}.

\subsection{Constraint Generation with Check Constraints}

\begin{figure*}
  \begin{minipage}{1.0\linewidth}
\[
\begin{array}{lrcl}
  \text{variables} & x,y \\
  \text{numbers} & n & \in & \mathbb{Z}\\
  \langs~\text{expressions} & s & ::= & x \mid n \mid s + s \mid \fntx{x}{U}{U}{s} \mid s~s \mid  \refxi{s}{U} \mid  \derefx{s} \mid \mutx{s}{s}\\
  \langd~\text{expressions} & d & ::= & x \mid n \mid d + d \mid \fntx{x}{\alpha}{\alpha}{d} \mid  d~d \mid \letx{x}{d}{d} \mid \refxi{d}{\alpha} \mid \\&&& \derefx{d} \mid  \mutx{d}{d} \mid \checkx{d}{S}\\
  \text{types} & U & ::= & \dyn \mid \intt \mid \reft{U} \mid U\to U\\
  \text{type tags} & S & ::= & \dyn \mid \intt \mid ~\funs~ \mid \refs \\
\end{array}
\]  
\boxed{U \sim U}
\begin{gather*}
\inferrule{ }{\dyn \sim U}\qquad
\inferrule{ }{U \sim \dyn}\qquad
\inferrule{ }{\intt \sim \intt}\qquad
\inferrule{U_1 \sim U_3 \\ U_2 \sim U_4}{\funt{U_1}{U_2} \sim \funt{U_3}{U_4}}\qquad
\inferrule{U_1 \sim U_2}{\reft{U_1} \sim \reft{U_2}}
\end{gather*}

\end{minipage}

\begin{minipage}{1.0\linewidth}
  \begin{minipage}{.5\linewidth}
\boxed{\matchrel{U}{U}}
\[
\begin{array}{rcl}
  \funt{U_1}{U_2} & \rhd & \funt{U_1}{U_2} \\
  \dyn & \rhd & \funt{\dyn}{\dyn} \\
  \reft{U} & \rhd & \reft{U} \\
  \dyn & \rhd & \reft{\dyn}
\end{array}
\]
  \end{minipage}
  \begin{minipage}{.5\linewidth}    
\boxed{\lowlty{U}=S}
\[
\begin{array}{rcl}
  \lowlty{\funt{U_1}{U_2}} & = & \funs \\
  \lowlty{\reft{U}} & = & \refs \\
  \lowlty{\intt} & = & \intt \\
  \lowlty{\dyn} & = & \dyn
\end{array}
\]
  \end{minipage}
\end{minipage}

\boxed{\Gamma\vdash s\leadsto d:U}\hspace{4in}
\begin{gather*}
\inferrule*{\Gamma,x{:}U_1\vdash s\leadsto d:U_2' \\ U_2' \sim U_2 \\ \alpha,\beta~\text{fresh}}
  {\Gamma\vdash \fntx{x}{U_1}{U_2}{s}\leadsto\fntx{x}{\alpha}{\beta}{\letx{x}{\checkx{x}{\lowlty{U_1}}}{d}}:\funt{U_1}{U_2}} \\[1ex]
\inferrule*{\Gamma\vdash s_1\leadsto d_1:U \\ \matchrel{U}{\funt{U_1}{U_2}} \\\\
           \Gamma\vdash s_2\leadsto d_2:U_1' \\ U_1' \sim U_1}
  {\Gamma\vdash s_1~s_2\leadsto (\checkx{(\checkx{d_1}{\funs})~d_2)}{\lowlty{U_2}}:U_2}\qquad
\inferrule*{\Gamma(x)=U}{\Gamma\vdash x\leadsto x:U} \qquad 
\inferrule*{ }{\Gamma\vdash n\leadsto n:\intt}\\[1ex]
\inferrule*{\Gamma\vdash s\leadsto d:U_2 \\ U_2\sim U_1 \\ \alpha~\text{fresh}}{\Gamma\vdash \refxi{s}{U_1}\leadsto \refxi{d}{\alpha}:\reft{U_1}}\qquad
\inferrule*{\Gamma\vdash s\leadsto d:U \\ \matchrel{U}{\reft{U'}}}
  {\Gamma\vdash\derefx{s}\leadsto \checkx{\derefx{(\checkx{d}{\refs})}}{\lowlty{U'}}:U'}\\[1ex]
\inferrule*{\Gamma\vdash s_1\leadsto d_1:U \\ \matchrel{U}{\reft{U'}} \\\\
           \Gamma\vdash s_2\leadsto d_2:U'' \\ U'' \sim U'}
  {\Gamma\vdash\mutx{s_1}{s_2}\leadsto \mutx{(\checkx{d_1}{\refs})}{d_2}:\intt}\qquad
\inferrule*{\Gamma\vdash s_1\leadsto d_1:U_1 \\ U_1\sim \intt \\\\
           \Gamma\vdash s_2\leadsto d_2:U_2 \\ U_2\sim \intt}
  {\Gamma\vdash s_1~+~s_2\leadsto \checkx{d_1}{\intt}~+~\checkx{d_2}{\intt}:\intt}
\end{gather*}

  \caption{The $\langs$ calculus and translation from $\langs$ to
    $\langd$ (based on \citet{vitousek:2017trans}).}
  \label{cp:fig:stod}
\end{figure*}

In this section we describe our approach to generating type
constraints for programs in a transient calculus $\langd$.  Programs
in $\langd$ are not the surface programs written by the programmer,
and $\langd$ is not a gradually typed language. Instead, $\langd$
programs are the result of translating programs in a gradually typed
surface language $\langs$ into $\langd$. This translation, which along
with the syntax for $\langs$ is shown in Figure \ref{cp:fig:stod},
inserts transient checks throughout translated programs in order to
enforce the type annotations present in the surface program, exactly
as described by \citet{vitousek:2017trans}. The $\langd$ calculus is
analogous to the cast calculi in guarded gradual typing
\cite{Siek:2006bh, Wadler:2009qv,Herman:2006uq} and the transient
translation from $\langs$ to $\langd$ is nearly identical to the
translation from the surface $\lambda^\dyn_\to$ calculus to the target
$\lambda^{\Downarrow}_\ell$ calculus presented by
\citet{vitousek:2017trans} (except that the calculi presented here
elide features needed for blame tracking). The $\langs$ calculus
includes functions and mutable references (with syntax $\refxi{s}{U}$
for introducing a reference with type $U$, $\derefx{s}$ for
dereferencing, and $\mutx{s}{s}$ for mutation).

Figure \ref{cp:fig:stod} shows the syntax for $\langd$, which
supports all the features of $\langs$ as well as transient checks,
written $\checkx{d}{S}$. The meta-variable $d$ ranges over expressions
of $\langd$.  In the dynamic semantics for $\langd$ (defined by
translation into a third calculus $\lange$ in Figure
\ref{cp:fig:dtoe} with evaluation rules given in Figure
\ref{cp:fig:dynsem}) such checks examine the value of an expression
$d$ to determine if it corresponds to the \emph{type tag} $S$, and
fails if not. Type tags, shown in Figure \ref{cp:fig:stod},
correspond to type constructors. Functions and references are
annotated with type variables.

\begin{figure*}
  \begin{minipage}{1.0\linewidth}
    \begin{minipage}{.7\linewidth}
      
\[
\begin{array}{rrcl}
  \text{type variables} & \alpha,\beta,\gamma \\
  \text{leaf types} & V & ::= & \alpha \mid \dyn \\
  \text{constraint types} & A & ::= & V \mid \intt \mid \reft{V} \mid \funt{V}{V} \\
  \text{constraints} & C & ::= & \flows{A}{A} \mid \depcon{A}{S}{A} \\
  \text{constraint sets} & \Omega & \in & \mathcal{P}(C)
\end{array}
\]
    \end{minipage}
    \begin{minipage}{.2\linewidth}
\boxed{\specmatches{A}{A}{S}}
\[
\begin{array}{rcl}
  \funt{V_1}{V_2} & \rhd_\funs & \funt{V_1}{V_2} \\
  \alpha & \rhd_\funs & \funt{\beta}{\gamma} \\& \multicolumn{2}{l}{ \;\;\text{with}~\beta,\gamma~\text{fresh}} \\
  \reft{V} & \rhd_\refs & \reft{V} \\
  \alpha & \rhd_\refs & \reft{\beta} \\& \multicolumn{2}{l}{\;\;\text{with}~\beta~\text{fresh}}\\
  \intt & \rhd_\intt & \intt \\
  V & \rhd_\intt & \intt\\
  A & \rhd_\dyn & A
\end{array}
\]
      
    \end{minipage}\\[2ex]

\boxed{\Gamma\vdash d:A\yO}
\begin{gather*}
\inferrule{\Gamma\vdash d:A_1\yO \\ \specmatches{A_1}{A_2}{S}}
  {\Gamma\vdash \checkx{d}{S}:A_2\yO \cup \{\depcon{A_1}{S}{A_2}\}}\qquad
\inferrule{\Gamma,x{:}\alpha\vdash d:A\yO}
  {\Gamma\vdash \fntx{x}{\alpha}{\beta}{d}:\funt{\alpha}{\beta}\yO\cup\{\flows{A}{\beta}\}} \\[1ex]
\inferrule{\Gamma\vdash d_1:V_1\to V_2\yO_1 \\
           \Gamma\vdash d_2:A\yO_2}
  {\Gamma\vdash d_1~d_2:V_2\yO_1\cup\Omega_2\cup\{\flows{A}{V_1}\}}\qquad
\inferrule{\Gamma\vdash d:A\yO}{\Gamma\vdash \refxi{d}{\alpha}:\reft{\alpha}\yO\cup\{\flows{A}{\alpha}\}}\\[1ex]
\inferrule{\Gamma\vdash d_1:\reft{V}\yO_1 \\
           \Gamma\vdash d_2:A\yO_2}
  {\Gamma\vdash\mutx{d_1}{d_2}:\intt\yO_1\cup\Omega_2\cup\{\flows{A}{V}\}}\qquad
\inferrule{\Gamma\vdash d_1:\intt\yO_1 \\
           \Gamma\vdash d_2:\intt\yO_2}
  {\Gamma\vdash d_1~+~d_2:\intt\yO_1\cup\Omega_2}\\[1ex]
\inferrule{\Gamma(x)=A}{\Gamma\vdash x:A;\emptyset} \qquad 
\inferrule{ }{\Gamma\vdash n:\intt;\emptyset}\qquad
\inferrule{\Gamma\vdash d:\reft{V}\yO}
  {\Gamma\vdash\derefx{d}:V\yO}
\end{gather*}
  \end{minipage}
  \caption{Syntax and constraint generation for $\langd$}
  \label{cp:fig:constrgen}
\end{figure*}

To infer types for these type variables such that the overall program
is well-typed, we first generate constraints using the syntax-directed
rules defined in Figure \ref{cp:fig:constrgen}, in the style of
\citet{Aiken:1995ly}. These rules generate sets $\Omega$ of
constraints $C$ over types $A$, also defined in Figure
\ref{cp:fig:constrgen}. Types $A$ are not inductively
defined---function and reference types can only contain type variables
or $\dyn$, not arbitrary types. 

The rules in Figure \ref{cp:fig:constrgen} generate
constraints. Subtyping constraints are generated from function and
reference introduction and elimination sites, to ensure that any
solution found for these variables is well-typed. 



The rule for transient checks differs from the others.  First, there
is the question of what constraint type to give the result of a
check. Consider the program
$\fntx{x}{\alpha}{\beta}{(\checkx{x}{{\funs}})}$.  In the body of the
function $x$ has type $\alpha$, but the type of the check expression
$\checkx{x}{{\funs}}$ ought to be a function, because the
check will fail at runtime if $x$ is \emph{not} a function. The check
cannot, however, specify argument and return types. Therefore the type
of $\checkx{x}{{\funs}}$ is a function whose argument and return types
are fresh type variables. This type is obtained using the $\rhd_S$
relation in Figure \ref{cp:fig:constrgen}, where $S$ is the type tag
checked against (in this case $\funs$). If the type on the left
already corresponds to $S$, the type on the right is the same, but if
it is a variable the right-hand side is a new type that corresponds
to $S$ but is otherwise inhabited by fresh variables.

Checks do not generate subtype constraints: checking that something with type $\alpha$ is a function should
not introduce the constraint $\flows{\alpha}{\beta\to\gamma}$, because
the solution to $\alpha$ might not actually be a function. For
example, if all values that flow into a variable with type $\alpha$
are integers, then at runtime this check will fail, which
is an acceptable behavior for gradually typed programs. However, if
only functions inhabit $x$, then the argument and return types of
those functions must be equal to $\beta$ and $\gamma$
respectively. Check expressions instead generate check constraints,
written $\depcon{A_1}{S}{A_2}$, which constrain the solution for
$A_2$ so that, if the solution of $A_1$ corresponds to the type tag
$S$, then $A_1=A_2$.

Note that these type-directed constraint generation rules show
that $\langd$ is \emph{not} a gradually typed language: for example,
the application rule requires that the expression in function position
has a function type, not that it be consistent with a function
type. This is because the translation process in Figure
\ref{cp:fig:stod} has inserted check expressions throughout the
program already; any well-typed, closed $\langs$ terms can be
translated to a well-typed $\langd$ term.

\begin{theorem}
  If $\emptyset\vdash s\leadsto d:U$, then $\emptyset\vdash d:A;\Omega$.
\end{theorem}
This lemma is an immediate corollary of Lemma
\ref{cp:lem:apx:stodpres}, given in Appendix \ref{cp:apx:proof}.

\subsubsection{Constraints from the open world}

The constraint generation system in Figure \ref{cp:fig:constrgen} is
sufficient to find a solution if the program does not interact with
external code (a \emph{closed} world), but not if the program can
interact with code that may not know about or respect the types it
expects. \citet{vitousek:2017trans} showed that transient check
insertion ensures safety in open world contexts, but optimizing
transient programs based only on \emph{internal} information could
result in the deletion of checks critical to preserving open-world
safety.

Fortunately, \citet{Rastogi:2012fk} observed that the overall type of
a program will encode the information flows that it exchanges with the
open world. Figure \ref{cp:fig:owconstr} shows additional constraints
generated from the overall type of a program to protect it from the
open world. Constraints such as $\dyn <: V$ constrain $V$ to be
dynamic, while constraints like $V <: \dyn$ allow $V$ to be more
specific types, but guarantee that any type variables that flow
through $V$ in a contravariant position will be constrained to
$\dyn$. For example, given the constraints 
\[V_1 <: \dyn,~ V_2\to V_3 <: \dyn \]
the rules shown below in Section \ref{cp:sec:alg} guarantee that $V_2$
will be constrained to $\dyn$; this is essential because $V_2$ is in a
contravariant position, and if the a function with type $V_2\to V_3$
flows into the open world, the open world can pass whatever it wishes
into $V_2$. This also allows our analysis to be modular: individual
modules can be optimized in isolation by making pessimistic
assumptions about what kinds of values flow from one module to
another.

\begin{figure}
  \begin{minipage}{1\linewidth}

  \boxed{\vdash A : \Omega}
\begin{gather*}
\inferrule{ }{\vdash V : \{\flows{V}{\dyn}\}}\qquad
\inferrule{ }{\vdash \intt : \emptyset}\qquad
\inferrule{ }{\vdash \reft{V} : \{\flows{V}{\dyn},\flows{\dyn}{V}\}} \qquad
\inferrule{ }{\vdash \funt{V_1}{V_2} : \{\flows{\dyn}{V_1},\flows{V_2}{\dyn}\}}
\end{gather*}
    
  \end{minipage}

  \caption{Open world constraint generation.}
  \label{cp:fig:owconstr}
\end{figure}






\subsection{Computing Constraint Solutions}\label{cp:sec:alg}

We compute solutions $\sigma$ for constraint sets $\Omega$ with an
approach based on the algorithm of Aiken and collaborators
\cite{Aiken:1995ly,Aiken:1994fk,aiken:1993inclusion}, but with the
addition of check constraints, requiring constraint sets to be solved
incrementally.

\begin{figure}
  \begin{minipage}{1.0\linewidth}
    
\[
\begin{array}{rrcl}
  \text{types} & T & ::= & T \to T \mid \reft{T} \mid \intt \mid \dyn \mid \alpha \\
  \text{constraints} & C & ::= & \flows{A}{A} \mid \depcon{A}{S}{A} \colorbox{tableblue}{$\mid A = A \mid \tagcon{\alpha}{S} \mid \defcon{\alpha}{T}$} \\
  \text{maps} & \sigma & ::= & \alpha\mapsto T, \ldots,\alpha\mapsto T
\end{array}
\]

\begin{minipage}{1\linewidth}
\begin{minipage}{.5\linewidth}
\boxed{\mathit{parts}(A)=\{\alpha,\ldots,\alpha\}}
\[
\begin{array}{rcl}
  \mathit{parts}(\dyn) & = & \emptyset \\
  \mathit{parts}(\intt) & = & \emptyset \\
  \mathit{parts}(\alpha) & = & \emptyset \\
  \mathit{parts}(\alpha_1\to\alpha_2) & = & \{\alpha_1,\alpha_2\} \\
  \mathit{parts}(\reft{\alpha}) & = & \{\alpha\}
\end{array}
\]
\end{minipage}
\hfill
\begin{minipage}{.45\linewidth}
\boxed{\uplty{S}=A}
\[
\begin{array}{rcl}
  \uplty{\intt} & = & \intt\\
  \uplty{\refs} & = & \reft{\dyn} \\
  \uplty{\funs} & = & \dyn\to\dyn \\
  \uplty{\dyn} & = & \dyn
\end{array}
\]
\end{minipage}
\end{minipage}

\boxed{\Omega \steps \Omega}
\[
\begin{array}{rl}
 \Omega \cup \{\flows{\funt{V_1}{V_2}}{\dyn}\} & \steps \Omega \cup \{\flows{\dyn}{V_1},~\flows{V_2}{\dyn}\} \\ 
 \Omega \cup \{\flows{\funt{V_1}{V_2}}{\funt{V_3}{V_4}}\} & \steps \Omega \cup \{\flows{V_3}{V_1},~\flows{V_2}{V_4}\} \\ 
 \Omega \cup \{\flows{\reft{V}}{\dyn}\} & \steps \Omega \cup \{V=\dyn\} \\ 
 \Omega \cup \{\flows{\reft{V_1}}{\reft{V_2}}\} & \steps \Omega \cup \{V_1=V_2\}\\
 \Omega \cup \{\flows{V}{V}\} & \steps \Omega \\
 \Omega \cup \{\flows{\intt}{\dyn}\} & \steps \Omega \\
 \Omega \cup \{\funt{V_1}{V_2}=\funt{V_3}{V_4}\} & \steps \Omega \cup \{V_1=V_3,V_2=V_4\} \\ 
 \Omega \cup \{\reft{V_1}=\reft{V_2}\} & \steps \Omega \cup \{V_1=V_2\}\\
 \Omega \cup \{A=A\} & \steps \Omega \\
 \Omega \cup \{A=\alpha\} & \steps \Omega\cup \{\alpha=A\} \\
  &
  \begin{array}{ll}
    \quad\text{where} & A\neq\alpha'
  \end{array}\\
 \Omega\cup \{\alpha=A\} & \steps \Omega[\alpha/A] \cup \{\defcon{\alpha}{A}\}\\
  &
  \begin{array}{ll}
    \quad\text{where} & \alpha\not\in\mathit{vars}(A)~\text{and}~ (\defcon{\alpha}{T})\not\in\Omega
  \end{array}\\
  \Omega\cup \{\depcon{A}{S}{A} \} & \steps \Omega\\
  \Omega\cup \{\tagcon{\alpha}{S}, \depcon{\alpha}{S}{A} \} & \steps \Omega\cup \{\alpha = A\}\\
  &
  \begin{array}{ll}
    \quad\text{where} & (\depcon{\alpha}{S'}{A'}) \not\in \Omega\\
    & A \neq \alpha
  \end{array}\\
  \Omega\cup \{\tagcon{\alpha}{S_1}, \depcon{\alpha}{S_2}{A} \} & \steps \Omega\cup \{\tagcon{\alpha}{S_1}\}\cup \{\alpha' = \dyn \mid \forall \alpha' \in \mathit{parts}(A)\} \\
  &
  \begin{array}{ll}
    \quad\text{where} & S_1 \neq S_2
  \end{array}\\
  \Omega\cup \{\depcon{\alpha}{S}{A_1}, \depcon{\alpha}{S}{A_2} \} & \steps \Omega\cup \{\depcon{\alpha}{S}{A_1},A_2=A_1\} \\
  \Omega\cup \{\tagcon{\alpha}{S} \} & \steps \Omega\cup\{\alpha=A\}\\
  &
  \begin{array}{ll}
    \quad\text{where} & (\depcon{\alpha}{S'}{A'}) \not\in \Omega ~\text{and}~(\defcon{\alpha}{T}) \not\in\Omega~\text{and}~\\
    & (\alpha=A') \not\in\Omega~\text{and}~\specmatches{\alpha}{A}{S}
  \end{array}
\end{array}
\]
  \end{minipage}
  \caption{Simplification of constraint sets.}
  \label{cp:fig:constsimpl}
\end{figure}

\begin{figure}
  
\begin{minipage}{1.0\linewidth}
\begin{minipage}{1.0\linewidth}
\begin{minipage}{.5\linewidth}
\boxed{S \sqcup S = S}
\[
\begin{array}{rcll}
  S \sqcup S & = & S \\
  S_1 \sqcup S_2 & = & \dyn & \text{if}~S_1\neq S_2
\end{array}
\]
\end{minipage}
\hfill
\begin{minipage}{.45\linewidth}
\boxed{\lowlty{A}=S}
\[
\begin{array}{rcl}
  \lowlty{\intt} & = & \intt\\
  \lowlty{\reft{V}} & = & \refs \\
  \lowlty{V_1\to V_2} & = & \funs \\
  \lowlty{\dyn} & = & \dyn
\end{array}
\]
\end{minipage}
\end{minipage}

\boxed{\Omega\Downarrow\sigma}
  \begin{gather*}
  \inferrule{\Omega=\Omega' \cup \{\flows{A_1}{\alpha},\ldots,\flows{A_n}{\alpha}\} \\ \alpha \in \mathit{vars}(\Omega') \vee n > 0 \\ \forall i \le n,~A_i\neq\alpha'\wedge \alpha\not\in\mathit{vars}(A_i) \\ S=\sqcup_{i\le n} \lowlty{A_i} \\ \vsol{\Omega'}{\alpha} \\ 
             \Omega~\text{normal} \\ \Omega\cup \{\tagcon{\alpha}{S}\} \Downarrow \sigma}
            {\Omega \Downarrow\sigma}\\[2ex]
  \inferrule{\Omega\steps^*\Omega' \\ \Omega' \Downarrow \sigma}
            {\Omega\Downarrow \sigma}\qquad
  \inferrule{\Omega=\{\defcon{\alpha_1}{T_1},\ldots,\defcon{\alpha_n}{T_n}\} \\\\ \forall \alpha,i \le n,~\alpha\not\in\mathit{vars}(T_i)}
            {\Omega\Downarrow \alpha_1{\mapsto}T_1,\ldots,\alpha_n{\mapsto}T_n}
\end{gather*}
\boxed{\vsol{\Omega}{\alpha}}
\begin{gather*}
  \inferrule{\Omega=\Omega_1\cup\Omega_2\cup\Omega_3\cup\Omega_4
             \\ \Omega_2 = \{\defcon{\alpha_1}{T_1},\ldots,\defcon{\alpha_n}{T_n}\} \\ \Omega_3 = \{\depcon{\alpha}{S_1}{A_1},\ldots,\depcon{\alpha}{S_m}{A_m}\} \\ \Omega_4 =  \{\flows{\alpha}{V_1},\ldots,\flows{\alpha}{V_p}\}
             \\ \alpha\not\in\mathit{vars}(\Omega_1) \\ \forall i \le n,~\alpha_i \neq \alpha \\ \forall j \le m,~\alpha\not\in\mathit{vars}(A_j)}
            {\vsol{\Omega}{\alpha}}
\end{gather*}
\end{minipage}

  \caption{Solving constraint sets.}
  \label{cp:fig:constsolve}
\end{figure}

Figure \ref{cp:fig:constsimpl} shows rules for simplifying constraint
sets. 
These rules introduce three additional forms of constraints:
\emph{equality constraints} $A_1=A_2$ indicate equality between $A_1$
and $A_2$ \cite{Hindley:1969ri,Milner:1978kh}, \emph{tag constraints}
$\tagcon{\alpha}{S}$ indicate that the solution to the variable
$\alpha$ must have the type constructor corresponding to $S$, and
\emph{definition constraints} $\defcon{\alpha}{T}$ constrain the
solution of variables $\alpha$ to be exactly $T$, which is a ``full''
inductive type as shown in Figure \ref{cp:fig:constsimpl}.

The first six rules decompose subtyping constraints, and are followed
by rules for decomposing equalities. Equality constraints on variables
$\alpha=A$ immediately become definition constraints
$\defcon{\alpha}{A}$ (exploiting the fact that all shallow types $A$
are syntactically also full types $T$) and result in $\alpha$ being
substituted by $A$ in the rest of the constraint set. Tag constraints
$\tagcon{\alpha}{S_1}$ and check constraints $\depcon{\alpha}{S_2}{A}$
combine to generate equality constraints for $\alpha$: if the tag
constraint's tag and the check constraint's tag are equal ($S_1=S_2$)
then the constraint $\alpha=A$ is added. Otherwise, the leaves
(or \emph{parts}) of $A$ are constrained to be $\dyn$, indicating that
the system was unable to prove that the check that generated the check
constraint will succeed. The next rule causes multiple check
constraints on the same variable and tag to be combined, and finally,
if the system contains a tag constraint but no check constraint on its
variable, the variable is constrained to be a type with a constructor
corresponding to $S$ but fresh variables in its leaves.

Figure \ref{cp:fig:constsolve} shows the solution algorithm
$\Omega\Downarrow\sigma$. This algorithm applies the simplification
rules until exhaustion and then selects an unsolved variable $\alpha$
and determines its tag. The variable to be solved can only appear at
the top level of subtyping constraints (e.g.,
$\flows{\alpha\to\beta}{\dyn}\not\in\Omega$). Further, $\alpha$ must
have only other variables or $\dyn$ as upper bounds---but this
requirement is satisfied by all constraints generated from the rules
in Figure \ref{cp:fig:constrgen}. These conditions are specified by the
$\vsol{\Omega}{\alpha}$ relation in Figure \ref{cp:fig:constsolve}. If
these conditions hold, $\alpha$'s tag is the join of all the lower
bounds of $\alpha$ (using the $\sqcup$ operator in Figure
\ref{cp:fig:constsolve}). The resulting tag constraint is then added to
$\Omega$ and the result is simplified. This process terminates once
$\Omega$ is only inhabited by definition constraints, and results in a
substitution $\sigma$ generated from these constraints.

\subsection{Check Removal}

\begin{figure}
\hrulefill\vspace{1ex}

\[
\begin{array}{rcl}
  a & \in & \text{addresses}\\
  e & ::= & a \mid x \mid n \mid e + e \mid \fnx{x}{e} \mid (e~e) \mid  \letx{x}{e}{e} \mid \refx{e} \mid \derefxp{e}{} \mid \mutxpINV{e}{e}{} \mid \checkx{e}{S} \mid  \mathtt{fail}\\
  \Sigma & ::= & \cdot \mid \Sigma,a{:}T 
\end{array}
\]
  \caption{Syntax for $\lange$.}
  \label{cp:fig:lange}
\end{figure}

\begin{figure*}
  \begin{minipage}{1.0\linewidth}
    
\boxed{\Gamma;\sigma\vdash d \leadsto e: T}
\begin{gather*}
\inferrule[DAbs]{\Gamma,x{:}\alpha;\sigma\vdash d \leadsto e : T \\ \vdash T <: \sigma\beta}
 {\Gamma;\sigma\vdash \fntx{x}{\alpha}{\beta}{d} \leadsto \fnx{x}{e} : \funt{\sigma\alpha}{\sigma\beta}}\qquad
\inferrule[DLet]{\Gamma;\sigma\vdash d_1 \leadsto e_1 : T_1 \\ \Gamma,x{:}T_1;\sigma\vdash d_2 \leadsto e_2 : T_2}
 {\Gamma;\sigma\vdash\letx{x}{d_1}{d_2} \leadsto \letx{x}{e_1}{e_2} : T_2}\\[1ex]
\inferrule[DApp]{\Gamma;\sigma\vdash d_1 \leadsto e_1 : \funt{T_1}{T_2} \\
   \Gamma;\sigma\vdash d_2 \leadsto e_2 : T_1' \\ \vdash T_1' <: T_1}
 {\Gamma;\sigma\vdash d_1~d_2 \leadsto (e_1~e_2) : T_2}\qquad
\inferrule[DRef]{\Gamma;\sigma\vdash d \leadsto e : T \\ \vdash T <: \sigma\alpha}{\Gamma;\sigma\vdash \refxi{d}{\alpha} \leadsto \refx{e} :\reft{\sigma\alpha}}\\[1ex]
\inferrule[DDeref]{\Gamma;\sigma\vdash d \leadsto e : \reft{T}}
 {\Gamma;\sigma\vdash\derefx{d}\leadsto \derefxp{e}{} : T}\qquad
\inferrule[DUpdt]{\Gamma;\sigma\vdash d_1 \leadsto e_1 : \reft{T} \\\\
   \Gamma;\sigma\vdash d_2 \leadsto e_2 :
 T' \\ \vdash T' <: T}
 {\Gamma;\sigma\vdash\mutx{d_1}{d_2} \leadsto \mutxpINV{e_1}{e_2}{} :\intt}\\[1ex]
\inferrule[DVar]{\Gamma(x)=A}{\Gamma;\sigma\vdash x \leadsto x : \sigma A} \qquad 
\inferrule[DInt]{ }{\Gamma;\sigma\vdash n \leadsto n :\intt}\\[1ex]
\inferrule[DAdd]{\Gamma;\sigma\vdash d_1 \leadsto e_1 : \intt \\ 
   \Gamma;\sigma\vdash d_2 \leadsto e_2 : \intt }
 {\Gamma;\sigma\vdash d_1~+~d_2 \leadsto e_1~+~e_2 :\intt}\qquad
\inferrule[DCheckRemove]{\Gamma;\sigma\vdash d \leadsto e : T \\ \lowlty{T} \preceq S}
 {\Gamma;\sigma\vdash \checkx{d}{S} \leadsto e : T}\\[1ex]
\inferrule[DCheckKeep]{\Gamma;\sigma\vdash d \leadsto e: \dyn \\ S \neq \dyn}
 {\Gamma;\sigma\vdash \checkx{d}{S} \leadsto \checkx{e}{S}: \uplty{S}}\qquad
\inferrule[DCheckFail]{\Gamma;\sigma\vdash d \leadsto e: T \\ T \neq \dyn \\ \lowlty{T} \not\preceq S}
 {\Gamma;\sigma\vdash \checkx{d}{S} \leadsto \mathtt{fail}: T'}
\end{gather*}
  \end{minipage}
  \caption{Translation from $\langd$ to $\lange$.}
  \label{cp:fig:dtoe}
\end{figure*}

When a constraint set $\Omega$ is solved by a substitution $\sigma$,
the types in $\sigma$ can be relied on because they were inferred
directly from the program. They may be more or less precise, or
entirely different, from the programmer's annotated types. Figure
\ref{cp:fig:dtoe} shows the rules for translating the program to final
target language, $\lange$, which is defined in Figure
\ref{cp:fig:lange}. This translation uses $\sigma$ to decide which
checks can be deleted, as shown in rules \textsc{DCheckRemove},
\textsc{DCheckKeep}, and \textsc{DCheckFail}: if the inferred type for
the checked expression $d$ is found to be more or equally precise
(with the $\preceq$ operator) than the tag it is checked against, then
the check can be removed by \textsc{DCheckRemove}. If the type of $d$
is less precise (which is only the case when $d$'s type is $\dyn$),
the check remains by \textsc{DCheckKeep}, and if $S$ and $d$'s type
are unrelated, then the check will \emph{always} fail at runtime, so
by \textsc{DCheckFail} the check is replaced by $\texttt{fail}$, an
expression which errors when evaluated. Since this expression always
fails, it would be reasonable to additionally warn the programmer of
the error at compile-time, and Reticulated Python's implementation of
this algorithm does so. However, the \textsc{DCheckFail} rule does not
cause the program to be statically rejected, because this analysis is
a runtime optimization that maintains the semantics of the gradually
typed language, where a check would detect a runtime error and fail.

\subsubsection{Soundness of constraint solving}

To prove that the solution algorithm shown in Figure
\ref{cp:fig:constsolve} generates valid solutions to constraint sets,
we must first define what a \emph{solution} to a constraint set
$\Omega$ must consist of. A constraint set is solved by a mapping
$\sigma$ if all the constraints in $\Omega$ are satisfied, and a
subtyping constraint $\flows{A_1}{A_2}$ is satisfied if
$\vdash \sigma A_1 <: \sigma A_2$, where $\sigma A$ is the
substitution of all variables in $A$ with their definitions in
$\sigma$, resulting in a type $T$, and using the subtyping relation on
$T$ defined in Figure \ref{cp:fig:subty}. Note that the subtyping
rules for types $T$ do not admit $\dyn$ as a universal supertype
(i.e.\ $\top$): function types are only subtypes of $\dyn$ if the
function is a subtype of $\dyn\to\dyn$, which in turn requires the
function's source type to be $\dyn$ itself. This is required to ensure
that function types that are passed into $\dyn$ cannot make any
assumptions about what arguments are passed into them. Similar
reasoning applies to why references are only subtypes of $\dyn$ if
they are subtypes of $\reft{\dyn}$.

\begin{figure*}
\hfill\vspace{1ex}

  \begin{minipage}{1\linewidth}

\begin{minipage}{.33\linewidth}
\boxed{\uplty{S}=T}
\[
\begin{array}{rcl}
  \uplty{\funs} & = & \dyn\to\dyn \\
  \uplty{\refs} & = & \reft{\dyn} \\
  \uplty{\intt} & = & \intt \\
  \uplty{\dyn} & = & \dyn 
\end{array}
\]  
\end{minipage}
\begin{minipage}{.33\linewidth}
\boxed{\lowlty{T}=S}
\[
\begin{array}{rcl}
  \lowlty{\funt{T_1}{T_2}} & = & \funs \\
  \lowlty{\reft{T}} & = & \refs \\
  \lowlty{\intt} & = & \intt \\
  \lowlty{\dyn} & = & \dyn
\end{array}
\]  
\end{minipage}
\begin{minipage}{.33\linewidth}
\boxed{S \preceq S}
\begin{gather*}
\inferrule{ }
  {\funs~\preceq~\funs}\qquad
\inferrule{ }
  {\refs \preceq \refs}\\[1ex]
\inferrule{ }
  {\intt \preceq \intt}\qquad
\inferrule{ }
  {S \preceq \dyn}
\end{gather*}
\end{minipage}

\vspace{1ex}
  \boxed{\vdash T <: T}
\begin{gather*}
\inferrule{ }{\vdash \intt <: \intt}\qquad
\inferrule{\vdash T_3 <: T_1 \\ \vdash T_2 <: T_4}{\vdash \funt{T_1}{T_2} <: \funt{T_3}{T_4}}\qquad
\inferrule{\vdash T_1 <: T_2 \\ \vdash T_2 <: T_1}{\vdash \reft{T_1} <: \reft{T_2}}\\[1ex]
\inferrule{ }{\vdash \dyn <: \dyn}\qquad
\inferrule{ }{\vdash \intt <: \dyn}\qquad
\inferrule{\vdash \funt{T_1}{T_2} <: \funt{\dyn}{\dyn}}{\vdash \funt{T_1}{T_2} <: \dyn}\qquad
\inferrule{\vdash \reft{T} <: \reft{\dyn}}{\vdash \reft{T} <: \dyn}
\end{gather*}    
  \end{minipage}
\caption{Full types and subtyping.}
  \label{cp:fig:subty}
\end{figure*}

A check constraint $\depcon{A_1}{S}{A_2}$ can be satisfied in one of
two ways: first, if $\sigma A_1$ corresponds to the tag $S$, then the
check constraint is simply an equality constraint, and it is satisfied
if $\sigma A_1$ is syntactically equal to $\sigma A_2$, written $\vdash \sigma A_1=\sigma A_2$. On the other hand, if $\sigma A_1$ does
\emph{not} correspond to $S$, then all that can be known about
$\sigma A_2$ is that it \emph{does} correspond to $S$, since the
transient check that generated the constraint will fail if values
flowing through it do not.  If $S$ is $\funs$, then $\lowlty{\sigma A_2}=~\funs$ but
the argument and return types of $\sigma A_2$ must be dynamic---the
values flowing through the check may be from the open world, beyond
the reach of the analysis. Likewise, reference types must be references to $\dyn$ if
they are on the right of a check constraint with a mismatch between
the checked type on the left and the tag.


Therefore, the definition of a solution
$\sigma$ to $\Omega$ is as follows:

\begin{definition}\label{cp:def:solution}
  A mapping $\sigma$ is a \emph{solution} to $\Omega$ if:
  \begin{enumerate}
  \item For all $\flows{A_1}{A_2}\in\Omega$, $\vdash \sigma A_1 <: \sigma A_2$.
  \item For all $\depcon{A_1}{S}{A_2}\in\Omega$:
    \begin{enumerate}
    \item If $\lowlty{\sigma A_1}=S$, then $\vdash \sigma A_1=\sigma A_2$.
    \item Otherwise, for all $\alpha \in \mathit{parts}(A_2)$, $\vdash \sigma \alpha=\dyn$. 
    \end{enumerate}
  \item For all $A_1=A_2 \in \Omega$, $\vdash \sigma A_1=\sigma A_2$.
  \item For all $\defcon{\alpha}{T}\in\Omega$, $\vdash \sigma\alpha=T$.
  \item For all $\tagcon{\alpha}{S}\in\Omega$, $\lowlty{\sigma \alpha}=S$.
  \end{enumerate}
\end{definition}  

We can now prove that the solution algorithm
in Figure \ref{cp:fig:constsolve} generates valid solutions.

\begin{restatable}{theorem}{solvesound}\label{cp:thm:solvesound}
  If $\Omega\Downarrow\sigma$, then $\sigma$ is a solution to $\Omega$.
\end{restatable}

This proof is shown in Appendix \ref{cp:apx:proof:solvesound} and is by
induction on $\Omega\Downarrow\sigma$. It relies on a lemma showing
that a solution to any $\Omega$ is also a solution for $\Omega'$ if
$\Omega'\steps\Omega$.

\subsubsection{Correctness of check removal}

\begin{figure}
\hrulefill\vspace{1ex}
  \begin{minipage}{1.0\linewidth}
    \begin{minipage}{.49\linewidth}      
    \[
\begin{array}{rcl}
  \rho & ::= & \cdot \mid \rho,x=v
\end{array}
\]
\boxed{\Sigma\vdash \rho:\Gamma}
\begin{gather*}
  \inferrule{ }{\Sigma\vdash \cdot:\emptyset}\qquad
  \inferrule{\emptyset;\Sigma\vdash v:T \\ \Sigma\vdash \rho:\Gamma}{\Sigma\vdash \rho,x=v:\Gamma,x:T}
\end{gather*}

    \end{minipage}
    \begin{minipage}{.49\linewidth}
\boxed{|d|=e}
\[
\begin{array}{rcl}
  |x| & = & x\\
  |n| & = & n\\
  |\fntx{x}{X}{Y}{d}| & = & \fnx{x}{|d|}\\
  |d_1~d_2| & = & |d_1|~|d_2| \\
  |\refxi{d}{X}| & = & \refx{|d|} \\ 
  |\derefx{d}| & = & \derefxp{|d|}{} \\
  |\mutx{d_1}{d_2}| & = & \mutxpINV{|d_1|}{|d_2|}{} \\  
  |d_1~+~d_2| & = & |d_1|~+~|d_2| \\
  |\checkx{d}{S}| & = & \checkx{|d|}{S} \\
\end{array}
\]  
    \end{minipage}

  \end{minipage}
  \caption{Rules for environments and for erasing to $\langd$ to $\lange$.}
  \label{cp:fig:erase}
\end{figure}

To show that our our approach preserves the semantics of the program,
we prove that any check that \emph{would} be removed by the
translation process (as given in Figure \ref{cp:fig:dtoe}) cannot
fail if it remained in the program. Given some $\langd$ program $d$,
suppose that the constraint generation process from Figure
\ref{cp:fig:constrgen} gives $d$ the type $A$ (which may be a
variable) and generates the constraints $\Omega$, and assume that
$\Omega$ is solved by some $\sigma$. We can translate $d$ directly
into $\lange$ by simply erasing its type annotations, \emph{without}
using $\sigma$ to perform the syntax-directed optimization process from
Figure \ref{cp:fig:dtoe}. This erased program $|d|$, generated using
the erasure rules defined in Figure \ref{cp:fig:erase}, contains all
the checks originally present in $d$. Suppose that $|d|$ evaluates to
a value using the small-step dynamic semantics for $\lange$ defined in
Figure \ref{cp:fig:dynsem}, and is then checked against the type tag
of its solution, $\lowlty{\sigma A}$. We prove that this check will
always succeed at runtime. This is the same criterion used to actually
remove checks in Figure \ref{cp:fig:dtoe}, so by showing that the
checks that would be removed are redundant, we verify that they can be
removed.

\begin{figure}
  \begin{minipage}{1.0\linewidth}
    
\[
\begin{array}{rcl}
  \varsigma & ::= & \cfg{e}{\mu} \mid \mathtt{fail} \\
  v & ::= & a \mid n \mid \fnx{x}{e} \\
  \mu & ::= & \cdot \mid \mu[a:=v]\\
  E & ::= & \Box \mid E + e \mid v + E \mid (E~e) \mid (v~E) \mid \letx{x}{E}{e} \mid \refx{E} \mid \derefxp{E}{} \mid  \mutxpINV{E}{e}{} \mid  \mutxpINV{v}{E}{} \mid \checkx{E}{S}
\end{array}
\]
\boxed{\cfg{e}{\mu}\steps\varsigma}
\[
\begin{array}{lrcll}
  \textsc{ECheck} & \cfg{\checkx{v}{S}}{\mu} & \steps & \cfg{v}{\mu} & \text{if }\hastypenb{v}{S}\\
  \textsc{ECheckFail} & \cfg{\checkx{v}{S}}{\mu} & \steps & \mathtt{fail} & \text{if }\neg\hastypenb{v}{S}\\
  \textsc{EFail} & \cfg{\mathtt{fail}}{\mu} & \steps & \mathtt{fail} \\[1ex]
  \textsc{ERef} & \cfg{\refx{v}}{\mu} & \steps & \cfg{a}{\mu[a:=v]} & \text{where }a\text{ fresh}\\
  \textsc{EDeref} & \cfg{\derefxp{a}{}}{\mu} & \steps & \cfg{v}{\mu} & \text{where }\mu(a)=v\\
  \textsc{EUpdt} & \cfg{\mutxpINV{a}{v}{}}{\mu} & \steps & \cfg{0}{\mu[a:=v]} & \text{where }\mu(a)=v'\\[1ex]
  \textsc{EApp} & \cfg{((\fnx{x}{e})~v)}{\mu} & \steps & \cfg{e[x/v]}{\mu}\\
  \textsc{ELet} & \cfg{\letx{x}{v}{e}}{\mu} & \steps & \cfg{e[x/v]}{\mu}\\[1ex]
  \textsc{EAdd} & \cfg{n_1~+~n_2}{\mu} & \steps & \cfg{n'}{\mu} & \text{where }n_1+n_2=n'
\end{array}
\]
\begin{gather*}
\inferrule{\cfg{e}{\mu}\steps\cfg{e'}{\mu'}}
  {\cfg{E[e]}{\mu}\longmapsto\cfg{E[e']}{\mu'}}\qquad
\inferrule{\cfg{e}{\mu}\steps\mathtt{fail}}
  {\cfg{E[e]}{\mu}\longmapsto\mathtt{fail}}
\end{gather*}

\boxed{\hastypenb{v}{S}}
\begin{gather*}
\inferrule{ }{\hastypenb{v}{\dyn}}\qquad
\inferrule{ }{\hastypenb{n}{\intt}}\qquad
\inferrule{ }{\hastypenb{\fnx{x}{e}}{\to}}\qquad
\inferrule{ }{\hastypenb{a}{\refs}}
\end{gather*}
  \end{minipage}
  \caption{Dynamic semantics for $\lange$.}
  \label{cp:fig:dynsem}
\end{figure}

In this theorem, $d$ may be an open term---$\rho(|d|)$ substitutes
values from an environment $\rho$ (defined and given a typing judgment in Figure
\ref{cp:fig:erase}) into $|d|$.

\begin{restatable}{theorem}{correct}\label{cp:thm:correct}
  Suppose $\Gamma\vdash d:A;\Omega$ and $\sigma$ is a solution to
  $\Omega$ and $\Sigma\vdash \rho:\sigma\Gamma$ and $\Sigma\vdash\mu$ and
  $\cfg{\rho(|d|)}{\mu}\steps^*\cfg{v}{\mu'}$. If
  $\lowlty{\sigma A} \preceq S$, then
  $\cfg{\checkx{v}{S}}{\mu'}\not\steps\mathtt{fail}$.
\end{restatable}
This proof is shown in Appendix \ref{cp:apx:proof:correct}. It relies on
a preservation lemma and a lemma showing that, from the canonical
forms lemma on $\lange$, any value $v$ with type $\sigma A$ will
necessarily correspond to that type.

\section{Performance of Optimized Transient Gradual Typing}\label{cp:sec:optperf}

In this section we apply the above approach to optimize Reticulated
Python, and summarize Reticulated's performance characteristics when
running on CPython. This required expanding the type system and
constraint generation of Section~\ref{cp:sec:infer} to handle Python
features such as objects and classes, data structures such as lists
and dictionaries, bound and unbound methods, and variadic
functions. In addition, the constraint generation includes polymorphic
functions and intersection types, although they only occur in the
pre-loaded type definitions for Python libraries and builtin
functions.\footnote{Link to repository removed for the purpose of
  double-blind review.}

Figure \ref{cp:fig:cpythonperf}, as previously discussed, shows the
execution time and overheads for optimized configurations from the
typing lattice of each benchmark when executed using
CPython. Optimized configurations are shown as blue
triangles. Performance is dramatically improved compared to the
unoptimized (red circle) configurations. In several benchmarks, the
overhead is entirely eliminated because the optimization is able to
delete nearly every check in every configuration from the typing
lattice. For these results, we make the closed world assumption
for the benchmarks.

The \texttt{suffixtree} benchmark, which \citet{takikawa:2016gtdead}
tested in Typed Racket and found overheads of up to 105$\xtimes$, also
performs worse than other benchmarks in Reticulated Python after
optimization: although it has negligible overhead in configurations
with high type weight, some configurations with intermediate type
weight still have an overhead of over 2$\xtimes$. This is because
\texttt{suffixtree}, unlike the other benchmarks, cannot be fully
statically typed using Reticulated Python's type system: the version
with the highest type weight, used to generate the other samples,
still includes a function, \li{node_follow}, whose return type cannot
be given a static type by the Reticulated Python type system or the
inference process, because it can return either functions or booleans.
%
%
%
Because \li{node_follow}'s return values flow into statically typed
code, checks will be needed.  In most configurations, checks occur
when the result of \li{node_follow} passes into a statically typed
function, and checks can be removed from the rest of the
program. However, some configurations allow the dynamic values to flow
further into the program before encountering a check; the dynamicity
has ``infected'' more of the program and degraded performance.

\begin{figure}
  \centering
  \begin{minipage}{1.0\linewidth}
    \centering
    \begin{tabular}{r|rrr|rrr}
         & \multicolumn{3}{l|}{Unopt. overheads} & \multicolumn{3}{l}{Opt. overheads}  \\
        Benchmark  & Mean & Max & Static & Mean & Max & Static \\\hhline{|-|------|} 
    \texttt{pystone} & 2.39$\xtimes$ & \color{tablered}3.82$\xtimes$ & \color{tablered}3.72$\xtimes$ & \cellcolor{tableblue}1.01$\xtimes$ & \cellcolor{tableblue}1.03$\xtimes$ & \cellcolor{tableblue}1.00$\xtimes$ \\
    \texttt{chaos} & 1.84$\xtimes$ & \color{tablered}3.22$\xtimes$ & \color{tablered}3.15$\xtimes$ & \cellcolor{tableblue}1.10$\xtimes$ & 1.46$\xtimes$ & \cellcolor{tableblue}0.98$\xtimes$ \\
    \texttt{snake} & 2.31$\xtimes$ & \color{tablered}3.79$\xtimes$ & \color{tablered}3.70$\xtimes$ & \cellcolor{tableblue}1.04$\xtimes$ & 1.49$\xtimes$ & \cellcolor{tableblue}0.98$\xtimes$ \\
    \texttt{go} & 2.32$\xtimes$ & \color{tablered}4.87$\xtimes$ & \color{tablered}4.56$\xtimes$ & \cellcolor{tableblue}1.02$\xtimes$ & \cellcolor{tableblue}1.14$\xtimes$ & \cellcolor{tableblue}1.02$\xtimes$ \\
    \texttt{meteor\_contest} & 1.82$\xtimes$ & \color{tablered}3.20$\xtimes$ & \color{tablered}3.05$\xtimes$ & \cellcolor{tableblue}1.00$\xtimes$ & \cellcolor{tableblue}1.10$\xtimes$ & \cellcolor{tableblue}1.00$\xtimes$ \\
    \texttt{suffixtree} & 2.49$\xtimes$ & \color{tablered}4.48$\xtimes$ & \color{tablered}4.34$\xtimes$ & 1.27$\xtimes$ & 2.38$\xtimes$ & \cellcolor{tableblue}0.98$\xtimes$ \\
    \texttt{float} & 2.04$\xtimes$ & \color{tablered}3.53$\xtimes$ & \color{tablered}3.53$\xtimes$ & \cellcolor{tableblue}1.01$\xtimes$ & \cellcolor{tableblue}1.12$\xtimes$ & \cellcolor{tableblue}1.00$\xtimes$ \\
    \texttt{nbody} & 2.70$\xtimes$ & \color{tablered}5.95$\xtimes$ & \color{tablered}5.19$\xtimes$ & \cellcolor{tableblue}0.98$\xtimes$ & 1.27$\xtimes$ & \cellcolor{tableblue}0.95$\xtimes$ \\
    \texttt{sieve} & 1.52$\xtimes$ & 2.17$\xtimes$ & 2.09$\xtimes$ & \cellcolor{tableblue}1.01$\xtimes$ & \cellcolor{tableblue}1.06$\xtimes$ & \cellcolor{tableblue}1.01$\xtimes$ \\
    \texttt{spectral\_norm} & 2.19$\xtimes$ & \color{tablered}3.33$\xtimes$ & 2.98$\xtimes$ & \cellcolor{tableblue}1.00$\xtimes$ & \cellcolor{tableblue}1.20$\xtimes$ & \cellcolor{tableblue}0.99$\xtimes$ \\
    \hline Average & 2.21$\xtimes$ & \color{tablered}5.95$\xtimes$ & \color{tablered}3.63$\xtimes$ & \cellcolor{tableblue}1.06$\xtimes$ & 2.38$\xtimes$ & \cellcolor{tableblue}0.99$\xtimes$ \\
\end{tabular}
  \end{minipage}
  \caption{Performance details for Reticulated Python benchmarks under CPython. {\color{tablered}Red text} indicates \emph{worse than} 3-deliverability and \colorbox{tableblue}{blue highlighting} indicates 1.25-deliverability.}
  \label{cp:fig:perfdetailscpy}
\end{figure}

Figure \ref{cp:fig:perfdetailscpy} summarizes the performance results
for both optimized and unoptimized Reticulated Python under
CPython. This table shows the mean overhead, maximum overhead, and
overhead for the fully typed (or nearest to fully typed)
configurations for each benchmark with the original unoptimized
approach and with our optimization. Without optimization, the average
of all configurations from each benchmark meets the cutoff of
3-deliverability as suggested by \citet{takikawa:2016gtdead}, meaning
that their overheads were 3$\xtimes$ or less, but most static cases in
the benchmarks are not 3-deliverable. Our optimization dramatically
improves the results: not only are all configurations 3-deliverable,
but all fully typed configurations pass the stricter cutoff of
\emph{1.25-deliverability}, with slowdowns of 25\% or less over the
original untyped program. In all, we found that with the optimization,
the average overhead across the typing lattices of all benchmarks was
only 6\%.

\section{Performance on PyPy, a Tracing JIT}\label{cp:sec:pypyperf}

The analysis described in Section \ref{cp:sec:infer} and evaluated in
Section \ref{cp:sec:optperf} removed checks when they can be statically
guaranteed to never fail, which suggests that a dynamic analysis could
accomplish the same task.  We examined this question using a tracing
JIT implementation of Python 3, PyPy~\cite{Bolz:2009kx}.

\begin{figure*}
  \centering
   \begin{tabular}[t]{ll}
     Benchmark: \texttt{pystone}  & Benchmark: \texttt{chaos}  \\
     \small (206 SLoC, 532 configurations) &  \small (184 SLoC, 982 configurations) \\
     \insgraphw{pystone_pypy.py_pypy3} & \insgraphw{bm_chaos.py_pypy3} \\
     Benchmark: \texttt{snake}& Benchmark: \texttt{go}    \\
      \small (112 SLoC, 662 configurations) &  \small (394 SLoC, 1001 configurations) \\
     \insgraphw{snake.py_pypy3} & \insgraphw{bm_go.py_pypy3} \\
     Benchmark: \texttt{meteor\_contest} & Benchmark: \texttt{suffixtree} \\
     \small (106 SLoC, 972 configurations) &  \small (338 SLoC, 1001 configurations) \\
     \insgraphw{bm_meteor_contest.noaliases.py_pypy3} & \insgraphw{suffixtree.py_pypy3} \\
     Benchmark: \texttt{float}  & Benchmark: \texttt{nbody} \\
      \small (48 SLoC, 162 configurations) & \small (74 SLoC, 892 configurations) \\
     \insgraphw{bm_float.py_pypy3} & \insgraphw{bm_nbody.noaliases.py_pypy3} \\
     Benchmark: \texttt{sieve} & Benchmark: \texttt{spectral\_norm}  \\
      \small (50 SLoC, 282 configurations) &  \small (44 SLoC, 312 configurations) \\
     \insgraphw{sieve.py_pypy3}  & \insgraphw{bm_spectral_norm.py_pypy3} 
   \end{tabular}
  \Put(-245,-1100){\inslegend}
  \caption{Typing lattices for Reticulated Python benchmarks under PyPy.}
  \label{cp:fig:pypyperf}
\end{figure*}

\begin{figure}
  \centering
  \begin{minipage}{1.0\linewidth}
    \centering
\begin{tabular}{r|rrr|rrr}
        & \multicolumn{3}{l|}{Unopt. overheads} & \multicolumn{3}{l}{Opt. overheads}  \\
        Benchmark  & Mean & Max & Static & Mean & Max & Static \\\hhline{|-|------} 
    \texttt{pystone} & \cellcolor{tableblue}1.24$\xtimes$ & 1.65$\xtimes$ & 1.62$\xtimes$ & \cellcolor{tableblue}1.01$\xtimes$ & \cellcolor{tableblue}1.04$\xtimes$ & \cellcolor{tableblue}1.00$\xtimes$ \\
    \texttt{chaos} & \cellcolor{tableblue}1.18$\xtimes$ & 2.61$\xtimes$ & 2.31$\xtimes$ & \cellcolor{tableblue}1.02$\xtimes$ & \cellcolor{tableblue}1.12$\xtimes$ & \cellcolor{tableblue}1.00$\xtimes$ \\
    \texttt{snake} & \cellcolor{tableblue}1.25$\xtimes$ & 2.28$\xtimes$ & 2.24$\xtimes$ & \cellcolor{tableblue}1.02$\xtimes$ & \cellcolor{tableblue}1.23$\xtimes$ & \cellcolor{tableblue}1.00$\xtimes$ \\
    \texttt{go} & \cellcolor{tableblue}0.61$\xtimes$ & 1.29$\xtimes$ & \cellcolor{tableblue}0.57$\xtimes$ & \cellcolor{tableblue}0.92$\xtimes$ & \cellcolor{tableblue}1.19$\xtimes$ & \cellcolor{tableblue}1.16$\xtimes$ \\
    \texttt{meteor\_contest} & \cellcolor{tableblue}1.02$\xtimes$ & \cellcolor{tableblue}1.12$\xtimes$ & \cellcolor{tableblue}1.01$\xtimes$ & \cellcolor{tableblue}1.01$\xtimes$ & \cellcolor{tableblue}1.05$\xtimes$ & \cellcolor{tableblue}1.05$\xtimes$ \\
    \texttt{suffixtree} & \cellcolor{tableblue}1.06$\xtimes$ & 1.31$\xtimes$ & \cellcolor{tableblue}1.06$\xtimes$ & \cellcolor{tableblue}1.00$\xtimes$ & \cellcolor{tableblue}1.16$\xtimes$ & \cellcolor{tableblue}0.96$\xtimes$ \\
    \texttt{float} & \cellcolor{tableblue}0.98$\xtimes$ & \cellcolor{tableblue}1.05$\xtimes$ & \cellcolor{tableblue}0.98$\xtimes$ & \cellcolor{tableblue}0.97$\xtimes$ & \cellcolor{tableblue}1.00$\xtimes$ & \cellcolor{tableblue}0.95$\xtimes$ \\
    \texttt{nbody} & \cellcolor{tableblue}1.00$\xtimes$ & \cellcolor{tableblue}1.05$\xtimes$ & \cellcolor{tableblue}1.01$\xtimes$ & \cellcolor{tableblue}1.00$\xtimes$ & \cellcolor{tableblue}1.02$\xtimes$ & \cellcolor{tableblue}0.99$\xtimes$ \\
    \texttt{sieve} & \cellcolor{tableblue}1.05$\xtimes$ & \cellcolor{tableblue}1.18$\xtimes$ & \cellcolor{tableblue}1.06$\xtimes$ & \cellcolor{tableblue}1.05$\xtimes$ & \cellcolor{tableblue}1.11$\xtimes$ & \cellcolor{tableblue}1.04$\xtimes$ \\
    \texttt{spectral\_norm} & \cellcolor{tableblue}1.08$\xtimes$ & \cellcolor{tableblue}1.15$\xtimes$ & \cellcolor{tableblue}1.12$\xtimes$ & \cellcolor{tableblue}1.01$\xtimes$ & \cellcolor{tableblue}1.03$\xtimes$ & \cellcolor{tableblue}1.01$\xtimes$ \\
    \hline Average & \cellcolor{tableblue}1.03$\xtimes$ & 2.61$\xtimes$ & 1.30$\xtimes$ & \cellcolor{tableblue}1.00$\xtimes$ & \cellcolor{tableblue}1.23$\xtimes$ & \cellcolor{tableblue}1.02$\xtimes$ \\
\end{tabular}

  \end{minipage}
  \caption{Performance details for Reticulated Python benchmarks under PyPy.}
  \label{cp:fig:perfdetailspypy}
\end{figure}

Figure \ref{cp:fig:pypyperf} shows the typing lattices for our benchmarks
when run on PyPy, both with the standard and optimized
transient approaches. Figure \ref{cp:fig:perfdetailspypy}
summarizes these results. Without static check removal, PyPy's
performance varies but is almost always better than CPython's relative
to the baseline performance of the untyped benchmarks. PyPy performs
on average 4.3$\xtimes$ better than CPython on its own when comparing
untyped benchmarks, so when examining the performance cost of
transient gradual typing, we always compare to the untyped execution
time for the implementation. In some benchmarks
(\texttt{float}, \texttt{nbody}, \texttt{sieve},
\texttt{spectral\_norm}) configurations across the typing lattice
perform almost as well as the untyped version of the program without
any static optimizations. In other cases (\texttt{pystone},
\texttt{chaos}, \texttt{snake}) performance still degrades as types
are added, although to a lesser degree than with CPython. In one case,
\texttt{go}, transient gradual typing significantly \emph{improves}
performance on average, suggesting that transient checks may cause the
JIT compiler to activate earlier or with better traces than
it would otherwise. All configurations were 3-deliverable and
most were 1.25-deliverable even without optimization, with an overall
mean slowdown of 3\%, though some configurations have significant
overhead (up to 2.61$\xtimes$).

The best results were obtained by combining our optimization and the
PyPy JIT. When run with optimized transient gradual typing, the average overhead was
0\% over the baseline, and every configuration fell within the
1.25x-deliverable range. The result is an approach that appears
practical for real-world applications. In future work, we
will further examine the interactions between JITs and the transient
approach, for example, to better understand the speedups seen in \texttt{go}.

Reticulated Python is not the only approach to gradual typing that can
use a tracing JIT---\citet{bauman:2017pycket} showed that Pycket, a
language based on Typed Racket but implemented in RPython (a language
for automatically generating tracing JITs that PyPy itself is
implemented in) \cite{Bolz:2009kx}, performed better than standard
Typed Racket but still displayed worst-case overheads of up to
10.5$\xtimes$. Pycket uses the guarded strategy, and while the tracing
JIT was successful at reducing the overhead of that approach, our
results suggest that the transient enforcement strategy is especially
suited to use with tracing JITs.

\section{Related Work}\label{cp:sec:relwork}

\paragraph{Gradual type systems.} In our work, checks are removed when
we infer that the types they verify can be trusted. This is suggestive
of the \emph{strict confined gradual typing} of
\citet{Allende:2014aa}, which allows programmers to restrict types
such that their inhabitants must never have passed through dynamic
code (indicated by ${\downarrow} T$) or will never pass through
dynamic code in the future (indicated by ${\uparrow} T$). If a term
has type ${\downarrow} T$, the type system verifies that value
originated from an introduction form for values of that type. We
suspect that this information could be used to remove transient checks
or perform other optimizations.



Our approach is also related to \emph{concrete types}
\cite{Wrigstad:2010fk,richards:2015strongscript}, which are inhabited
only by non-proxied values (using the guarded cast strategy). Concrete types have limited interoperability with dynamic
types and with \emph{like types}, which are the types of values which
may or may not be proxied. While the formulation of concrete vs.\ like
types given by \citet{Wrigstad:2010fk} is appropriate for the guarded
semantics, splitting types into those which can be statically relied
on and those which need runtime verification is similar to our
approach.

\paragraph{Performance analysis for gradual typing.}

\citet{takikawa:2016gtdead} performed the most detailed analysis of
the effect of gradual typing on efficiency to date, in the context of
Typed Racket, a mature gradually typed language. They introduced the
typing lattice to evaluate gradually typed languages and found that
the overhead of Typed Racket in some configurations was high enough to
threaten the viability of gradual typing; our paper finds a more
optimistic result in a different context and with a very different
approach to gradual typing. \citet{greenman:2017retic} analyzed
Reticulated Python programs across the typing lattice using a
sampling-based methodology similar to ours, and reported performance
results similar to those that we report in Section
\ref{cp:sec:nooptperf}.

\citet{muehlboeck:2017nom} designed a gradually typed language with
nominal object types and without structural types or functions and
examined its performance across the lattices of several benchmarks,
finding negligible overhead.

\citet{Rastogi:2012fk} analyzed the performance of ActionScript, a
gradually typed language, by taking fully-typed benchmarks and
removing all type annotations except for those on interfaces. They
found that this resulted in significant overhead compared to the
original fully-typed versions. They then used a inference-based
optimization (discussed below) to reconstruct type annotations; this
fully recovered performance in most benchmarks.

\paragraph{Type inference.} 

\citet{aiken:1993inclusion} generalize equational constraints used in
standard Hindley-Milner type inference
\cite{Hindley:1969ri,Milner:1978kh} to subset constraints
$T \subseteq S$. In this work types are interpreted as subsets of a
semantic domain of values and the subset relation is equivalent to
subtyping. Their type system includes union and intersection types and
generates systems of constraints that may be simplified by rewriting,
similar to the rules for constraint simplification shown in Figure
\ref{cp:fig:constsimpl}. \citet{Aiken:1995ly} specialize this approach
to determine where coercions are needed in dynamically typed programs
with tagging and untagging in order to optimize Scheme programs. This
approach is similar to ours, except that while our type system is much
simpler and does not include untagged types and values, transient
checks must generate check constraints. Check constraints are similar
to the \emph{conditional constraints} of \citet{Pottier:2000qi}, but
rather than allowing arbitrary implications, check constraints reason
exclusively about type tags.

Soft type systems \cite{Cartwright:1991ng, Aiken:1994fk} use a similar
approach to integrating static and dynamic typing with type inference,
but with a different goal: to allow programs written in dynamically
typed languages to leverage the benefits of static typing. Soft type
systems use type inference to determine where runtime type checks must
be inserted into dynamically typed programs so they are well-typed in
some static type system. \citet{Cartwright:1991ng} reconstruct types
for their language and determine where checks are needed using
circular unification on equality constraints \cite{weis:1987caml} over
an encoding of the supertypes of each type in the program which
factors out subtyping \cite{remy:1989recvar}. \citet{Aiken:1994fk}
adapt subset constraint generation to soft typing with similar
results. The runtime checks, or \emph{narrowers}, used in soft typing
are similar to transient checks: a narrower checks that a value
corresponds to a specific type constructor. The value is returned
unmodified if so, and an error is raised if not. Narrowers serve a
different purpose than checks, however: narrowers let programs written
in dynamically typed languages pass static typechecking for
optimization, while checks enforce the programmer's claims about the
types of terms.

The constraint system used by Flow's type inference system, described
for the \textsc{Flowcore} calculus by \citet{chaudhuri:2017flow},
reasons about tags and tagchecks in order to find sound typings for
JavaScript programs, though it relies on and trusts type annotations
at module boundaries. Likewise \citet{guha:2011flow} relate type tags
and types similarly to our approach and their tagchecks are equivalent
to transient checks, but their goal is to insert the tagchecks needed
to type function bodies with respect to their (trusted) annotations,
rather than to detect violations of those annotations by untrusted
callers.

\citet{Rastogi:2012fk} present an approach to optimizing gradually
typed programs by inferring more precise types for program
annotations, while preserving the program's semantics. 
Our approach to ensuring interoperability is based on theirs: visible
variables in the overall type of a program and their co- or
contravariant positions encodes escape analysis, and solutions that
can soundly interoperate with arbitrary code can be generated by
adding constraints on these variables. Our constraints, however, are
based on subtyping rather than on consistency, because our constraints
arise from checks and from elimination forms rather than from casts,
which are appropriate for guarded gradual typing rather than
transient.

\section{Conclusions}\label{cp:sec:conclusion}

Gradual typing allows programmers to combine static and dynamic typing
in the same language, but allowing interaction between static and
dynamic code while ensuring soundness incurs a runtime
cost. \citet{takikawa:2016gtdead} found that this cost can be serious
obstacle to the practical use of Typed Racket, a popular gradually
typed language which uses the guarded approach to gradual typing. In
this paper, we perform a detailed performance analysis of the
transient approach in Reticulated Python by analyzing configurations from the
typing lattices of \benchcount benchmarks. We show that, in
combination with the standard Python interpreter,
performance under the transient design degrades as programs evolve from dynamic to static types,
to a maximum of 6$\xtimes$ slowdown compared to equivalent untyped
Python programs. To reduce this overhead, we use an static type inference
algorithm based on subtyping and check constraints to optimize
programs after transient checks are inserted. This allows many
redundant checks to be removed. We evaluated the performance of this
approach with an implementation in Reticulated Python and found that
performance across the typing lattices of our benchmarks improved
to nearly the efficiency of the original programs---a very promising
result for the practicality of gradual typing. Finally, we re-analyzed
our Reticulated Python benchmarks using PyPy, a tracing JIT, as a
backend, and found that it produced good performance even without type
inference, and that it displayed a no overhead when used in
combination with our static  optimization.



{
\bibliographystyle{ACM-Reference-Format}
\bibliography{bib.bib}
}

 \appendix
 \newpage

\section{Appendix: Semantics}\label{cp:apx:calculus}
Figure \ref{cp:fig:apx:stysys} shows the static type system for
$\langs$. Figure \ref{cp:fig:apx:dtysys} shows the shallow static type
system for $\langd$. Figure \ref{cp:fig:apx:approx} relates surface
types $U$ with constraint types $A$.

Figure \ref{cp:fig:apx:lange} shows the syntax for $\lange$, the final
target language of translation. Figure \ref{cp:fig:apx:langety} defines the
Curry-style type system for $\lange$.

Figure \ref{cp:fig:apx:dynsem} shows the dynamic semantics of $\lange$,
while utility relations are shown in Figure
\ref{cp:fig:apx:moresem}. Figure \ref{cp:fig:apx:ht} relates weaker heap
types with stronger ones.

Figure \ref{cp:fig:apx:erase} shows rules for translating $\langd$ to
$\lange$ directly by removing type annotations (without removing
checks). It also shows definition and typing
rules for value environments.

\begin{figure}
  \begin{minipage}{1.0\linewidth}    
\boxed{\Gamma\vdash s:U}
\begin{gather*}
\inferrule[SAbs]{\Gamma,x{:}U_1\vdash s:U_2' \\ U_2' \sim U_2}
  {\Gamma\vdash \fntx{x}{U_1}{U_2}{s}:\funt{U_1}{U_2}} \qquad
\inferrule[SApp]{\Gamma\vdash s_1:U \\\\ \matchrel{U}{\funt{U_1}{U_2}} \\\\
           \Gamma\vdash s_2:U_1' \\ U_1' \sim U_1}
  {\Gamma\vdash s_1~s_2:U_2}\\[1ex]
\inferrule[SRef]{\Gamma\vdash s:U_2 \\ U_2\sim U_1}{\Gamma\vdash \refxi{s}{U_1}:\reft{U_1}}\qquad
\inferrule[SDeref]{\Gamma\vdash s:U \\ \matchrel{U}{\reft{U'}}}
  {\Gamma\vdash\derefx{s}:U'}\\[1ex]
\inferrule[SUpdt]{\Gamma\vdash s_1:U \\ \matchrel{U}{\reft{U'}} \\
           \Gamma\vdash s_2:U'' \\ U'' \sim U'}
  {\Gamma\vdash\mutx{s_1}{s_2}:\intt}\\[1ex]
\inferrule[SAdd]{\Gamma\vdash s_1:U_1 \\ U_1\sim \intt \\\\
           \Gamma\vdash s_2:U_2 \\ U_2\sim \intt}
  {\Gamma\vdash s_1~+~s_2:\intt}\qquad
\inferrule[SVar]{\Gamma(x)=U}{\Gamma\vdash x:U} \qquad 
\inferrule[SInt]{ }{\Gamma\vdash n:\intt}
\end{gather*}
  \end{minipage}
  \caption{Type system for $\langs$.}
  \label{cp:fig:apx:stysys}
\end{figure}

\begin{figure}
  \begin{minipage}{1.0\linewidth}    
\boxed{\Gamma\vdash d:S}
\begin{gather*}
\inferrule[PAbs]{\Gamma,x{:}\dyn\vdash d:\dyn}
  {\Gamma\vdash \fntx{x}{X}{Y}{d}:\funs} \qquad
\inferrule[PApp]{\Gamma\vdash d_1:\funs \\
           \Gamma\vdash d_2:\dyn}
  {\Gamma\vdash d_1~d_2:\dyn}\\[1ex]
\inferrule[PRef]{\Gamma\vdash d:\dyn}{\Gamma\vdash \refxi{d}{X}:\refs}\qquad
\inferrule[PDeref]{\Gamma\vdash d:\refs}
  {\Gamma\vdash\derefx{d}:\dyn}\\[1ex]
\inferrule[PUpdt]{\Gamma\vdash d_1:\refs \\ 
           \Gamma\vdash d_2:\dyn}
  {\Gamma\vdash\mutx{d_1}{d_2}:\intt}\\[1ex]
\inferrule[PAdd]{\Gamma\vdash d_1:\intt \\
           \Gamma\vdash d_2:\intt}
  {\Gamma\vdash d_1~+~d_2:\intt}\qquad
\inferrule[PVar]{\Gamma(x)=S}{\Gamma\vdash x:S} \qquad 
\inferrule[PInt]{ }{\Gamma\vdash n:\intt}
\end{gather*}
  \end{minipage}
  \caption{Simple type system for $\langd$.}
  \label{cp:fig:apx:dtysys}
\end{figure}

\begin{figure}
  \begin{minipage}{1\linewidth}
    \boxed{\R{U}{A}}
\begin{gather*}
  \inferrule{ }{\R{U}{V}}\qquad
  \inferrule{ }{\R{U_1\to U_2}{\funt{V_1}{V_2}}}\\[1ex]
  \inferrule{ }{\R{\reft{U}}{\reft{V}}}\qquad
  \inferrule{ }{\R{\intt}{\intt}}
\end{gather*}
  \end{minipage}
\caption{Relating $U$ and $A$.}
\label{cp:fig:apx:approx}
\end{figure}

\begin{figure}
\hrulefill\vspace{1ex}

\[
\begin{array}{rcl}
  a & \in & \text{addresses}\\
  e & ::= & a \mid x \mid n \mid e +^w e \mid \fnx{x}{e} \mid (e~e)^w \mid  \letx{x}{e}{e} \mid \refx{e} \mid \derefxp{e}{w} \mid \mutxpINV{e}{e}{w} \mid \checkx{e}{S} \mid  \mathtt{fail}\\
  w & ::= & \tyorig \mid\unorig\\
  \Sigma & ::= & \cdot \mid \Sigma,a{:}T 
\end{array}
\]
  \caption{Syntax for $\lange$.}
  \label{cp:fig:apx:lange}
\end{figure}

\begin{figure*}
  \begin{minipage}{1.0\linewidth}
\boxed{\GS e : T}
\begin{gather*}
\inferrule[TSubsump]{\GS e:T_1 \\ \vdash T_1 <: T_2}
 {\GS e:T_2}\qquad
\inferrule[TAbs]{\Gamma,x{:}T_1;\Sigma\vdash e : T_2}
 {\Gamma\vdash \fnx{x}{e} : \funt{T_1}{T_2}}\qquad
\inferrule[TLet]{\GS e_1 : T_1 \\ \Gamma,x{:}T_1;\Sigma\vdash e_2 : T_2}
 {\Gamma\vdash\letx{x}{e_1}{e_2} : T_2}\\[1ex]
\inferrule[TApp]{\GS e_1 : \funt{T_1}{T_2} \\
   \Gamma\vdash e_2 : T_1}
 {\Gamma\vdash (e_1~e_2)^\tyorig : T_2}\qquad
\inferrule[TAppOW]{\GS e_1 : \dyn \\
   \Gamma\vdash e_2 : \dyn}
 {\Gamma\vdash (e_1~e_2)^\unorig : \dyn}\\[1ex]
\inferrule[TRef]{\GS e : T}{\GS\refx{e} :\reft{T}}\qquad
\inferrule[TDeref]{\GS e : \reft{T}}
 {\GS\derefxp{e}{\tyorig} : T}\\[1ex]
\inferrule[TDerefOW]{\GS e : \dyn}
 {\GS\derefxp{e}{\unorig} : \dyn}\\[1ex]
\inferrule[TUpdt]{\GS e_1 : \reft{T} \\
   \GS e_2 : T}
 {\GS\mutxpINV{e_1}{e_2}{\tyorig} :\intt}\qquad
\inferrule[TUpdtOW]{\GS e_1 : \dyn \\
   \GS e_2 : \dyn}
 {\GS\mutxpINV{e_1}{e_2}{\unorig} :\intt}\qquad
\inferrule[TVar]{\Gamma(x)=T}{\GS x : T} \\[1ex] 
\inferrule[TAddr]{\Sigma(a)=T}{\GS a : \reft{T}} \qquad 
\inferrule[TInt]{ }{\GS n :\intt}\\[1ex]
\inferrule[TAdd]{\GS e_1 : \intt \\ 
   \GS e_2 : \intt }
 {\GS e_1~+^\tyorig~e_2 :\intt}\qquad
\inferrule[TAddOW]{\GS e_1 : \dyn \\ 
   \GS e_2 : \dyn }
 {\GS e_1~+^\unorig~e_2 :\intt}\\[1ex]
\inferrule[TCheck]{\GS e : \dyn}
 {\GS \checkx{e}{S} : \uplty{S}}\qquad
\inferrule[TFail]{ }
 {\GS \mathtt{fail} : T}\qquad
\inferrule[TCheckFail]{\GS e : T \\ T \neq \dyn \\ \lowlty{T} \not\preceq S}
 {\GS \checkx{e}{S} : T'}\\[1ex]
\inferrule[TCheckRedundant]{\GS e : T \\ \lowlty{T} \preceq S}
 {\GS \checkx{e}{S} : T}
\end{gather*}
  \end{minipage}
  \caption{Type system for $\lange$.}
  \label{cp:fig:apx:langety}
\end{figure*}

\begin{figure}
  \begin{minipage}{1.0\linewidth}
    
\[
\begin{array}{rcl}
  \varsigma & ::= & \cfg{e}{\mu} \mid \mathtt{fail} \\
  v & ::= & a \mid n \mid \fnx{x}{e} \\
  \mu & ::= & \cdot \mid \mu[a:=v]\\
  E & ::= & \Box \mid E +^w e \mid v +^w E \mid (E~e)^w \mid (v~E)^w \mid \letx{x}{E}{e} \mid \refx{E} \mid \derefxp{E}{w} \mid  \mutxpINV{E}{e}{w} \mid \\&& \mutxpINV{v}{E}{w} \mid \checkx{E}{S}
\end{array}
\]
\boxed{\cfg{e}{\mu}\steps\varsigma}
\[
\begin{array}{lrcll}
  \textsc{ECheck} & \cfg{\checkx{v}{S}}{\mu} & \steps & \cfg{v}{\mu} & \text{if }\hastypenb{v}{S}\\
  \textsc{ECheckFail} & \cfg{\checkx{v}{S}}{\mu} & \steps & \mathtt{fail} & \text{if }\neg\hastypenb{v}{S}\\
  \textsc{EFail} & \cfg{\mathtt{fail}}{\mu} & \steps & \mathtt{fail} \\[1ex]
  \textsc{ERef} & \cfg{\refx{v}}{\mu} & \steps & \cfg{a}{\mu[a:=v]} & \text{where }a\text{ fresh}\\
  \textsc{EDeref} & \cfg{\derefxp{a}{w}}{\mu} & \steps & \cfg{v}{\mu} & \text{where }\mu(a)=v\\
  \textsc{EUpdt} & \cfg{\mutxpINV{a}{v}{w}}{\mu} & \steps & \cfg{0}{\mu[a:=v]} & \text{where }\mu(a)=v'\\[1ex]
  \textsc{EApp} & \cfg{((\fnx{x}{e})~v)^w}{\mu} & \steps & \cfg{e[x/v]}{\mu}\\
  \textsc{ELet} & \cfg{\letx{x}{v}{e}}{\mu} & \steps & \cfg{e[x/v]}{\mu}\\[1ex]
  \textsc{EAdd} & \cfg{n_1~+^w~n_2}{\mu} & \steps & \cfg{n'}{\mu} & \text{where }n_1+n_2=n'
\end{array}
\]
\begin{gather*}
\inferrule{\cfg{e}{\mu}\steps\cfg{e'}{\mu'}}
  {\cfg{E[e]}{\mu}\longmapsto\cfg{E[e']}{\mu'}}\qquad
\inferrule{\cfg{e}{\mu}\steps\mathtt{fail}}
  {\cfg{E[e]}{\mu}\longmapsto\mathtt{fail}}
\end{gather*}

\boxed{\hastypenb{v}{S}}
\begin{gather*}
\inferrule{ }{\hastypenb{v}{\dyn}}\qquad
\inferrule{ }{\hastypenb{n}{\intt}}\qquad
\inferrule{ }{\hastypenb{\fnx{x}{e}}{\to}}\qquad
\inferrule{ }{\hastypenb{a}{\refs}}
\end{gather*}
  \end{minipage}
  \caption{Dynamic semantics for $\lange$.}
  \label{cp:fig:apx:dynsem}
\end{figure}

\begin{figure*}
  \begin{minipage}{1.0\linewidth}
    
\boxed{\stucknb{e}{\mu}{w}}
\begin{gather*}
\inferrule{ }{\stucknb{(n~v)^w}{\mu}{w}}\qquad
\inferrule{ }{\stucknb{(a~v)^w}{\mu}{w}}\\[1ex]
\inferrule{ }{\stucknb{a +^w v}{\mu}{w}}\qquad
\inferrule{ }{\stucknb{(\fnx{x}{e}) +^w v}{\mu}{w}}\\[1ex]
\inferrule{ }{\stucknb{n +^w a}{\mu}{w}}\qquad
\inferrule{ }{\stucknb{n +^w (\fnx{x}{e})}{\mu}{w}}\qquad
\inferrule{a \not\in\dom{\mu}}{\stucknb{\derefxp{a}{w}}{\mu}{w}}\qquad
\inferrule{ }{\stucknb{\derefxp{n}{w}}{\mu}{w}}\\[1ex]
\inferrule{ }{\stucknb{\derefxp{(\fnx{x}{e})}{w}}{\mu}{w}}\qquad
\inferrule{a\not\in\dom{\mu}}{\stucknb{\mutxpINV{a}{v}{w}}{\mu}{w}}\qquad
\inferrule{ }{\stucknb{\mutxpINV{n}{v}{w}}{\mu}{w}}\\[1ex]
\inferrule{ }{\stucknb{\mutxpINV{(\fnx{x}{e})}{v}{w}}{\mu}{w}}\qquad
\inferrule{\stucknb{e}{\mu}{w}}{\stucknb{E[e]}{\mu}{w}}
\end{gather*}

\boxed{\Sigma\vdash\mu}
\begin{gather*}
\inferrule{\mathit{dom}(\Sigma)=\mathit{dom}(\mu) \\ \forall a \in \dom{\Sigma},~\emptyset{;}\Sigma\vdash\mu(a):\Sigma(a)}{\Sigma\vdash\mu}
\end{gather*}

  \end{minipage}
  \caption{Additional semantics for $\lange$.}
  \label{cp:fig:apx:moresem}
\end{figure*}

\begin{figure}
  \begin{minipage}{1.0\linewidth}
\boxed{\Sigma\ledyn\Sigma}
\begin{gather*}
\inferrule{\forall a \in \dom{\Sigma_2},~\Sigma_1(a)=\Sigma_2(a)}{\Sigma_1\ledyn\Sigma_2}
\end{gather*}
  \end{minipage}
  \caption{Weaker and stronger heap types.}
  \label{cp:fig:apx:ht}
\end{figure}

\begin{figure}
  \begin{minipage}{1.0\linewidth}
    \begin{minipage}{.49\linewidth}      
    \[
\begin{array}{rcl}
  \rho & ::= & \cdot \mid \rho,x=v
\end{array}
\]
\boxed{\Sigma\vdash \rho:\Gamma}
\begin{gather*}
  \inferrule{ }{\Sigma\vdash \cdot:\emptyset}\qquad
  \inferrule{\emptyset;\Sigma\vdash v:T \\ \Sigma\vdash \rho:\Gamma}{\Sigma\vdash \rho,x=v:\Gamma,x:T}
\end{gather*}

    \end{minipage}
    \begin{minipage}{.49\linewidth}
\boxed{|d|=e}
\[
\begin{array}{rcl}
  |x| & = & x\\
  |n| & = & n\\
  |\fntx{x}{X}{Y}{d}| & = & \fnx{x}{|d|}\\
  |d_1~d_2| & = & (|d_1|~|d_2|)^\tyorig \\
  |\refxi{d}{X}| & = & \refx{|d|} \\ 
  |\derefx{d}| & = & \derefxp{|d|}{\tyorig} \\
  |\mutx{d_1}{d_2}| & = & \mutxpINV{|d_1|}{|d_2|}{\tyorig} \\  
  |d_1~+~d_2| & = & |d_1|~+^\tyorig~|d_2| \\
  |\checkx{d}{S}| & = & \checkx{|d|}{S} \\
\end{array}
\]  
    \end{minipage}

  \end{minipage}
  \caption{Rules for environments and for erasing to $\langd$ to $\lange$.}
  \label{cp:fig:apx:erase}
\end{figure}

\section{Appendix: Proofs}\label{cp:apx:proof}
\setcounter{figure}{0}    
\subsection{Soundness of $\langs$}

\begin{lemma}\label{cp:lem:specmatchrel}
  If $\specmatches{A_1}{A_2}{\lowlty{U}}$, then $\R{U}{A_2}$.
\end{lemma}
\begin{proof}
  By cases on $\specmatches{A_1}{A_2}{\lowlty{U}}$.
\end{proof}

\begin{lemma}\label{cp:lem:matchspecmatch}
  Suppose $\R{U}{A}$.
  \begin{enumerate}
  \item If $\matchrel{U}{U_1\to U_2}$, then $\specmatches{A}{V_1\to V_2}{\funs}$.
  \item If $\matchrel{U}{\reft{U'}}$, then $\specmatches{A}{\reft{V}}{\refs}$.
  \item If $U \sim \intt$, then $\specmatches{A}{\intt}{\intt}$.
  \end{enumerate}
\end{lemma}
\begin{proof}
  We prove part 1 by cases on $U$. If $U=\dyn$, then $A=V$, and
  $\specmatches{V}{\inflow{V}\to \outflow{V}}{\funs}$. If $U=U_1\to
  U_2$, then either $A=V$ and the theorem holds as above, or
  $A=\reft{V_1}{V_2}$, and $\specmatches{A}{V_1\to V_2}{\funs}$.

  The proofs of parts 2 and 3 are similar.
\end{proof}

\begin{lemma}\label{cp:lem:weakening}
  If $\Gamma\vdash d:A;\Omega$ and for all $x\in\mathit{dom}(\Gamma)$, $\Gamma'(x)=\Gamma(x)$, then $\Gamma'\vdash d:A;\Omega$.
\end{lemma}
\begin{proof}
  Induction on $\Gamma\vdash d:A;\Omega$.
\end{proof}

\begin{lemma}\label{cp:lem:apx:stodpres}
  If $\Gamma\vdash s\leadsto d:U$ and for all $x\in\mathit{dom}(\Gamma)$, $\R{\Gamma(x)}{\Gamma'(x)}$, then $\Gamma'\vdash d:A;\Omega$ and $\R{U}{A}$.
\end{lemma}
\begin{proof}
  By induction on $\Gamma\vdash s\leadsto d:U$.
  \begin{description}
  \item[Case] \textsc{UAbs}: \[
    \inferrule{\Gamma,x{:}U_1\vdash s\leadsto d:U_2' \\ U_2' \sim U_2 \\ X,Y~\text{fresh}}
    {\Gamma\vdash \fntx{x}{U_1}{U_2}{s}\leadsto \\ \fntx{x}{X}{Y}{\letx{x}{\checkx{x}{\lowlty{U_1}}}{d}}:\funt{U_1}{U_2}}
    \]
    By \textsc{IVar}, $\Gamma',x:X\vdash x:X;\emptyset$.\\
    Have that $\specmatches{X}{A_1}{\lowlty{U_1}}$.\\
    By \textsc{ICheck}, $\Gamma',x:X\vdash \checkx{x}{\lowlty{U_1}}:A_1;\emptyset$.\\
    By Lemma \ref{cp:lem:specmatchrel}, $\R{U_1}{A_1}$.\\
    By the IH, $\Gamma',x{:}A_1\vdash d:A_2;\Omega$ and $\R{U_2'}{A_2}$.\\
    By Lemma \ref{cp:lem:weakening}, $\Gamma',x{:}X,x{:}A_1\vdash d:A_2;\Omega$.\\
    By \textsc{ILet}, $\Gamma',x{:}X\vdash\letx{x}{\checkx{x}{\lowlty{U_1}}}{d}:A_2;\Omega$.\\
    By \textsc{IAbs}, $\Gamma'\vdash\fntx{x}{X}{Y}{\letx{x}{\checkx{x}{\lowlty{U_1}}}{d}}:\funt{X}{Y};\Omega,\flows{A_2}{Y}$.\\
    Have that $\R{\funt{U_1}{U_2}}{\funt{X}{Y}}$.
  \item[Case] \textsc{UApp}: \[
    \inferrule{\Gamma\vdash s_1\leadsto d_1:U \\ \matchrel{U}{\funt{U_1}{U_2}} \\
      \Gamma\vdash s_2\leadsto d_2:U_1' \\ U_1' \sim U_1}
    {\Gamma\vdash s_1~s_2\leadsto \checkx{(\checkx{d_1}{\funs})~d_2}{\lowlty{U_2}}:U_2}
    \]
    By the IH, $\Gamma'\vdash d_1:A_1;\Omega_1$ and $\R{U}{A_1}$.\\
    By Lemma \ref{cp:lem:matchspecmatch}, $\specmatches{A_1}{V_1\to V_2}{\funs}$.\\
    By \textsc{ICheck}, $\Gamma'\vdash\checkx{d_1}{\funs}:V_1\to V_2$.\\
    By the IH, $\Gamma'\vdash d_2:A_2;\Omega_2$ and $\R{U_1'}{A_2}$.\\
    By \textsc{IApp}, $\Gamma'\vdash(\checkx{d_1}{\funs})~d_2:V_2;\Omega_1,\Omega_2,\flows{A_2}{V_1}$.\\
    Have that $\specmatches{V_2}{A_3}{\lowlty{U_2}}$.\\
    By \textsc{ICheck}, $\Gamma'\vdash\checkx{((\checkx{d_1}{\funs})~d_2)}{\lowlty{U_2}}:A_3;\Omega_1,\Omega_2,\flows{A_2}{V_1}$.\\
    By Lemma \ref{cp:lem:specmatchrel}, $\R{U_2}{A_3}$.
  \item[Case] \textsc{URef}: \[
    \inferrule{\Gamma\vdash s\leadsto d:U_2 \\ U_2\sim U_1 \\ X~\text{fresh}}{\Gamma\vdash \refxi{s}{U_1}\leadsto \refxi{d}{X}:\reft{U_1}}
    \]
    By the IH, $\Gamma'\vdash d:A;\Omega$ and $\R{U_2}{A}$.\\
    By \textsc{IRef}, $\Gamma'\vdash \refxi{d}{X}:\reft{X};\Omega,\flows{A}{X}$.\\
    Have that $\R{\reft{U}}{\reft{X}}$.
  \item[Case] \textsc{UDeref}: \[
    \inferrule{\Gamma\vdash s\leadsto d:U \\ \matchrel{U}{\reft{U'}}}
    {\Gamma\vdash\derefx{s}\leadsto \checkx{\derefx{(\checkx{d}{\refs})}}{\lowlty{U'}}:U'}
    \]
    By the IH, $\Gamma'\vdash d:A_1;\Omega$ and $\R{U}{A_1}$.\\
    By Lemma \ref{cp:lem:matchspecmatch}, $\specmatches{A_1}{\reft{V}}{\refs}$.\\
    By \textsc{ICheck}, $\Gamma'\vdash\checkx{d}{\refs}:\reft{V};\Omega$.\\
    By \textsc{IDeref}, $\Gamma'\vdash\derefx{(\checkx{d}{\refs})}:V;\Omega$.\\
    Have that $\specmatches{V}{A_2}{\lowlty{U'}}$.\\
    By \textsc{ICheck}, $\Gamma'\vdash\checkx{\derefx{(\checkx{d}{\refs})}}{\lowlty{U'}}:A_2;\Omega$.\\
    By Lemma \ref{cp:lem:specmatchrel}, $\R{U'}{A_2}$.
  \item[Case] \textsc{UUpdt}: \[
    \inferrule{\Gamma\vdash s_1\leadsto d_1:U \\ \matchrel{U}{\reft{U'}} \\
      \Gamma\vdash s_2\leadsto d_2:U'' \\ U'' \sim U'}
    {\Gamma\vdash\mutx{s_1}{s_2}\leadsto \mutx{\checkx{d_1}{\refs}}{d_2}:\intt}
    \]
    By the IH, $\Gamma'\vdash d_1:A_1;\Omega_1$ and $\R{U}{A_1}$.\\
    By Lemma \ref{cp:lem:matchspecmatch}, $\specmatches{A_1}{\reft{V}}{\refs}$.\\
    By \textsc{ICheck}, $\Gamma'\vdash\checkx{d_1}{\refs}:\reft{V};\Omega_1$.\\
    By the IH, $\Gamma'\vdash d_2:A_2;\Omega_2$ and $\R{U''}{A_2}$.\\
    By \textsc{IUpdt}, $\Gamma'\vdash\mutx{\checkx{d_1}{\refs}}{d_2}:\intt;\Omega_1,\Omega_2,\flows{A_2}{V}$.
  \item[Case] \textsc{UAdd}: \[
    \inferrule{\Gamma\vdash s_1\leadsto d_1:U_1 \\ U_1\sim \intt \\
      \Gamma\vdash s_2\leadsto d_2:U_2 \\ U_2\sim \intt}
    {\Gamma\vdash s_1~+~s_2\leadsto \checkx{d_1}{\intt}~+~\checkx{d_2}{\intt}:\intt}
    \]
    By the IH, $\Gamma'\vdash d_1:A_1;\Omega_1$ and $\R{U_1}{A_1}$.\\
    By Lemma \ref{cp:lem:matchspecmatch}, $\specmatches{A_1}{\intt}{\intt}$.\\
    By \textsc{ICheck}, $\Gamma'\vdash\checkx{d_1}{\intt}:\intt;\Omega_1$.\\
    By the IH, $\Gamma'\vdash d_2:A_2;\Omega_2$ and $\R{U_2}{A_2}$.\\
    By Lemma \ref{cp:lem:matchspecmatch}, $\specmatches{A_2}{\intt}{\intt}$.\\
    By \textsc{ICheck}, $\Gamma'\vdash\checkx{d_2}{\intt}:\intt;\Omega_2$.\\
    By \textsc{IAdd}, $\Gamma'\vdash \checkx{d_1}{\intt}~+~\checkx{d_2}{\intt}:\intt;\Omega_1,\Omega_2$.
  \item[Case] \textsc{UVar}: \[
    \inferrule{\Gamma(x)=U}{\Gamma\vdash x\leadsto x:U}
    \]
    Have that $\Gamma'(x)=A$ and $\R{U}{A}$.\\
    By \textsc{IVar}, $\Gamma'\vdash x:A;\emptyset$.
  \item[Case] \textsc{UInt}: Immediate.
  \end{description}
\end{proof}

\subsection{Soundness of $\langd$}\label{cp:apx:proof:langd}

\begin{lemma}\label{cp:lem:apx:solsplit}
  If $\sigma$ is a solution to $\Omega_1\cup \Omega_2$, then $\sigma$ is a solution to $\Omega_1$ and $\sigma$ is a solution to $\Omega_2$. 
\end{lemma}
\begin{proof}
  Since $\sigma$ is a solution for every constraint in $\Omega$, and for all $C\in\Omega_1$, $C\in\Omega$, so $\sigma$ is a solution for every constraint in $\Omega_1$, so it is a solution to $\Omega_1$.
  Likewise for $\Omega_1$.
\end{proof}

\begin{lemma}\label{cp:lem:sttrans}
  If $\vdash T_1 <: T_2$ and $\vdash T_2 <: T_3$, then $\vdash T_1 <: T_3$.
\end{lemma}
\begin{proof}
  Straightforward induction.
\end{proof}

\begin{lemma}\label{cp:lem:stsameconst}
  If $\lowlty{T_1} \preceq S$ and $\vdash T_2 <: T_1$, then $\lowlty{T_2} \preceq S$. 
\end{lemma}
\begin{proof}
  Cases on $T_1$.
\end{proof}

\begin{lemma}\label{cp:lem:apx:soundsubclosed}
  If $\Gamma\vdash  d:A;\Omega$ and $\sigma$ is a solution for $\Omega$, then $\Gamma;\sigma\vdash d\leadsto e:T$ and $\vdash T <: \sigma A$.
\end{lemma}
\begin{proof}
  By induction on $\Gamma\vdash d:A;\Omega$
  \begin{description}
  \item[Case] \textsc{IVar}: \[
    \inferrule{\Gamma(x)=A}{\Gamma\vdash x:A;\emptyset}
    \]
    Have that $\Gamma(x)=A$.\\
    By \textsc{DVar}, $\Gamma;\sigma\vdash x\leadsto x:\sigma A$.
  \item[Case] \textsc{IAbs}: \[
    \inferrule{\Gamma,x{:}\alpha\vdash d:A\yO}
    {\Gamma\vdash \fntx{x}{\alpha}{\beta}{d}:\funt{\alpha}{\beta}\yO,\flows{A}{\beta}}  
    \]
    Since $\sigma$ is a solution for $\Omega \cup \{\flows{A}{\beta}\}$, by Lemma \ref{cp:lem:apx:solsplit}, $\sigma$ is a solution for $\Omega$.\\
    By the IH, $\Gamma,x{:}\alpha;\sigma\vdash d\leadsto e: T$ and $\vdash T' <: \sigma A$.\\
    Since $\sigma$ is a solution to $\Omega,\flows{A}{\beta}$, $\vdash \sigma A <: \sigma \beta$.\\
    By Lemma \ref{cp:lem:sttrans}, $\vdash T' <: \sigma \beta$.\\ 
    By \textsc{DAbs}, $\Gamma;\sigma\vdash \fntx{x}{\alpha}{\beta}{d} \leadsto \fnx{x}{e}: \funt{\sigma \alpha}{\sigma \beta}$.
  \item[Case] \textsc{IApp}: \[
    \inferrule{\Gamma\vdash d_1:\funt{V_1}{V_2}\yO_1 \\ 
      \Gamma\vdash d_2:A\yO_2}
    {\Gamma\vdash d_1~d_2:V_2\yO_1,\Omega_2,\flows{A}{V_1}}
    \]
    Since $\sigma$ is a solution for $\Omega_1,\Omega_2,\flows{A_1}{V_1},\flows{A}{\funt{V_1}{V_2}}$, by Lemma \ref{cp:lem:apx:solsplit}, $\sigma$ is a solution for $\Omega_1$ and $\sigma$ is a solution for $\Omega_2$.\\
    By the IH, $\Gamma;\sigma\vdash d_1\leadsto e_1: T_1$ and $\vdash T_1 <: \funt{\sigma V_1}{\sigma V_2}$.\\
    Therefore $T_1=\funt{T_{11}}{T_{12}}$ and $\vdash \sigma V_1 <: T_{11}$ and $\vdash T_{12} <: \sigma V_2$.\\
    By the IH, $\Gamma;\sigma\vdash d_2\leadsto e_2: T_2$ and $\vdash T_2 <: \sigma A$.\\
    Since $\sigma$ is a solution to $\Omega,\flows{A}{V_1}$, $\vdash \sigma A <: \sigma V_1$.\\
    By Lemma \ref{cp:lem:sttrans}, $\vdash T_2 <: T_{11}$.\\
    By \textsc{DApp},
    $\Gamma;\sigma\vdash d_1~d_2\leadsto (e_1~e_2)^\tyorig: T_{12}$.
  \item[Case] \textsc{ICheck}: \[
    \inferrule{\Gamma\vdash d:A_1\yO \\ \specmatches{A_1}{A_2}{S}}
    {\Gamma\vdash \checkx{d}{S}:A_2\yO;\depcon{A_1}{S}{A_2}}
    \]
    By the IH, $\Gamma;\sigma\vdash d\leadsto e: T$ and $\vdash T <: \sigma A_1$.\\
    If $S=\dyn$, then $A_2=A_1$, and by \textsc{DCheckRemove}, $\Gamma;\sigma\vdash \checkx{d}{S}\leadsto e:T$.\\
    We proceed by cases on $A_1$ and $\lowlty{\sigma A_1}$.\\
    \begin{description}
    \item[Case] $A_1\neq \alpha$ and $S=\lowlty{\sigma A_1}$:\\
      Then $S=\lowlty{A_1}$.\\
      Then $A_2=A_1$.\\
      By Lemma \ref{cp:lem:stsameconst}, $S=\lowlty{T}$.\\
      By \textsc{DCheckRemove}, $\Gamma;\sigma\vdash \checkx{d}{S}\leadsto e:T$.
    \item[Case] $A_1\neq \alpha$ and $S$ is not the constructor of $\sigma A_1$: Vacuous, since $S\neq\dyn$.
    \item[Case] $A_1=\alpha$ and $S$ is the constructor of $\sigma A_1$:\\
      Then $\sigma A_1=\sigma A_2$.\\
      Therefore $\vdash T <: \sigma A_2$.\\
      By Lemma \ref{cp:lem:stsameconst}, $\lowlty{T} \preceq S$.\\
      By \textsc{DCheckRemove}, $\Gamma;\sigma\vdash \checkx{d}{S}\leadsto e:T$.
    \item[Case] $A_1=\alpha$ and $S$ is not the constructor of $\sigma A_1$:\\
      Since $\sigma$ is a solution for $\Omega \cup \{\depcon{A_1}{S}{A_2}\}$, for all $\beta\in\mathit{parts}(A_2)$, $\sigma\beta=\dyn$.\\
      By the definition of $\rhd_S$, since $S\neq\dyn$, $\lowlty{A_2}=S$.\\
      Therefore $\sigma A_2=\uplty{S}$.\\
      If $T=\dyn$, then by \textsc{DCheckKeep}, $\Gamma;\sigma\vdash \checkx{d}{S}\leadsto \checkx{e}{S}:\uplty{S}$.\\
      Otherwise, by \textsc{DCheckFail}, $\Gamma;\sigma\vdash \checkx{d}{S}\leadsto\mathtt{fail}:\uplty{S}$.
    \end{description}
  \item[Case] \textsc{IRef}: \[
    \inferrule{\Gamma\vdash d:A\yO}{\Gamma\vdash \refxi{d}{\alpha}:\reft{\alpha}\yO,\flows{A}{\alpha}}\
    \]
    By the IH, $\Gamma;\sigma\vdash d\leadsto e: T'$ and $\vdash T' <: \sigma A$.\\
    Since $\sigma$ is a solution to $\Omega \cup \{\flows{A}{\alpha}\}$, $\vdash \sigma A <: \sigma \alpha$.\\
    By Lemma \ref{cp:lem:sttrans}, $\vdash T' <: \sigma \alpha$.\\
    By \textsc{DRef}, $\Gamma;\sigma\vdash \refxi{d}{\alpha}\leadsto \refx{e}: \reft{\sigma \alpha}$.
  \item[Case] \textsc{IDeref}: \[
    \inferrule{\Gamma\vdash d:\reft{V}\yO}
    {\Gamma\vdash\derefx{d}:V\yO}
    \]
    By the IH, $\Gamma;\sigma\vdash d\leadsto e: T$ and $\vdash T <: \reft{\sigma V}$.\\
    Therefore $T=\reft{T'}$ and $\vdash T' <: \sigma V$.\\
    By \textsc{DDeref}, $\Gamma;\sigma\vdash \derefx{d}\leadsto \derefxp{e}{\tyorig}:T'$.
  \item[Case] \textsc{IUpdt}: \[
    \inferrule{\Gamma\vdash d_1:\reft{V}\yO_1 \\
      \Gamma\vdash d_2:A\yO_2}
    {\Gamma\vdash\mutx{d_1}{d_2}:\intt\yO_1,\Omega_2,\flows{A}{V}}
    \]
    Since $\sigma$ is a solution to $\Omega_1\cup\Omega_2\cup\{\flows{A}{V}\}$, by Lemma \ref{cp:lem:apx:solsplit}, $\sigma$ is a solution to $\Omega_1$ and $\sigma$ is a solution to $\Omega_2$ and $\sigma$ is a solution to $\{\flows{A}{V}\}$.\\
    By the IH, $\Gamma;\sigma\vdash d_1\leadsto e_1: T_1$ and $\vdash T_1 <: \reft{\sigma V}$.\\
    Therefore $T_1=\reft{T_1'}$ and $\vdash \sigma V <: T_1'$\\
    By the IH, $\Gamma;\sigma\vdash d_2\leadsto e_2: T_2$ and $\vdash T_2 <: \sigma A$.\\
    Since $\sigma$ is a solution to $\{\flows{A}{V}\}$, $\vdash \sigma A <: \sigma V$.\\
    By Lemma \ref{cp:lem:sttrans}, $\vdash T_2 <: T_1'$.\\
    By \textsc{DDeref}, $\Gamma;\sigma\vdash \mutx{d_1}{d_2}\leadsto \mutxpINV{e_1}{e_2}{\tyorig}:\intt$.
  \item[Case] \textsc{IAdd}: \[
    \inferrule{\Gamma\vdash d_1:\intt\yO_1 \\
      \Gamma\vdash d_2:\intt\yO_2}
    {\Gamma\vdash d_1~+~d_2:\intt\yO_1,\Omega_2}
    \]
    Since $\sigma$ is a solution to $\Omega_1\cup\Omega_2$, by Lemma \ref{cp:lem:apx:solsplit}, $\sigma$ is a solution to $\Omega_1$ and $\sigma$ is a solution to $\Omega_2$.\\
    By the IH, $\Gamma;\sigma\vdash d_1\leadsto e_1: T_1$ and $\vdash T_1 <: \intt$.\\
    By the IH, $\Gamma;\sigma\vdash d_2\leadsto e_2: T_2$ and $\vdash T_2 <: \intt$.\\
    Therefore $T_1=T_2=\intt$.\\
    Therefore by \textsc{DAdd}, $\Gamma;\sigma\vdash d_1~+~ d_2\leadsto (e_1~+~e_2)^\tyorig: \intt$.
  \item[Case ] \textsc{IInt}: Immediate by \textsc{DInt}.
  \end{description}
\end{proof}

\begin{restatable}{lemma}{soundsub}\label{cp:lem:soundsub}
  If $\Gamma\vdash d:A;\Omega_1$ and $\vdash A:\Omega_2$ and $\sigma$
  is a solution for $\Omega_1\cup\Omega_2$, then $\Gamma;\sigma\vdash
  d\leadsto e:T$ and $\vdash T <: \sigma A$.  
\end{restatable}
\begin{proof}
  By Lemma \ref{cp:lem:apx:solsplit}, $\sigma$ is a solution for
  $\Omega_1$. \\
  By Lemma \ref{cp:lem:apx:soundsubclosed},
  $\Gamma;\sigma \vdash d \leadsto e:T$ and $\vdash T <: \sigma A$.
\end{proof}

\begin{restatable}{lemma}{dtoe}\label{cp:lem:dtoe}
  If $\Gamma;\sigma\vdash d\leadsto e:T$, then $\sigma\Gamma;\emptyset \vdash e:T$. 
\end{restatable}
\begin{proof}
  By induction on $\Gamma;\sigma\vdash d \leadsto e: T$.
  \begin{description}
  \item[Case ] \textsc{OAbs}:\[
    \inferrule{\Gamma,x{:}\alpha;\sigma\vdash\ei \leadsto e : T \\ \vdash T <: \sigma\beta}
    {\Gamma;\sigma\vdash \fntx{x}{\alpha}{\beta}{\ei}\leadsto \fnx{x}{e}: \funt{\sigma\alpha}{\sigma\beta}}
      \]
      By the IH, $\sigma\Gamma,x{:}\sigma\alpha{;}\emptyset\vdash e:T$.\\
      By \textsc{TSubsump}, $\sigma\Gamma,x{:}T_1{;}\emptyset\vdash e:\sigma\beta$.\\
      By \textsc{TAbs}, $\sigma\LGE \fnx{x}{e}:\funt{\sigma\alpha}{\sigma\beta}$.
    \item[Case ] \textsc{OLet}:\[
      \inferrule{\Gamma;\sigma\vdash \ein{1}\leadsto e_1: T_1 \\ \Gamma,x{:}T_1;\sigma\vdash \ein{2}\leadsto e_2: T_2}
      {\Gamma;\sigma\vdash\letx{x}{\ein{1}}{\ein{2}}\leadsto\letx{x}{e_1}{e_2}: T_2}
      \]
      By the IH, $\sigma\LGE e_1:T_1$.\\
      By the IH, $\sigma\Gamma,x{:}T_1{;}\emptyset\vdash e_2:T_2$.\\
      By \textsc{TLet}, $\sigma\LGE \letx{x}{e_1}{e_2}:T_2$.
    \item[Case ] \textsc{OApp}:\[
      \inferrule{\Gamma;\sigma\vdash \ein{1}\leadsto e_1 :\funt{T_1}{T_2} \\
        \Gamma;\sigma\vdash \ein{2}\leadsto e_2 :T_1' \\ \vdash T_1' <: T_1}
      {\Gamma;\sigma\vdash \ein{1}~\ein{2}\leadsto (e_1~e_2)^\tyorig: T_2}\\[1ex]
      \]
      By the IH, $\sigma\LGE e_1:\funt{T_1}{T_2}$.\\
      By the IH, $\sigma\LGE e_2:T_1'$.\\
      By \textsc{TSubsump}, $\sigma\LGE e_2:T_2$.\\
      By \textsc{TApp}, $\sigma\LGE (e_1~e_2)^\tyorig: T_2$.
    \item[Case ] \textsc{ORef}: \[
      \inferrule{\Gamma;\sigma\vdash \ei\leadsto e:T \\ \vdash T <: \sigma\alpha}{\Gamma;\sigma\vdash \refxi{\ei}{\alpha}\leadsto \refx{e}:\reft{T}}
      \]
      By the IH, $\sigma\LGE e:T$.\\
      By \textsc{TSubsump}, $\sigma\LGE e:\sigma\alpha$.\\
      By \textsc{TRef}, $\sigma\LGE \refx{e}:\reft{\sigma\alpha}$.
    \item[Case ] \textsc{ODeref}:\[
      \inferrule{\Gamma;\sigma\vdash\ei\leadsto e : \reft{T}}
      {\Gamma;\sigma\vdash\derefx{\ei}\leadsto \derefxp{e}{\tyorig}: T}
      \]
      By the IH, $\sigma\LGE e:\reft{T}$.\\
      By \textsc{TDeref}, $\sigma\LGE \derefxp{e}{\tyorig} : T$.
    \item[Case ] \textsc{OUpdt}:\[
      \inferrule{\Gamma;\sigma\vdash\ein{1}\leadsto e_1: \reft{T} \\
        \Gamma;\sigma\vdash\ein{2}\leadsto e_2:T' \\ \vdash T' <: T}
      {\Gamma;\sigma\vdash\mutx{\ein{1}}{\ein{2}}\leadsto\mutxpINV{e_1}{e_2}{\tyorig}:\intt}\]
      By the IH, $\sigma\LGE e_1:\reft{T}$.\\
      By the IH, $\sigma\LGE e_2:T'$.\\
      By \textsc{TSubsump}, $\sigma\LGE e_2:T$.\\ 
      By \textsc{TDeref}, $\sigma\LGE \mutxpINV{e_1}{e_2}{\tyorig} : \intt$.
    \item[Case ] \textsc{OVar}: \[
      \inferrule{\Gamma(x)=A}
                {\Gamma;\sigma\vdash x\leadsto x:\sigma A}
      \]
      Have that $\sigma\Gamma(x)=\sigma A$.\\
      By \textsc{TVar}, $\sigma\LGE x:\sigma A$.
    \item[Case ] \textsc{OInt}: Immediate from \textsc{TInt}.
    \item[Case ] \textsc{OAdd}: \[
      \inferrule{\Gamma;\sigma\vdash \ein{1}\leadsto e_1: \intt \\ 
        \Gamma;\sigma\vdash \ein{2}\leadsto e_2:\intt}
      {\Gamma;\sigma\vdash \ein{1}~+~\ein{2}\leadsto e_1~+^\tyorig~e_2:\intt}
      \]
      By the IH, $\sigma\LGE e_1:\intt$.\\
      By the IH, $\sigma\LGE e_2:\intt$.\\
      By \textsc{TAdd}, $\sigma\LGE e_1~+^\tyorig~e_2:\intt$.
    \item[Case ] \textsc{OCheckRemove}:\[
      \inferrule{\Gamma;\sigma\vdash \ei\leadsto e : T \\ \lowlty{T} \preceq S}
      {\Gamma;\sigma\vdash \checkx{\ei}{S} \leadsto e : T}
      \]
      Immediate from the IH.
    \item[Case ] \textsc{OCheckKeep}: \[
      \inferrule{\Gamma;\sigma\vdash \ei\leadsto e : \dyn}
      {\Gamma;\sigma\vdash \checkx{\ei}{S} \leadsto \checkx{e}{S} : \uplty{S}}
      \]
      By the IH, $\sigma\LGE e:\dyn$.\\
      By \textsc{TCheck}, $\sigma\LGE \checkx{e}{S}:\uplty{S}$.\\
    \item[Case ] \textsc{OCheckFail}: \[
      \inferrule{\Gamma;\sigma\vdash \ei\leadsto e : T \\ T \neq \dyn \\ \lowlty{T} \not\preceq S}
      {\Gamma;\sigma\vdash \checkx{\ei}{S} \leadsto \mathtt{fail} : T'}
      \]
      Immediate from \textsc{TFail}.
    \end{description}
\end{proof}

\subsection{Soundness of $\lange$}\label{cp:apx:proof:lange}

\begin{lemma}[Inversion]\label{cp:lem:inversion}
  Suppose $\GS e:T$. Then:
  \begin{itemize}
  \item If $e=\fnx{x}{e}$, then there exists $T_1$ such that $\Gamma,x{:}T_1{;}\Sigma\vdash e:T_2$ and $\vdash T_1\to T_2 <: T$.
  \item If $e=a$, then $\Sigma(a)=T'$ and $\vdash \reft{T'} <: T$.
  \item If $e=n$, then $\vdash \intt <: T$.
  \item If $e=\checkx{e}{S}$, then either:
    \begin{enumerate}
    \item $\GS e:\dyn$ and $\vdash \uplty{S} <: T$, or
    \item $\GS e:T'$ and $\lowlty{T'} \preceq S$ and $\vdash T' <: T$, or
    \item $\GS e:T''$ and $\lowlty{T'}\not\preceq S$ and $T''\neq\dyn$.
    \end{enumerate}
  \item If $e=x$, then $\Gamma(x)=T'$ and $\vdash T' <: T$.
  \item If $e=(e_1~e_2)^\tyorig$, then $\GS e_1:\funt{T_1}{T_2}$ and $\GS e_2:T_1$ and $\vdash T_2 <: T$.
  \item If $e=(e_1~e_2)^\unorig$, then $\GS e_1:\dyn$ and $\GS e_2:\dyn$ and $\vdash \dyn <: T$.
  \item If $e=\letx{x}{e_1}{e_2}$, then $\GS e_1:T_1$ and $\Gamma{,}x:T_1{;}\Sigma\vdash e_2:T_2$ and $\vdash T_2 <: T$.
  \item If $e=\refx{e}$, then $\GS e:T'$ and $\vdash \reft{T'} <: T$.
  \item If $e=\derefxp{e}{\tyorig}$, then $\GS e:\reft{T'}$ and $\vdash T' <: T$.
  \item If $e=\derefxp{e}{\unorig}$, then $\GS e:\dyn$ and $\vdash \dyn <: T$.
  \item If $e=\mutxpINV{e_1}{e_2}{\tyorig}$, then $\GS e_1:\reft{T'}$ and $\GS e_2:T'$ and $\vdash \intt <: T$.
  \item If $e=\mutxpINV{e_1}{e_2}{\unorig}$, then $\GS e_1:\dyn$ and $\GS e_2:\dyn$ and $\vdash \intt <: T$.
  \item If $e=e_1~+^\tyorig~e_2$, then $\GS e_1:\intt$ and $\GS e_2:\intt$ and $\vdash \intt <: T$.
  \item If $e=e_1~+^\unorig~e_2$, then $\GS e_1:\dyn$ and $\GS e_2:\dyn$ and $\vdash \intt <: T$.
  \end{itemize}
\end{lemma}
\begin{proof}
  Induction on $\GS e:R$.
\end{proof}

\begin{lemma}[Canonical forms]\label{cp:lem:apx:canonical}
  Suppose $\emptyset{;}\Sigma\vdash v:T$ and $\Sigma\vdash\mu$.
  \begin{enumerate}
  \item If $T=\intt$, then $v=n$.
  \item If $T=\funt{T_1}{T_2}$, then $v=\fnx{x}{e}$.
  \item If $T=\reft{T'}$, then $v=a$ and $\mu(a)=v'$.
  \item If $T=\dyn$, then $\exists T'$, $T'\neq\dyn$, such that $\emptyset{;}\Sigma\vdash v:T'$.
  \end{enumerate}
\end{lemma}
\begin{proof}
  We prove each part separately.
  \begin{enumerate}
  \item We prove by induction on $\emptyset{;}\Sigma\vdash v:\intt$. Most cases vacuous.
    \begin{description}
    \item[Case ]\textsc{TInt} \[
      \inferrule{ }{\emptyset{;}\Sigma \vdash n : \intt}
      \]
      Immediate.
    \item[Case ]\textsc{TSubsump}  \[
      \inferrule{\emptyset{;}\Sigma \vdash v : T_1 \\ \vdash T_1 <: \intt}{\emptyset{;}\Sigma \vdash v : \intt}
      \]
      Since $\vdash T_1 <: \intt$, $T_1 = \intt$.\\
      By the IH, $v=n$.
    \end{description}
  \item We prove by induction on $\emptyset{;}\Sigma\vdash v:\funt{T_1}{T_2}$. Most cases vacuous.
    \begin{description}
    \item[Case ]\textsc{TAbs} \[
      \inferrule{x{:}T_1{;}\Sigma\vdash e:T_2}{\emptyset{;}\Sigma \vdash \fnx{x}{e} : \funt{T_1}{T_2}}
      \]
      Immediate.
    \item[Case ]\textsc{TSubsump}  \[
      \inferrule{\emptyset{;}\Sigma \vdash v : T' \\ \vdash T' <: \funt{T_1}{T_2}}{\emptyset{;}\Sigma \vdash v : \funt{T_1}{T_2}}
      \]
      Since $\vdash T' <: \funt{T_1}{T_2}$, $T' = \funt{T_1'}{T_2'}$.\\
      By the IH, $v=\fnx{x}{e}$.
    \end{description}
  \item We prove by induction on $\emptyset{;}\Sigma\vdash v:\reft{T'}$. Most cases vacuous.
    \begin{description}
    \item[Case ]\textsc{TAddr} \[
      \inferrule{\Sigma(a)=T'}{\emptyset{;}\Sigma \vdash a : \reft{T'}}
      \]
      Immediately have $v=a$. Since $\Sigma(a)=T'$, exists $v'$ such that $\mu(a)=v'$.
    \item[Case ]\textsc{TSubsump}  \[
      \inferrule{\emptyset{;}\Sigma \vdash v : T'' \\ \vdash T'' <: \reft{T'}}{\emptyset{;}\Sigma \vdash v : \reft{T'}}
      \]
      Since $\vdash T'' <: \reft{T'}$, $T'' = \reft{T'''}$.\\
      By the IH, $v=a$ and $\mu(a)=v'$.
    \end{description}
  \item We prove by induction on $\emptyset{;}\Sigma\vdash v:\dyn$. Most cases vacuous.
    \begin{description}
    \item[Case ]\textsc{TSubsump}  \[
      \inferrule{\emptyset{;}\Sigma \vdash v : T' \\ \vdash T' <: \dyn}{\emptyset{;}\Sigma \vdash v : \dyn}
      \]
      If $T'=\dyn$, then we apply the IH to find that $\exists T''$, $T''\neq\dyn$, such that $\emptyset{;}\Sigma\vdash v:T''$.\\
      If $T'\neq\dyn$, then immediate.
    \end{description}
  \end{enumerate}
\end{proof}

\begin{lemma}[Heap weakening]\label{cp:lem:apx:heapweak}
  If $\GS e:T$ and $\Sigma'\ledyn\Sigma$, then $\Gamma{;}\Sigma'\vdash e:T$.
\end{lemma}
\begin{proof}
  By induction on $\GS e:T$. Only interesting case:
  \begin{description}
  \item[Case ] \textsc{TAddr}: \[
    \inferrule{\Sigma(a)=T}{\GS a:\reft{T}}
    \]
    Because $\Sigma'\ledyn\Sigma$, $\Sigma'(a)=T$.\\
    By \textsc{TAddr}, $\GS a:\reft{T}$.
  \end{description}
\end{proof}

\begin{lemma}[Heap extension]\label{cp:lem:heapext}
  If $\Sigma\vdash\mu$ and $\emptyset{;}\Sigma{,}a{:}T\vdash v:T$ and
  $a\not\in\dom{\Sigma}$, then $\Sigma{,}a{:}T\vdash\mu[a:=v]$.
\end{lemma}
\begin{proof}
  Since $a\not\in\dom{\Sigma}$, for all $a' \in \dom{\Sigma}$, $\Sigma(a')=(\Sigma{,}a{:}T)(a')$.\\
  Therefore $\Sigma{,}a{:}T \ledyn \Sigma$.\\
  Suppose $a'\in\dom{\Sigma{,}a{:}T}$. If $a=a'$, then immediately $\emptyset{;}\Sigma{,}a{:}T\vdash \mu[a:=v](a):(\Sigma{,}a{:}T)(a)$.\\
  If $a\neq a'$, then $\ES \mu(a'):\Sigma(a')$.\\
  Immediately have $\mu(a')=\mu[a:=v](a')$ and $\Sigma(a')=(\Sigma{,}a{:}T)(a')$.\\
  By Lemma \ref{cp:lem:apx:heapweak}, $\emptyset{;}\Sigma{,}a{:}T \vdash \mu[a:=v](a'):(\Sigma{,}a{:}T)(a')$.\\
  Therefore, $\Sigma{,}a{:}T\vdash\mu[a:=v]$.
\end{proof}

\begin{lemma}[Substitution]\label{cp:lem:apx:subst}
  If $\Gamma{,}x{:}T_1{;}\Sigma\vdash e:T_2$ and $\GS v:T_1$, then $\GS e[x/v]:T_2$.
\end{lemma}
\begin{proof}
  By induction on $\Gamma{,}x{:}T_1{;}\Sigma\vdash e:T_2$. Only interesting cases:
  \begin{description}
  \item[Case] \textsc{TVar}: \[
    \inferrule{(\Gamma{,}x{:}T_1)(y)=T_2}{\Gamma{,}x{:}T_1{;}\Sigma\vdash y:T_2}
    \]
    If $y\neq x$, then $\Gamma(x)=T_2$. Then by \textsc{TVar}, $\GS y:T_2$.\\
    If $y=x$, then $y[x/v]=v$ and $(\Gamma{,}x{:}T_1)(x)=T_1$, so $T_1 = T_2$. Have immediately that $\GS v:T_2$.
  \item[Case] \textsc{TAbs}: \[
    \inferrule{\Gamma{,}x{:}T_1{,}y{:}T_1'{;}\Sigma\vdash e:T_2'} 
    {\Gamma{,}x{:}T_1{;}\Sigma\vdash \fnx{y}{e}:\funt{T_1'}{T_2'}}
    \]
    If $y=x$, then $(\fnx{y}{e})[x/v]=\fnx{y}{e}$ and $\Gamma{,}x{:}T_1{,}y{:}T_1'\equiv \Gamma{,}y{:}T_1'$.\\
    Therefore $\Gamma{,}y{:}T_1'{;}\Sigma\vdash e:T_2'$.\\
    By \textsc{TAbs}, $\GS \fnx{y}{e}:\funt{T_1'}{T_2'}$.\\
    If $y\neq x$, then $\Gamma{,}x{:}T_1{,}y{:}T_1'\equiv \Gamma{,}y{:}T_1'{,}x{:}T_1$.\\
    Therefore $\Gamma{,}y{:}T_1'{,}x{:}T_1{;}\Sigma\vdash e:T_2'$.\\
    By the IH, $\Gamma{,}y{:}T_1'\Sigma\vdash e[x/v]:T_2'$.\\
    By \textsc{TAbs}, $\GS \fnx{y}{e[x/v]}:\funt{T_1'}{T_2'}$.
  \item[Case] \textsc{TLet}:\[
    \inferrule{\Gamma{,}x{:}T_1{;}\Sigma\vdash e_1:T' \\ \Gamma{,}x{:}T_1{,}y{:}T'{;}\Sigma\vdash e_2:T_2}
    {\Gamma{,}x{:}T_1{;}\Sigma\vdash \letx{y}{e_1}{e_2}:T_2}
    \]
    By the IH, $\GS e_1[x/v]:T'$.\\
    If $y=x$, then $(\letx{y}{e_1}{e_2})[x/v]=\letx{y}{e_1[x/v]}{e_2}$ and $\Gamma{,}x{:}T_1{,}y{:}T'\equiv \Gamma{,}y{:}T'$.\\
    Therefore $\Gamma{,}y{:}T'{;}\Sigma\vdash e_2:T_2$.\\
    By \textsc{TLet}, $\GS \letx{y}{e_1[x/v]}{e_2}:T_2$.\\
    If $y\neq x$, then $\Gamma{,}x{:}T_1{,}y{:}T'\equiv \Gamma{,}y{:}T'{,}x{:}T_1$.\\
    Therefore $\Gamma{,}y{:}T'{,}x{:}T_1{;}\Sigma\vdash e_2:T_2$.\\
    By the IH, $\Gamma{,}y{:}T'\Sigma\vdash e_2[x/v]:T_2$.\\
    By \textsc{TLet}, $\GS \letx{y}{e_1[x/v]}{e_2[x/v]}:T_2$.
  \end{description}
\end{proof}

\begin{lemma}\label{cp:lem:apx:nothastype}
  If $\emptyset;\Sigma\vdash v:T$ and $\lowlty{T} \not\preceq S$ and $S \not\preceq \lowlty{T}$, then $\neq\hastypenb{v}{S}$.
\end{lemma}
\begin{proof}
  By cases on $T$.
  \begin{description}
  \item[Case] $T=\intt$: By Lemma \ref{cp:lem:apx:canonical}, $v=n$. Have that $S \not\in\{\intt,\dyn\}$, so $\neg\hastypenb{n}{S}$.
  \item[Case] $T=T_1\to T_2$: By Lemma \ref{cp:lem:apx:canonical}, $v=\lambda x.~e$. Have that $S \not\in\{\funs,\dyn\}$, so $\neg\hastypenb{\lambda x.~e}{S}$.
  \item[Case] $T=\reft{T'}$: By Lemma \ref{cp:lem:apx:canonical}, $v=a$. Have that $S \not\in\{\refs,\dyn\}$, so $\neg\hastypenb{a}{S}$.
  \item[Case] $T=\dyn$: Vacuous.
  \end{description}
\end{proof}

\begin{lemma}[Preservation]\label{cp:lem:apx:preservation}
  Suppose $\emptyset{;}\Sigma\vdash e:T$ and $\Sigma\vdash\mu$. If
  $\cfg{e}{\mu}\steps\cfg{e'}{\mu'}$, then $\emptyset{;}\Sigma'\vdash
  e':T$ and $\Sigma'\vdash\mu'$ and $\Sigma' \sqsubseteq \Sigma$.
\end{lemma}
\begin{proof}
  By induction on $\cfg{e}{\mu}\steps\cfg{e'}{\mu'}$.
  \begin{description}
  \item[Case ] \textsc{ECheck}:\[
    \cfg{\checkx{v}{S}}{\mu}  \steps \cfg{v}{\mu}\qquad\text{if }\hastypenb{v}{S}
    \]
    By Lemma \ref{cp:lem:inversion}, we have three cases:
    \begin{description}
    \item[Case] $\emptyset;\Sigma\vdash v:T'$ and $\lowlty{T'} \preceq S$ and $\vdash T' <: T$:\\
      Immediate by subsumption.
    \item[Case] $\emptyset;\Sigma\vdash v:T''$ and $\lowlty{T'}\not\preceq S$ and $T''\neq\dyn$:\\
      Vacuous, by Lemma \ref{cp:lem:apx:nothastype}.
    \item[Case] $\emptyset{;}\Sigma\vdash \checkx{v}{S}:\uplty{S}$ and $\emptyset{;}\Sigma\vdash v:\dyn$.\\
    We proceed by cases on $S$:
    \begin{description}
    \item[Subcase ] $S=\dyn$: Then $\uplty{S}=\dyn$, and the theorem holds.
    \item[Subcase ] $S=\intt$: Then $\uplty{S}=\intt$. Because $\hastypenb{v}{\intt}$, $v=n$. Thus by \textsc{TInt}, $\emptyset{;}\Sigma\vdash v:\intt$.
    \item[Subcase ] $S=\to$: Then
      $\uplty{S}=\funt{\dyn}{\dyn}$.\\
      Because $\hastypenb{v}{\to}$, $v=\fnx{x}{e}$. \\
      By Lemma \ref{cp:lem:inversion}, for some $T_1,T_2$ have
      $\emptyset{,}x{:}T_1{;}\Sigma\vdash e:T_2$ and $\vdash \funt{T_1}{T_2} <: \dyn$.\\
      Then have that $\vdash T_2 <: \dyn$ and $\vdash \dyn <: T_1$, and hence $T_1=\dyn$.\\
      By \textsc{TSubsump}, $\emptyset{,}x{:}\dyn{;}\Sigma\vdash e:\dyn$.\\
      Therefore $\emptyset{;}\Sigma\vdash \fnx{x}{e}:\funt{\dyn}{\dyn}$.
    \item[Subcase ] $S=\refs$: Then
      $\uplty{S}=\reft{\dyn}$.\\
      Because $\hastypenb{v}{\to}$, $v=a$. \\
      By Lemma \ref{cp:lem:inversion}, for some $T'$ have
      $\Sigma(a)=T'$ and $\vdash \reft{T'} <: \dyn$.\\
      Then have that $\vdash T' <: \dyn$ and $\vdash \dyn <: T'$, and hence $T'=\dyn$.\\
      Therefore by \textsc{TAddr} $\emptyset{;}\Sigma\vdash a:\reft{\dyn}$.
    \end{description}
    \end{description}
  \item[Case ] \textsc{ECheckFail}: Vacuous.
  \item[Case ] \textsc{EFail}: Vacuous.
  \item[Case ] \textsc{ERef}:\[
    \cfg{\refx{v}}{\mu}  \steps \cfg{a}{\mu[a:=v]}\qquad\text{where}~a~\text{fresh}
    \]
    By Lemma \ref{cp:lem:inversion}, for some $T'$ have $\ES v:T'$ and $\vdash \reft{T'} <: T$.\\
    Let $\Sigma'=\Sigma{,}a{:}T'$.\\
    Then by \textsc{TAddr}, $\emptyset{;}\Sigma'\vdash a:\reft{T'}$.\\
    By \textsc{TSubsump}, $\emptyset{;}\Sigma'\vdash a:T$.\\
    Since $a$ is fresh,  $\Sigma'\ledyn\Sigma$.\\
    By Lemma \ref{cp:lem:apx:heapweak}, $\emptyset{;}\Sigma'\vdash v:T'$.\\
    By Lemma \ref{cp:lem:heapext}, $\Sigma'\vdash\mu[x:=v]$.
  \item[Case ] \textsc{EApp}:\[
    \cfg{((\fnx{x}{e})~v)^w}{\mu}  \steps  \cfg{e[x/v]}{\mu}
    \]
    We proceed by cases on $w$.
    \begin{description}
    \item[Subcase ] $w=\tyorig$:\hfill\\
      By Lemma \ref{cp:lem:inversion}, for some $T_1, T_2$ have $\ES \fnx{x}{e}:\funt{T_1}{T_2}$ and $\ES v:T_1$ and $\vdash T_2 <: T$.\\
      By Lemma \ref{cp:lem:inversion}, for some $T_1', T_2'$ have $\emptyset{,}x{:}T_1'{;}\Sigma\vdash e:T_2'$ and $\vdash \funt{T_1'}{T_2'} <: \funt{T_1}{T_2}$.\\
      Hence $\vdash T_1 <: T_1'$ and $\vdash T_2' <: T_2$.\\
      By \textsc{TSubsump}, $\ES v:T_1'$.\\
      By Lemma \ref{cp:lem:apx:subst}, $\ES e[x/v]:T_2'$.\\
      By \textsc{TSubsump}, $\ES e[x/v]:T_2$, and by \textsc{TSubsump} again $\ES e[x/v]:T$. 
    \item[Subcase ] $w=\unorig$:\hfill\\
      By Lemma \ref{cp:lem:inversion}, have $\ES \fnx{x}{e}:\dyn$ and $\ES v:\dyn$ and $\vdash \dyn <: T$.\\
      Therefore $\dyn=T$.\\
      By Lemma \ref{cp:lem:inversion}, for some $T_1, T_2$ have $\emptyset{,}x{:}T_1{;}\Sigma\vdash e:T_2$ and $\vdash \funt{T_1}{T_2} <: \dyn$.\\
      Hence $\vdash \dyn <: T_1$, and therefore $T_1=\dyn$, and $\vdash T_2 <: \dyn$.\\
      By Lemma \ref{cp:lem:apx:subst}, $\ES e[x/v]:T_2$.\\
      By \textsc{TSubsump}, $\ES e[x/v]:\dyn$. 
    \end{description}
  \item[Case ] \textsc{ELet}:\[
    \cfg{\letx{x}{v}{e}}{\mu}  \steps  \cfg{e[x/v]}{\mu}
    \]
    By Lemma \ref{cp:lem:inversion}, for some $T_1, T_2$ have $\ES v:T_1$ and $\emptyset{,}x{:}T_1{;}\Sigma\vdash e:T_2$ and $\vdash T_2 <: T$.\\
    By Lemma \ref{cp:lem:apx:subst}, $\ES e[x/v]:T_2$.\\
    By \textsc{TSubsump}, $\ES e[x/v]:T$. 
  \item[Case ] \textsc{EDeref}:\[
    \cfg{\derefxp{a}{w}}{\mu}  \steps  \cfg{\mu(a)}{\mu}
    \]
    We proceed by cases on $w$.
    \begin{description}
    \item[Subcase ] $w=\tyorig$:\hfill\\
      By Lemma \ref{cp:lem:inversion}, for some $T'$ have $\ES a:\reft{T'}$ and $\vdash T' <: T$.\\
      By Lemma \ref{cp:lem:inversion}, $\Sigma(a)=T''$ and $\vdash \reft{T''} <: \reft{T'}$.\\
      Hence $\vdash T' <: T''$ and $\vdash T'' <: T'$.\\
      Since $\Sigma\vdash\mu$, $\ES \mu(a):T''$.
      By \textsc{TSubsump}, $\ES \mu(a):T'$, and by \textsc{TSubsump} again $\ES \mu(a):T$. 
    \item[Subcase ] $w=\unorig$:\hfill\\
      By Lemma \ref{cp:lem:inversion}, have $\ES a:\dyn$ and $\vdash \dyn <: T$.\\
      Therefore $\dyn=T$.\\
      By Lemma \ref{cp:lem:inversion}, $\Sigma(a)=T'$ and $\vdash \reft{T'} <: \dyn$.\\
      Hence $\vdash T' <: \dyn$ and $\vdash \dyn <: T'$, so $T'=\dyn$.\\
      Since $\Sigma\vdash\mu$, $\ES \mu(a):\dyn$.
    \end{description}
  \item[Case ] \textsc{EUpdt}:\[
    \cfg{\mutxpINV{a}{v}{w}}{\mu}  \steps  \cfg{0}{\mu[a:=v]}
    \]
    By \textsc{TInt}, $\ES 0:\intt$.\\
    We continue by cases on $w$.
    \begin{description}
    \item[Subcase ] $w=\tyorig$:\hfill\\
      By Lemma \ref{cp:lem:inversion}, for some $T'$ have $\ES a:\reft{T'}$ and $\ES v:T'$ and $\vdash \intt <: T$.\\
      By \textsc{TSubsump}, $\ES 0:T$.
      By Lemma \ref{cp:lem:inversion}, $\Sigma(a)=T''$ and $\vdash \reft{T''} <: \reft{T'}$.\\
      Hence $\vdash T' <: T''$ and $\vdash T'' <: T'$.\\
      By \textsc{TSubsump}, $\ES v:T''$.\\
      Therefore $\Sigma\vdash\mu[a:=v]$.
    \item[Subcase ] $w=\unorig$:\hfill\\
      By Lemma \ref{cp:lem:inversion}, have $\ES a:\dyn$ and $\ES v:\dyn$ and $\vdash \intt <: T$.\\
      By \textsc{TSubsump}, $\ES 0:T$.
      By Lemma \ref{cp:lem:inversion}, $\Sigma(a)=T'$ and $\vdash \reft{T'} <: \dyn$.\\
      Hence $\vdash T' <: \dyn$ and $\vdash \dyn <: T'$, so $T'=\dyn$.\\
      Therefore $\Sigma\vdash\mu[a:=v]$.
    \end{description}
  \item[Case ] \textsc{EAdd}:\[
    \cfg{n_1~+^w~n_2}{\mu}  \steps  \cfg{n'}{\mu} \qquad \text{where}~n_1+n_2=n'
    \]
    By Lemma \ref{cp:lem:inversion} (applying once for each case on $w$), $\vdash \intt <: T$.\\
    Immediately have $\ES n':\intt$.\\
    By \textsc{TSubsump}, $\ES n':T$.
  \end{description}
\end{proof}

\begin{lemma}[Multi-step preservation]\label{cp:lem:apx:multipres}
  Suppose $\emptyset{;}\Sigma\vdash e:T$ and $\Sigma\vdash\mu$. If
  $\cfg{e}{\mu}\steps*\cfg{e'}{\mu'}$, then $\emptyset{;}\Sigma'\vdash
  e':T$ and $\Sigma'\vdash\mu'$ and $\Sigma' \sqsubseteq \Sigma$.
\end{lemma}
\begin{proof}
  By induction on $\cfg{e}{\mu}\steps*\cfg{e'}{\mu'}$.
  \begin{description}
  \item[Cases] where $\cfg{e}{\mu}\steps^*\mathtt{fail}$ are vacuous.
  \item[Case] $\cfg{e}{\mu}\steps^*\cfg{e}{\mu}$: Immediate.
  \item[Case] \[\inferrule{\cfg{e}{\mu}\steps\cfg{e'}{\mu'} \\ \cfg{e'}{\mu'}\steps^*\cfg{e''}{\mu''}}{\cfg{e}{\mu}\steps^*\cfg{e''}{\mu''}}\]
      By Lemma \ref{cp:lem:apx:preservation}, $\emptyset;\Sigma'\vdash e':T$ and $\Sigma'\vdash\mu'$ and $\Sigma' \sqsubseteq \Sigma$.\\
      By the IH, $\emptyset;\Sigma''\vdash e'':T$ and $\Sigma''\vdash\mu''$ and $\Sigma'' \sqsubseteq \Sigma'$.\\
      By transitivity of equality on types, $\Sigma''\ledyn \Sigma$.
  \end{description}
\end{proof}

\begin{lemma}[Progress]\label{cp:lem:apx:progress}
  Suppose $\emptyset{;}\Sigma\vdash e:T$ and $\Sigma\vdash\mu$. Then
  either 
  \begin{itemize}
  \item $\cfg{e}{\mu}\steps\cfg{e'}{\mu'}$,
  \item $\cfg{e}{\mu}\steps\mathtt{fail}$,
  \item $e$ is a value, or
  \item $\stucknb{e}{\mu}{\unorig}$.
  \end{itemize}
\end{lemma}
\begin{proof}
  By induction on $\emptyset{;}\Sigma\vdash e:T$.\\
  For each case, if there exists some $E, e'$ with $e'$ not a value
  such that $e=E[e']$, then by the IH, $e'$ is either a value, it
  steps to \texttt{fail} or another expression, or it is an error
  blaming $\unorig$. If it steps to an expression, then $e$ steps to
  an expression by \textsc{EStep}. If it steps to \texttt{fail}, then
  $e$ steps to \texttt{fail} by \textsc{EFail}. If
  $\stucknb{e'}{\mu}{\unorig}$, then
  $\stucknb{e}{\mu}{\unorig}$.\\
  We now proceed to each case, assuming that no such $E, e'$ exists.\\
  \begin{description}
  \item[Cases ] \textsc{TAbs}, \textsc{TAddr}, and \textsc{TInt}: Immediately have $e=v$.
  \item[Case ] \textsc{TVar}: Vacuous.
  \item[Case ] \textsc{TSubsump}: with $T=T_2$, \[
    \inferrule{\emptyset{;}\Sigma\vdash e:T_1 \\ \vdash T_1 <: T_2}{\emptyset{;}\Sigma\vdash e:T_2}
    \]
    Immediate from the IH.
  \item[Case ] \textsc{TCheck}: with $T=\uplty{S}$, \[
    \inferrule{\emptyset{;}\Sigma\vdash e':\dyn}
    {\emptyset{;}\Sigma\vdash\checkx{e'}{S}:\uplty{S}}
    \]
    Assume $e'=v$. Then either $\hastypenb{v}{S}$ or $\neg\hastypenb{v}{S}$.\\
    In the former case, by \textsc{ECheck},
    $\cfg{\checkx{v}{S}}{\mu}\steps\cfg{v}{\mu}$.\\
    Otherwise, by \textsc{ECheckFail}, 
    $\cfg{\checkx{v}{S}}{\mu}\steps\mathtt{fail}$.
  \item[Case ] \textsc{TRedundantCheck}: \[
    \inferrule{\Gamma\vdash e' : T \\ \lowlty{T} \preceq S}
    {\Gamma\vdash \checkx{e'}{S} : T}
    \]
    Assume $e'=v$. Then either $\hastypenb{v}{S}$ or $\neg\hastypenb{v}{S}$.\\
    In the former case, by \textsc{ECheck},
    $\cfg{\checkx{v}{S}}{\mu}\steps\cfg{v}{\mu}$.\\
    Otherwise, by \textsc{ECheckFail}, 
    $\cfg{\checkx{v}{S}}{\mu}\steps\mathtt{fail}$.
  \item[Case ] \textsc{TFailCheck}: \[
\inferrule{\Gamma\vdash e' : T \\ T \neq \dyn \\ \lowlty{T} \not\preceq S}
 {\Gamma\vdash \checkx{e'}{S} : T'}
    \]
    Assume $e'=v$. Then either $\hastypenb{v}{S}$ or $\neg\hastypenb{v}{S}$.\\
    In the former case, by \textsc{ECheck},
    $\cfg{\checkx{v}{S}}{\mu}\steps\cfg{v}{\mu}$.\\
    Otherwise, by \textsc{ECheckFail}, 
    $\cfg{\checkx{v}{S}}{\mu}\steps\mathtt{fail}$.
  \item[Case ] \textsc{TFail}: \[
    \inferrule{ }
    {\Gamma\vdash \mathtt{fail} : T}
    \]
    Immediately by \textsc{EFail}, 
    $\cfg{\mathtt{fail}}{\mu}\steps\mathtt{fail}$.
  \item[Case ] \textsc{TRef}: with $T=\reft{T'}$, \[
    \inferrule{\GS e':T'}{\GS \refx{e'}:\reft{T'}}
    \]
    Assume $e'=v$.\\
    By \textsc{ERef}, with $a$ fresh, $\cfg{\refx{v}}{\mu}\steps\cfg{a}{\mu[a:=v]}$.\\
  \item[Case ] \textsc{TApp}: with $T=T_2$, \[
    \inferrule{\emptyset{;}\Sigma\vdash e_1:\funt{T_1}{T_2} \\ \emptyset{;}\Sigma\vdash e_2:T_1}
    {\emptyset{;}\Sigma\vdash (e_1~e_2)^\tyorig:T_2}
    \]
    Assume $e_1=v_1$ and $e_2=v_2$.\\
    By Lemma \ref{cp:lem:apx:canonical}, $v_1=\fnx{x}{e_1'}$.\\
    By \textsc{EApp}, $\cfg{((\fnx{x}{e_1'})~v_2)^\tyorig}{\mu}\steps\cfg{e_1'[x/v_2]}{\mu}$.
  \item[Case ] \textsc{TAppOW}: with $T=\dyn$, \[
    \inferrule{\emptyset{;}\Sigma\vdash e_1:\dyn \\ \emptyset{;}\Sigma\vdash e_2:\dyn}
    {\emptyset{;}\Sigma\vdash (e_1~e_2)^\unorig:\dyn}
    \]
    Assume $e_1=v_1$ and $e_2=v_2$.\\
    Have that $v_1 \in \{a, n, \fnx{x}{e_1'}\}$.\\
    If $v_1=a$ or $v_1=n$, $\stucknb{(v_1~v_2)^\unorig}{\mu}{\unorig}$.\\
    Otherwise, $v_1=\fnx{x}{e_1'}$ and $\cfg{((\fnx{x}{e_1'})~v_2)^\unorig}{\mu}\steps\cfg{e_1'[x/v_2]}{\mu}$.
  \item[Case ] \textsc{TLet}: with $T=T_2$, \[
    \inferrule{\GS e_1:T_1 \\ \Gamma,x{:}T_1{;}\Sigma\vdash e_2:T_2}
    {\GS \letx{x}{e_1}{e_2}:T_2}
    \]
    Assume $e_1=v$.\\
    By \textsc{ELet}, $\cfg{\letx{x}{v}{e_2}}{\mu}\steps\cfg{e_2[x/v]}{\mu}$.
  \item[Case ] \textsc{TAdd}: with $T=\intt$, \[
    \inferrule{\GS e_1:\intt \\ \GS e_2:\intt}
    {\GS e_1~+^\tyorig~e_2:\intt}
    \]
    Assume $e_1=v_1$ and $e_2=v_2$.\\
    By Lemma \ref{cp:lem:apx:canonical}, $v_1=n_1$ and $v_2=n_2$.\\
    By \textsc{EAdd}, with $n'=n_1+n_2$, $\cfg{n_1~+^\tyorig~n_2}{\mu}\steps\cfg{n'}{\mu}$.
  \item[Case ] \textsc{TAddOW}: with $T=\intt$, \[
    \inferrule{\GS e_1:\dyn \\ \GS e_2:\dyn}
    {\GS e_1~+^\unorig~e_2:\intt}
    \]
    Assume $e_1=v_1$ and $e_2=v_2$.\\
    Have that $v_1 \in \{a, n_1, \fnx{x}{e_1'}\}$.\\
    If $v_1=a$ or $v_1=\fnx{x}{e}$, $\stucknb{e_1~+^\unorig~e_2}{\mu}{\unorig}$.\\
    Otherwise $v_1=n_1$.\\
    By the same reasoning, either $\stucknb{e_1~+^\unorig~e_2}{\mu}{\unorig}$ or $v_2=n_2$.\\
    By \textsc{EAdd}, with $n'=n_1+n_2$, $\cfg{n_1~+^\tyorig~n_2}{\mu}\steps\cfg{n'}{\mu}$.
  \item[Case ] \textsc{TDeref}: \[
    \inferrule{\GS e':\reft{T}}{\GS \derefxp{e'}{\tyorig}:T}
    \]
    Assume $e'=v$.\\
    By Lemma \ref{cp:lem:apx:canonical}, $v=a$ and $\mu(a)=v'$.\\
    By \textsc{EDeref}, $\cfg{\derefxp{a}{\tyorig}}{\mu}\steps\cfg{\mu(a)}{\mu}$.
  \item[Case ] \textsc{TDerefOW}: with $T=\dyn$, \[
    \inferrule{\GS e':\dyn}{\GS \derefxp{e'}{\unorig}:\dyn}
    \]
    Assume $e'=v$.\\
    Have that $v \in \{a, n, \fnx{x}{e_1'}\}$.\\
    If $v=n$ or $v=\fnx{x}{e}$, $\stucknb{\derefxp{e'}{\unorig}}{\mu}{\unorig}$.\\
    Otherwise $v=a$.\\
    If $a\not\in\dom{\mu}$, then $\stucknb{\derefxp{a}{\unorig}}{\mu}{\unorig}$.\\
    Otherwise, by \textsc{EDeref}, $\cfg{\derefxp{a}{\tyorig}}{\mu}\steps\cfg{\mu(a)}{\mu}$.
  \item[Case ] \textsc{TUpdt}: with $T=\intt$, \[
    \inferrule{\GS e_1:\reft{T} \\ \GS e_2:T}
    {\GS \mutxpINV{e_1}{e_2}{\tyorig}:\intt}
    \]
    Assume $e_1=v_1$ and $e_2=v_2$.\\
    By Lemma \ref{cp:lem:apx:canonical}, $v_1=a$.\\
    By \textsc{EUpdt}, $\cfg{\mutxpINV{a}{v_2}{\tyorig}}{\mu}\steps\cfg{0}{\mu[a:=v_2]}$.
  \item[Case ] \textsc{TDerefOW}: with $T=\intt$, \[
    \inferrule{\GS e_1:\dyn \\ \GS e_2:\dyn}
    {\GS \mutxpINV{e_1}{e_2}{\unorig}:\intt}
    \]
    Assume $e_1=v_1$ and $e_2=v_2$.\\
    Have that $v_1 \in \{a, n, \fnx{x}{e_1'}\}$.\\
    If $v_1=n$ or $v_1=\fnx{x}{e}$, $\stucknb{\mutxpINV{v_1}{v_2}{\unorig}}{\mu}{\unorig}$.\\
    Otherwise $v_1=a$.\\
    If $a\not\in\dom{\mu}$, then $\stucknb{\mutxpINV{a}{v_2}{\unorig}}{\mu}{\unorig}$.\\
    Otherwise, by \textsc{EDeref}, $\cfg{\mutxpINV{a}{v_2}{\unorig}}{\mu}\steps\cfg{0}{\mu[a:=v_2]}$.
  \end{description}
\end{proof}

\begin{corollary}
  If $\emptyset\vdash d:A;\Omega$ and $\sigma$ is a solution for
  $\Omega$, then:
  \begin{itemize}
  \item $\emptyset \vdash \sigma d\leadsto e:T$, and 
  \item $\emptyset;\emptyset \vdash e:T$, and 
  \item if $\cfg{e}{\emptyset} \steps^* \cfg{v}{\mu}$, then $\Sigma\vdash\mu$ and $\emptyset;\Sigma\vdash v:T$ and $\vdash T <: \sigma A$.
  \end{itemize}
\end{corollary}
\begin{proof}
  By Lemma \ref{cp:lem:soundsub}, $\emptyset\vdash \sigma d: T$ and $\vdash T <: \sigma A$.\\
  By Lemma \ref{cp:lem:dtoe}, $\emptyset\vdash \sigma d\leadsto e:T$.\\
  If for all $\varsigma$ such that
  $\cfg{e}{\emptyset}\steps\varsigma$, there exists some $\varsigma'$
  such that $\varsigma\steps\varsigma'$, then the theorem is
  satisfied.\\
  Otherwise, there exists some $\varsigma$ such that
  $\varsigma\not\steps\varsigma'$.\\
  If $\varsigma=\mathtt{fail}$, then the theorem is satisfied.\\
  Otherwise, $\varsigma=\cfg{e'}{\mu}$.\\
  By repeating Lemma
  \ref{cp:lem:apx:preservation}, $\emptyset;\Sigma\vdash e':T$ and
  $\Sigma\vdash\mu$.\\
  By Lemma \ref{cp:lem:apx:progress}, either $e'=v$ or
  $\stucknb{e}{\mu}{\unorig}$.
\end{proof}

\subsection{Analysis correctness}\label{cp:apx:proof:correct}

\begin{lemma}\label{cp:lem:apx:soundcurry}
  If $\Gamma\vdash d:A;\Omega$ and $\sigma$ is a solution for $\Omega$, then $\sigma\Gamma;\emptyset\vdash |d|: \sigma A$.
\end{lemma}
\begin{proof}
  By induction on $\Gamma\vdash d:A;\Omega$.
  \begin{description}
  \item[Case] \textsc{IVar}: \[
    \inferrule{\Gamma(x)=A}{\Gamma\vdash x:A;\emptyset}
    \]
    Have that $\sigma\Gamma(x)=\sigma A$.\\
    By \textsc{TVar}, $\sigma\Gamma;\emptyset\vdash x:\sigma A$.
  \item[Case] \textsc{IAbs}: \[
    \inferrule{\Gamma,x{:}\alpha\vdash d:A\yO}
    {\Gamma\vdash \fntx{x}{\alpha}{\beta}{d}:\funt{\alpha}{\beta}\yO,\flows{A}{\beta}}  
    \]
    Since $\sigma$ is a solution for $\Omega\cup\{\flows{A}{\beta}\}$, by Lemma \ref{cp:lem:apx:solsplit}, $\sigma$ is a solution for $\Omega$.\\
    By the IH, $\sigma\Gamma,x{:}\sigma \alpha;\emptyset\vdash |d|: \sigma A$.\\
    Since $\sigma$ is a solution to $\Omega\cup\{\flows{A}{Y}\}$, $\vdash \sigma A <: \sigma \beta$.\\
    By \textsc{TSubsump}, $\sigma\Gamma,x{:}\sigma \alpha;\emptyset\vdash |d|: \sigma \beta$.\\
    By \textsc{TAbs}, $\sigma\Gamma;\emptyset\vdash \fnx{x}{|d|}:\sigma \alpha\to \sigma \beta$.
  \item[Case] \textsc{IApp}: \[
    \inferrule{\Gamma\vdash d_1:\funt{V_1}{V_2}\yO_1 \\ 
      \Gamma\vdash d_2:A\yO_2}
    {\Gamma\vdash d_1~d_2:V_2\yO_1,\Omega_2,\flows{A}{V_1}}
    \]
    Since $\sigma$ is a solution for $\Omega_1\cup\Omega_2\cup\{\flows{A}{V_1}\}$, by Lemma \ref{cp:lem:apx:solsplit}, $\sigma$ is a solution for $\Omega_1$ and $\sigma$ is a solution for $\Omega_2$ and $\sigma$ is a solution for $\flows{A}{V_1}$.\\
    By the IH, $\sigma\Gamma;\emptyset\vdash|d_1|: \funt{\sigma V_1}{\sigma V_2}$.\\
    By the IH, $\sigma\Gamma;\emptyset\vdash |d_2|: \sigma A$.\\
    Since $\sigma$ is a solution to $\{\flows{A}{V_1}\}$, $\vdash \sigma A <: \sigma V_1$.\\
    By \textsc{TSubsump}, $\sigma\Gamma;\emptyset\vdash |d_2|: \sigma V_1$.\\
    By \textsc{TApp}, $\sigma\Gamma;\emptyset\vdash (|d_1|~|d_2|)^\tyorig: \sigma V_2$.
  \item[Case] \textsc{ICheck}: \[
    \inferrule{\Gamma\vdash d:A_1\yO \\ \specmatches{A_1}{A_2}{S}}
    {\Gamma\vdash \checkx{d}{S}:A_2\yO}
    \]
    By the IH, $\sigma\Gamma;\emptyset\vdash |d|: \sigma A_1$.\\
    If $S=\dyn$ then $A_2=A_1$, and by \textsc{TRedundantCheck}, $\sigma\Gamma;\emptyset\vdash \checkx{|d|}{S}:\sigma A_2$.\\
    Otherwise, we proceed by cases on $A_1$ and $\sigma A_1$.
    \begin{description}
    \item[Case] $A_1\neq V$ and $S$ is the constructor of $\sigma A_1$:\\
      Then $S$ is the constructor of $A_1$.\\
      Then $A_2=A_1$.\\
      By \textsc{TRedundantCheck}, $\sigma\Gamma;\emptyset\vdash \checkx{|d|}{S}:\sigma A_2$.
    \item[Case] $A_1\neq V$ and $S$ is not the constructor of $\sigma A_1$: Vacuous, since $S\neq\dyn$.
    \item[Case] $A_1=V$ and $S$ is the constructor of $\sigma A_1$:\\
      Then $\sigma A_1=\sigma A_2$.\\
      By \textsc{TRedundantCheck}, $\sigma\Gamma;\emptyset\vdash \checkx{|d|}{S}: \sigma A_2$.
    \item[Case] $A_1=V$ and $S$ is not the constructor of $\sigma A_1$:\\
      Since $\sigma$ is a solution to $\Omega$, for all $\alpha\in\mathit{parts}(A_2)$ of
      $A_2$, $\sigma \alpha=\dyn$.\\
      Therefore $\uplty{S}=\sigma A_2$.\\
      Cases on $\vdash \sigma A_1 <: \dyn$:
      \begin{description}
      \item[Subcase] $\vdash \sigma A_1 <: \dyn$:\\
        By \textsc{TSubsump}, $\sigma\Gamma;\emptyset\vdash |d|:\dyn$.\\
        By \textsc{TCheck}, $\sigma\Gamma;\emptyset\vdash \checkx{|d|}{S}:\uplty{S}$.
      \item[Subcase] $\vdash \sigma A_1 \not<: \dyn$:\\
        Since $S$ is not the constructor of $\sigma A_1$ and $S\neq\dyn$, $\lowlty{\sigma A_1} \not\preceq S$.\\
        By \textsc{TFailCheck}, $\sigma\Gamma;\emptyset\vdash \checkx{|d|}{S}:\uplty{S}$.
      \end{description}
    \end{description}
  \item[Case] \textsc{IRef}: \[
    \inferrule{\Gamma\vdash d:A\yO}{\Gamma\vdash \refxi{d}{X}:\reft{X}\yO,\flows{A}{X}}\
    \]
    By the IH, $\sigma\Gamma;\emptyset\vdash |d|: \sigma A$.\\
    Since $\sigma$ is a solution to $\Omega,\flows{A}{X}$, $\vdash \sigma A <: \sigma X$.\\
    By \textsc{TSubsump}, $\sigma\Gamma;\emptyset\vdash |d|: \sigma X$.\\
    By \textsc{TRef}, $\sigma\Gamma;\emptyset\vdash \refx{|d|}: \reft{\sigma X}$.
  \item[Case] \textsc{IDeref}: \[
    \inferrule{\Gamma\vdash d:\reft{V}\yO}
    {\Gamma\vdash\derefx{d}:V\yO}
    \]
    By the IH, $\sigma\Gamma;\emptyset\vdash |d|:  \reft{\sigma V}$.\\
    By \textsc{TDeref}, $\sigma\Gamma;\emptyset\vdash \derefxp{|d|}{\tyorig}:\sigma V$.
  \item[Case] \textsc{IUpdt}: \[
    \inferrule{\Gamma\vdash d_1:\reft{V}\yO_1 \\
      \Gamma\vdash d_2:A\yO_2}
    {\Gamma\vdash\mutx{d_1}{d_2}:\intt\yO_1,\Omega_2,\flows{A}{V}}
    \]
    Since $\sigma$ is a solution to $\Omega_1\cup\Omega_2\cup\{\flows{A}{V}\}$, by Lemma \ref{cp:lem:apx:solsplit}, $\sigma$ is a solution to $\Omega_1$ and $\sigma$ is a solution to $\Omega_2$ and $\sigma$ is a solution to $\{\flows{A}{V}\}$.\\
    By the IH, $\sigma\Gamma;\emptyset\vdash|d_1|: \reft{\sigma V}$.\\
    By the IH, $\sigma\Gamma;\emptyset\vdash|d_2|: \sigma A$.\\
    Since $\sigma$ is a solution to $\{\flows{A}{V}\}$, $\vdash \sigma A <: \sigma V$.\\
    By \textsc{TSubsump}, $\sigma\Gamma;\emptyset\vdash|d_2|: \sigma V$.\\
    By \textsc{TUpdt}, $\sigma\Gamma;\emptyset\vdash \mutxpINV{|d_1|}{|d_2|}{\tyorig}:\intt$.
  \item[Case] \textsc{IAdd}: \[
    \inferrule{\Gamma\vdash d_1:\intt\yO_1 \\
      \Gamma\vdash d_2:\intt\yO_2}
    {\Gamma\vdash d_1~+~d_2:\intt\yO_1,\Omega_2}
    \]
    Since $\sigma$ is a solution to $\Omega_1\cup\Omega_2$, by Lemma \ref{cp:lem:apx:solsplit}, $\sigma$ is a solution to $\Omega_1$ and $\sigma$ is a solution to $\Omega_2$.\\
    By the IH, $\sigma\Gamma;\emptyset\vdash|d_1|: \intt$.\\
    By the IH, $\sigma\Gamma;\emptyset\vdash|d_2|: \intt$.\\
    Therefore by \textsc{TAdd}, $\sigma\Gamma;\emptyset\vdash|d_1|~+~|d_2|: \intt$.
  \item[Case ] \textsc{IInt}: Immediate by \textsc{TInt}.
  \end{description}
\end{proof}

\begin{lemma}\label{cp:lem:apx:eweakening}
  If $\Gamma;\Sigma\vdash e:T$ and for all $x\in\mathit{dom}(\Gamma)$, $\Gamma'(x)=\Gamma(x)$, then $\Gamma';\Sigma\vdash e:T$.
\end{lemma}
\begin{proof}
  Induction on $\Gamma;\Sigma\vdash e:T$. The only interesting cases
  are variables, where the correspondence between $\Gamma$ and
  $\Gamma'$ ensure that the result is the same, and functions, where we can immediately show that for all $x\in\mathit{dom}(\Gamma,y:T)$, $(\Gamma',y:T)(x)=(\Gamma,y:T)(x)$.
\end{proof}

\begin{lemma}\label{cp:lem:apx:masssubst}
  If $\Gamma;\Sigma\vdash e:T$ and $\Sigma\vdash \rho:\Gamma$, then $\emptyset;\Sigma\vdash\rho(e):T$.
\end{lemma}
\begin{proof}
  By induction on $\Sigma\vdash \rho:\Gamma$.
  \begin{description}
  \item[Case ] \[
    \inferrule{ }{\Sigma\vdash \cdot:\emptyset}\]
    Have that $\rho(e)=e$.\\
    Immediately $\emptyset;\Sigma\vdash e:T$.
  \item[Case ]\[
    \inferrule{\emptyset;\Sigma\vdash v:T' \\ \Sigma\vdash \rho:\Gamma}{\Sigma\vdash \rho,x=v:\Gamma,x:T'}
    \]
    Have $\Gamma,x:T;\Sigma\vdash e:T$ and $\emptyset,\Sigma\vdash v:T'$.\\
    By Lemma \ref{cp:lem:apx:eweakening}, $\Gamma;\Sigma\vdash v:T'$.\\
    By Lemma \ref{cp:lem:apx:subst}, $\Gamma;\Sigma\vdash e[x/v]:T$.\\
    By the IH, $\emptyset;\Sigma\vdash \rho(e[x/v]):T$.\\
    Have that $(\rho,x=v)(e)=\rho(e[x/v])$.
  \end{description}
\end{proof}

\begin{lemma}\label{cp:lem:apx:hastype}
  If $\Gamma;\Sigma\vdash v:T$ and $\Sigma\vdash\mu$, then $\hastypenb{v}{\lowlty{T}}$.
\end{lemma}
\begin{proof}
  By cases on $T$.
  \begin{description}
  \item[Case] $T=\intt$: By Lemma \ref{cp:lem:apx:canonical}, $v=n$. Therefore $\hastypenb{n}{\intt}$.
  \item[Case] $T=T_1\to T_2$: By Lemma \ref{cp:lem:apx:canonical}, $v=\lambda x.~e$. Therefore $\hastypenb{\lambda x.~e}{\funs}$.
  \item[Case] $T=\reft{T'}$: By Lemma \ref{cp:lem:apx:canonical}, $v=a$. Therefore $\hastypenb{a}{\refs}$.
  \item[Case] $T=\dyn$: Immediate.
  \end{description}
\end{proof}

\begin{customthm}{\ref{cp:thm:correct}}
  Suppose $\Gamma\vdash d:A;\Omega$ and $\sigma$ is a solution to
  $\Omega$ and $\Sigma\vdash \rho:\sigma\Gamma$ and $\Sigma\vdash\mu$ and
  $\cfg{\rho(|d|)}{\mu}\steps^*\cfg{v}{\mu'}$. If
  $\lowlty{\sigma A} \preceq S$, then
  $\cfg{\checkx{v}{S}}{\mu'}\not\steps\mathtt{fail}$.
\end{customthm}
\begin{proof}
  By Lemma \ref{cp:lem:apx:soundcurry}, $\sigma\Gamma;\emptyset\vdash |d|:\sigma A$.\\
  By Lemma \ref{cp:lem:apx:heapweak}, $\sigma\Gamma;\Sigma\vdash |d|:\sigma $.\\
  By Lemma \ref{cp:lem:apx:masssubst}, $\emptyset;\Sigma\vdash\rho(|d|):\sigma A$.\\
  By Lemma \ref{cp:lem:apx:multipres}, $\emptyset;\Sigma'\vdash v:\sigma A$ and $\Sigma'\vdash\mu'$.\\
  By Lemma \ref{cp:lem:apx:hastype}, $\hastypenb{v}{\lowlty{\sigma A}}$.\\
  Since $\lowlty{\sigma A} \preceq S$, either $S=\dyn$ or $S=\lowlty{\sigma A}$. In either case, $\hastypenb{v}{\lowlty{\sigma A}}$.\\
  Therefore \textsc{ECheckFail} does not apply, and no other step can be taken from $\cfg{\checkx{v}{S}}{\mu'}$ to \texttt{fail}.
\end{proof}

\subsection{Soundness of constraint solving}\label{cp:apx:proof:solvesound}

Figure \ref{cp:fig:apx:osimp} restates the definition for constraint set simplification from Figure \ref{cp:fig:constsimpl}.

\begin{figure*}
  \begin{minipage}{1.0\linewidth}
    
\boxed{\Omega \steps \Omega}
\begin{align}
 \Omega \cup \{\flows{\funt{V_1}{V_2}}{\dyn}\} & \steps \Omega \cup \{\flows{\dyn}{V_1},~\flows{V_2}{\dyn}\} \label{cp:osimp:fundynsub}\\ 
 \Omega \cup \{\flows{\funt{V_1}{V_2}}{\funt{V_3}{V_4}}\} & \steps \Omega \cup \{\flows{V_3}{V_1},~\flows{V_2}{V_4}\} \label{cp:osimp:funsub}\\ 
 \Omega \cup \{\flows{\reft{V}}{\dyn}\} & \steps \Omega \cup \{V=\dyn\} \label{cp:osimp:refdynsub}\\ 
 \Omega \cup \{\flows{\reft{V_1}}{\reft{V_2}}\} & \steps \Omega \cup \{V_1=V_2\}\label{cp:osimp:refsub}\\
 \Omega \cup \{\flows{V}{V}\} & \steps \Omega \label{cp:osimp:idsub}\\
 \Omega \cup \{\flows{\intt}{\dyn}\} & \steps \Omega \label{cp:osimp:intdynsub}\\
 \Omega \cup \{\funt{V_1}{V_2}=\funt{V_3}{V_4}\} & \steps \Omega \cup \{V_1=V_3,V_2=V_4\} \label{cp:osimp:funeq}\\ 
 \Omega \cup \{\reft{V_1}=\reft{V_2}\} & \steps \Omega \cup \{V_1=V_2\}\label{cp:osimp:refeq}\\
 \Omega \cup \{A=A\} & \steps \Omega \label{cp:osimp:tyeq}\\
 \Omega \cup \{A=\alpha\} & \steps \Omega\cup \{\alpha=A\} \label{cp:osimp:flipeq}\\
  &
  \begin{array}{ll}
    \qquad\text{where} & A\neq\alpha'
  \end{array}\nonumber\\
 \Omega\cup \{\alpha=A\} & \steps \Omega[\alpha/A] \cup \{\defcon{\alpha}{A}\}\label{cp:osimp:substeq}\\
  &
  \begin{array}{ll}
    \qquad\text{where} & \alpha\not\in\mathit{vars}(A)\\
    & (\defcon{\alpha}{B})\not\in\Omega
  \end{array}\nonumber\\
  \Omega\cup \{\depcon{A}{S}{A} \} & \steps \Omega\label{cp:osimp:tydep}\\
  \Omega\cup \{\tagcon{\alpha}{S}, \depcon{\alpha}{S}{A} \} & \steps \Omega\cup \{\alpha = A\}\label{cp:osimp:tagdep}\\
  &
  \begin{array}{ll}
    \qquad\text{where} & (\depcon{\alpha}{S'}{A'}) \not\in \Omega\\
    & A \neq \alpha
  \end{array}\nonumber\\
  \Omega\cup \{\tagcon{\alpha}{S_1}, \depcon{\alpha}{S_2}{A} \} & \steps \Omega\cup \{\tagcon{\alpha}{S_1}\} \cup \{\alpha' = \dyn \mid \forall \alpha' \in \mathit{parts}(A)\} \label{cp:osimp:untagdep}\\
  &
  \begin{array}{ll}
    \qquad\text{where} & S_1 \neq S_2
  \end{array}\nonumber\\
  \Omega\cup \{\depcon{\alpha}{S}{A_1}, \depcon{\alpha}{S}{A_2} \} & \steps \Omega\cup \{\depcon{\alpha}{S}{A_1},A_2=A_1\} \label{cp:osimp:multidep}\\
  \Omega\cup \{\tagcon{\alpha}{S} \} & \steps \Omega\cup\{\alpha=A\}\label{cp:osimp:justtag}\\
  &
  \begin{array}{ll}
    \qquad\text{where} & (\depcon{\alpha}{S'}{A'}) \not\in \Omega\\
    & (\defcon{\alpha}{T}) \not\in\Omega\\
    & (\alpha=A') \not\in\Omega \\
    & \specmatches{\alpha}{A}{S}
  \end{array}\nonumber
\end{align}

  \end{minipage}
  \caption{Simplification of constraint sets (restated from Figure \ref{cp:fig:constsimpl}).}
  \label{cp:fig:apx:osimp}
\end{figure*}

\begin{lemma}\label{cp:lem:apx:subeq}
  If $T_1=T_2$ then $\vdash T_1 <: T_2$ and $\vdash T_2 <: T_1$.
\end{lemma}
\begin{proof}
  Straightforward induction on $T_1$.
\end{proof}

\begin{lemma}\label{cp:lem:apx:soljoin}
  If $\sigma$ is a solution to $\Omega_1$ and $\sigma$ is a solution to $\Omega_2$, then $\sigma$ is a solution to $\Omega_1\cup \Omega_2$.
\end{lemma}
\begin{proof}
  Since $\sigma$ is a solution for every constraint in $\Omega_1$ and $\Omega_2$, and for all $C\in\Omega_1\cup \Omega_2$, $C\in\Omega_1$ or $C\in\Omega_2$, so $\sigma$ is a solution for every constraint in $\Omega_1\cup \Omega_2$ so it is a solution to $\Omega_1\cup \Omega_2$.\\
\end{proof}

\begin{lemma}\label{cp:lem:apx:simplsol}
  If $\Omega\steps\Omega'$ and $\sigma$ is a solution to $\Omega'$, then $\sigma$ is a solution to $\Omega$.
\end{lemma}
\begin{proof}
  By cases on $\Omega\steps\Omega'$. Many cases are immediate using Lemmas \ref{cp:lem:apx:solsplit} and \ref{cp:lem:apx:soljoin}.
  \begin{description}
  \item[Case] \ref{cp:osimp:fundynsub} $ \Omega \cup \{\flows{\funt{V_1}{V_2}}{\dyn}\}  \steps  \Omega \cup \{\flows{\dyn}{V_1},~\flows{V_2}{\dyn}\} $:\\
    Immediate from subtyping definitions.
  \item[Case]\ref{cp:osimp:funsub} 
 $\Omega \cup \{\flows{\funt{V_1}{V_2}}{\funt{V_3}{V_4}}\}  \steps  \Omega \cup \{\flows{V_3}{V_1},~\flows{V_2}{V_4}\} $:\\ 
    Immediate from subtyping definitions.
  \item[Case]\ref{cp:osimp:refdynsub} $\Omega \cup \{\flows{\reft{V}}{\dyn}\}  \steps  \Omega \cup \{V=\dyn\} $:\\
    Immediate from subtyping definitions and from Lemma \ref{cp:lem:apx:subeq}.
  \item[Case]\ref{cp:osimp:refsub} $\Omega \cup \{\flows{\reft{V_1}}{\reft{V_2}}\}  \steps  \Omega \cup \{V_1=V_2\}$:\\
    Immediate from subtyping definitions and from Lemma \ref{cp:lem:apx:subeq}.
  \item[Case]\ref{cp:osimp:idsub} 
 $\Omega \cup \{\flows{V}{V}\}  \steps  \Omega$:\\ 
 Immediate since subtyping is reflexive.
  \item[Case] \ref{cp:osimp:intdynsub}
 $\Omega \cup \{\flows{\intt}{\dyn}\}  \steps  \Omega$:\\
 Immediate from subtyping definitions.
  \item[Case] \ref{cp:osimp:funeq}
 $\Omega \cup \{V_1\to V_2=V_3\to V_4\}  \steps  \Omega\cup\{V_1=V_3,V_2=V_4\}$:\\ 
    Immediate.
  \item[Case] \ref{cp:osimp:refeq}
 $\Omega \cup \{\reft{V_1}=\reft{V_2}\}  \steps  \Omega\cup\{V_1=V_2\}$:\\ 
    Immediate.
  \item[Case] \ref{cp:osimp:tyeq}
 $\Omega \cup \{A=A\}  \steps  \Omega$:\\ 
    Immediate.
  \item[Case]  \ref{cp:osimp:flipeq} $ \Omega \cup \{A=\alpha\}  \steps  \Omega\cup \{\alpha=A\} $:\\
    Immediate.
  \item[Case]  \ref{cp:osimp:substeq} $\Omega\cup \{\alpha=A\}  \steps  \Omega[\alpha/A] \cup \{\defcon{\alpha}{A}\}$:\\
    Since $\sigma$ is a solution to $\{\defcon{\alpha}{A}\}$, have that $\sigma \alpha=\sigma A$.\\
    Therefore for any type $A'$ with $\alpha$ as a component, $\sigma A'[\alpha/A]=\sigma A'$.\\ 
    Therefore $\sigma$ is a solution to $\Omega$, and thus a solution to $\Omega\cup \{\alpha=A\}$.
  \item[Case] \ref{cp:osimp:tydep}
 $\Omega \cup \{\depcon{A}{S}{A}\}  \steps  \Omega$:\\ 
    Immediate.
  \item[Case]  \ref{cp:osimp:tagdep} $\Omega\cup \{\tagcon{\alpha}{S}, \depcon{\alpha}{S}{A} \} \steps  \Omega\cup \{\alpha = A\}$:\\
    Have that $\sigma\alpha=\sigma A$ and $\lowlty{\sigma \alpha}=S$.\\
    Therefore $\depcon{\alpha}{S}{A}$ is satisfied.\\
    Rest is immediate.
  \item[Case]  \ref{cp:osimp:untagdep} $\Omega\cup \{\tagcon{\alpha}{S_1}, \depcon{\alpha}{S_2}{A} \} \steps \Omega\cup \{\tagcon{\alpha}{S_1}\} \cup \{\alpha' = \dyn \mid \forall \alpha' \in \mathit{parts}(A)\}$:\\
    Have that $S_1 \neq S_2$.\\
    Have that $\lowlty{\sigma\alpha}=S_1$.\\
    Therefore $\lowlty{\sigma\alpha}\neq S_2$.\\
    Have that $\sigma\alpha' = \dyn$ for all $\alpha' \in \mathit{parts}(A)$.\\
    Therefore $\depcon{\alpha}{S_1}{A}$ is satisfied.\\
    Rest is immediate.
  \item[Case] \ref{cp:osimp:multidep} $\Omega\cup \{\depcon{\alpha}{S}{A_1}, \depcon{\alpha}{S}{A_2} \} \steps \Omega\cup \{\depcon{\alpha}{S}{A_1}, A_1=A_2\}$:\\
    First, suppose that $\sigma\alpha=\sigma A_1$ and $\lowlty{\sigma\alpha}=S$.\\
    Since $\sigma A_1=\sigma A_2$, $\depcon{\alpha}{S}{A_2}$ is satisfied.\\
    Now suppose that $\lowlty{\sigma\alpha}\neq S$.\\
    Then for all parts $\alpha_1$ of $A_1$, $\sigma\alpha_1=\dyn$.\\
    Since $\sigma\alpha_1=\sigma\alpha_2$, for all parts $\alpha_2$ of $A_2$, $\sigma\alpha_2=\dyn$.\\
    Therefore $\depcon{\alpha}{S}{A_2}$ is satisfied.\\
    Rest is immediate.
  \item[Case]  \ref{cp:osimp:justtag} $\Omega\cup \{\tagcon{\alpha}{S} \}  \steps \Omega\cup\{\alpha=A\}$:\\
    Since $\sigma\alpha=\sigma A$ and $\lowlty{A}=S$, $\lowlty{\sigma\alpha}=S$.
    Rest is immediate.
  \end{description}
\end{proof}

\begin{lemma}\label{cp:lem:apx:defsigma}
  If $\Omega=\{\defcon{\alpha_1}{T_1},\ldots,\defcon{\alpha_n}{T_n}\}$, and for all $i,j \le n$, $\alpha_i\not\in T_j$, then $\sigma=\alpha_1\mapsto T_1,\ldots,\alpha_n\mapsto T_n$ is a solution to $\Omega$.
\end{lemma}
\begin{proof}
  Since $\alpha_i\not\in\mathit{vars}(T_j)$ for any $i,j$, $\sigma\Omega=\{\defcon{T_1}{T_1},\ldots,\defcon{T_n}{T_n}\}$. Therefore $\sigma$ is a solution to $\Omega$.
\end{proof}

\begin{customthm}{\ref{cp:thm:solvesound}}
  If $\Omega\Downarrow\sigma$, then $\sigma$ is a solution to $\Omega$.
\end{customthm}
\begin{proof}
  By induction on $\Omega\Downarrow\sigma$.
  \begin{description}
  \item[Base case.] By Lemma \ref{cp:lem:apx:defsigma}, $\sigma$ is a solution to $\Omega$.
  \item[Simplification case.] By the IH, $\sigma$ is a solution to $\Omega''$, $\sigma$ is a solution to $\Omega'$.\\
    By repeating Lemma \ref{cp:lem:apx:simplsol}, $\sigma$ is a solution to $\Omega$.
  \item[Solving case.] By the IH, $\sigma$ is a solution to $\Omega\cup\{\tagcon{\alpha}{S}\}$.\\
    By Lemma \ref{cp:lem:apx:solsplit}, $\sigma$ is a solution to $\Omega$.
  \end{description}
\end{proof}


\end{document}